\pdfoutput=1
\documentclass[sigconf]{acmart}
\renewcommand\footnotetextcopyrightpermission[1]{} % removes footnote with conference information in first column
\makeatletter                   %1
\def\mdseries@tt{m}             %1
\makeatother                    %1
\usepackage[draft=true]{minted} %2

\hypersetup{
	colorlinks=true,
	linkcolor=blue,
	filecolor=red,      
	urlcolor=magenta,
	breaklinks=true,            %3
}
\usepackage{breakurl}           %3

\usepackage{natbib}
\usepackage{color,xspace}
\usepackage{todonotes}
\usepackage{algorithm}
\usepackage[noend]{algpseudocode}
\usepackage{caption}
\usepackage{float}

\usepackage{colortbl}
\usepackage{enumitem}

\usepackage{graphicx}

\usepackage{hyperref}

\usepackage{tikz,ifthen}

\newtheorem{thm}{theorem}[section]
\newtheorem{cor}[thm]{Corollary}
\newtheorem{prop}[thm]{Proposition}

\newtheorem{defn}[thm]{Definition}
\newtheorem{rem}[thm]{Remark}

\newtheorem{problem}[thm]{Problem}
\newtheorem{construction}[thm]{Construction}
\newtheorem{model}[thm]{Model}

\newcommand{\vx}{\mathbf{x}}
\newcommand{\vm}{\mathbf{m}}

\newcommand{\vzero}{\mathbf{0}}
\newcommand{\R}{\mathbb{R}}
\newcommand{\F}{\mathbb{F}}

\newcommand{\vA}{\mathbf{A}}

\newcommand{\cV}{\overline{\mathcal{V}}}
\newcommand{\cE}{\overline{\mathcal{E}}}

\newcommand{\Co}{\mathrm{C}}
\newcommand{\Fo}{\mathrm{W}}

\newcommand{\ip}[2]{\left\langle {#1},{#2}\right\rangle}

\newcommand{\parent}{\mathrm{parent}}
\newcommand{\children}{\mathrm{children}}

\newcommand{\mdw}{y}

\newcommand{\cK}{K}
\newcommand{\cH}{\mathcal{H}}
\newcommand{\cT}{\mathcal{T}}
\newcommand{\cD}{\mathcal{D}}
\newcommand{\cR}{\mathcal{R}}
\newcommand{\cG}{\mathcal{G}}
\newcommand{\ST}{\mathrm{ST}}
\newcommand{\SD}{\mathrm{DISJ}}

\newcommand{\cF}{\mathcal{F}}
\newcommand{\PR}{\mathrm{Pr}}
\newcommand{\eps}{\epsilon}

\newcommand{\MCF}{\mathrm{MCF}}

\newcommand{\eqdef}{\stackrel{\text{def}}{=}}

\newcommand{\tO}{O}
\newcommand{\tOm}{\tilde{\Omega}}

\usepackage{amsthm,amsmath,amssymb}
\newcommand{\vy}{\mathbf{y}}

\renewcommand{\tilde}{\widetilde}

\newcommand{\ceil}[1]{{\left\lceil#1\right\rceil}}

\newcommand{\TRB}{\mathrm{TRIBES}}
\newcommand{\BCQ}{\mathrm{BCQ}}

\newcommand{\MC}{\mathrm{MinCut}}

\newcommand{\Dom}{\mathrm{Dom}}
\newcommand{\D}{\mathbb{D}}

\renewcommand{\Join}{\bowtie}

\newcommand{\attr}{\mathrm{attr}}

\newcommand{\MPC}{\mathrm{MPC}}

\mathchardef\mhyphen="2D
\newcommand{\FAQ}{\mathsf{FAQ}}
\newcommand{\FAQS}{\mathsf{FAQ\mhyphen SS}}
\newcommand{\congest}{\mathsf{CONGEST}}

\newcommand{\hinf}[1]{H_{\infty}\left({#1}\right)}
\newcommand{\hinfe}[2]{H_{\infty}^{{#2}}\left({#1}\right)}

\newcommand{\tensprod}[1]{%
	\mathbin{\mathop{\otimes}\limits_{#1}}%
}
\newcommand{\tenssum}[1]{%
	\mathbin{\mathop{\oplus}\limits_{#1}}%
}

\newcommand{\ssupp}{{\mathrm{supp}}}
\newcommand{\calE}{\mathcal{E}}
\newcommand{\bone}{\mathbf{1}}
\newcommand{\vY}{\mathbf{Y}}

\newcommand{\tbc}{\textcolor{black}}

\tikzstyle{vertex} = [fill, shape=circle, opacity=1, node distance=20pt]
\tikzstyle{hyperedgeline} = [fill,cap=round, join=round,line width=20pt]
\tikzstyle{elabel} = [fill, shape=circle, node distance=30pt]

\pgfdeclarelayer{background}
\pgfsetlayers{background,main}
\usepackage{microtype,xparse,tcolorbox}
\newenvironment{reviewer-comment }{}{}
\tcbuselibrary{skins}
\tcolorboxenvironment{reviewer-comment }{empty,
	left = 1em, top = 1ex, bottom = 1ex,
	borderline west = {2pt} {0pt} {black!20},
}
\ExplSyntaxOn
\NewDocumentEnvironment {response} { +m O{black!20} } {
	\IfValueT {#1} {
		\begin{reviewer-comment~}
			\setlength\parindent{2em}
			\noindent
			\ttfamily #1
		\end{reviewer-comment~}
	}
	\par\noindent\ignorespaces
} { \bigskip\par }

\ExplSyntaxOff
\newenvironment{edits}{%
	\setlength{\parindent}{0pt}
	\color{black}
}{}

\newcommand{\minent}[1]{\hinf{#1}}
\newcommand{\prob}[2]{\Pr_{#1}\left[#2\right]}
\newcommand{\supp}[1]{\text{supp}\left(#1\right)}
\newcommand{\vard}[3]{{#2}\sim_{#1}{#3}}

\newcommand{\bits}{\F_2}

\newcommand{\ind}[1]{\mathbf{1}_{#1}}

\newcommand{\bad}[1]{B\left(#1\right)}
\newcommand{\cU}{\mathcal{U}}
\newcommand{\vz}{\mathbf{z}}
\newcommand{\hshan}[1]{H_{\text{Sh}}\left({#1}\right)}
\author{Michael Langberg}
\affiliation{University at Buffalo}

\author{Shi Li}
\affiliation{University at Buffalo}

\author{Sai Vikneshwar Mani Jayaraman}
\affiliation{University at Buffalo}

\author{Atri Rudra}
\affiliation{University at Buffalo}

\title{Topology Dependent Bounds For FAQs}

\begin{abstract}
In this paper, we prove topology dependent bounds on the number of rounds needed to compute Functional Aggregate Queries ($\FAQ$s) studied by Abo Khamis et al. [PODS 2016]  in a synchronous distributed network under the model considered by Chattopadhyay et al. [FOCS 2014, SODA 2017]. Unlike the recent work on computing database queries in the Massively Parallel Computation model, in the model of Chattopadhyay et al., nodes can communicate only via private point-to-point channels and we are interested in bounds that work over an {\em arbitrary} communication topology. This model, which is closer to the well-studied $\congest$ model in distributed computing and generalizes Yao's two party communication complexity model, has so far only been studied for problems that are common in the two-party communication complexity literature.

This is the first work to consider more practically motivated problems in this distributed model. For the sake of exposition, we focus on two special problems in this paper: Boolean Conjunctive Query ($\BCQ$) and computing variable/factor marginals in Probabilistic Graphical Models (PGMs). We obtain tight bounds on the number of rounds needed to compute such queries as long as the underlying hypergraph of the query is $O(1)$-degenerate and has $O(1)$-arity. In particular, the $O(1)$-degeneracy condition covers most well-studied queries that are efficiently computable in the centralized computation model like queries with constant treewidth. These tight bounds depend on a new notion of `width' (namely internal-node-width) for Generalized Hypertree Decompositions (GHDs) of acyclic hypergraphs, which minimizes the number of internal nodes in a sub-class of GHDs. To the best of our knowledge, this width has not been studied explicitly in the theoretical database literature. Finally, we consider the problem of computing the product of a vector with a chain of matrices and prove tight bounds on its round complexity (over the finite field of two elements) using a novel min-entropy based argument.
\end{abstract}

\begin{document}
\sloppy
\maketitle
\section{Introduction} \label{sec:intro}
In this paper, we prove topology dependent bounds on the number of rounds needed to compute Functional Aggregate Queries ($\FAQ$s) of~\cite{faq} in a synchronous distributed network under the model considered by Chattopadhyay et al.~\cite{topology-1,topology-3}. For ease of exposition, we consider the $\FAQS$ problem~\cite{BKOZ13,OZ15,faq} i.e., $\FAQ$ with a single semiring (called {\em Marginalize a Product Function} in~\cite{aji}), which is a special case of the general $\FAQ$ problem (defined in Section~\ref{sec:general_faq}). In $\FAQS$, we are given a multi-hypergraph $\cH=(\cV,\cE)$ where for each hyperedge $e\in\cE$ we are given an input function $f_e:\prod_{v\in e} \Dom(v)\to \D$. In addition, we are given a set of {\em free variables}\footnote{We would like to mention here that our results hold only for specific choices of free variables.} $\cF \subseteq \cV$ and our goal is to compute the function:
\begin{edits}
\begin{align*}
\phi_{\cF}(\vx) =\sum_{\vy \in \prod_{v \in \cV}\Dom(v) : \vy_{\cF} = \vx}  \quad \underset{e \in \cE}{\prod} f_e(\vy_e) \tag{1.0} \label{eq:faq-eq}
\end{align*}
\end{edits}
for every $\vx \in \prod_{v \in \cF} \Dom(v)$, where $\vy_e$ and $\vy_{\cF}$ are $\vy$ projected down to co-ordinates in $e \subseteq \cV$ for every $e \in \cE$ and $\cF \subseteq \cV$ respectively. Further, all the operations are over the \emph{commutative semiring}\footnote{\tbc{A triple $(\D,\oplus, \otimes)$ is a \emph{commutative semiring} if $\oplus$ and $\otimes$ are commutative binary operators over $\D$ satisfying the following: $(1)$ $(\D, \oplus)$ is a commutative monoid with an additive identity, denoted by $\mathbf{0}$. $(2)$ $(\D, \otimes)$ is a commutative monoid with a multiplicative identity, denoted by $\mathbf{1}$. (In the usual semiring definition, we do not need the multiplicative monoid to be commutative.) $(3)$ $\otimes$ distributes over $\oplus$. $(4)$ For any element $d \in \D$, we have $d \otimes \mathbf{0} = \mathbf{0} \otimes d = \mathbf{0}$.}} $(\D,+,\cdot)$ with additive identity $\vzero$. As with database systems, we assume that the functions are given in {\em listing} representation i.e., the function $f_e$ is represented as a list of its non-zero values: $R_e=\{(\vy,f_e(\vy))|\vy\in\prod_{v\in e} Dom(v):f_e(\vy)\neq \vzero\}$.\footnote{\tbc{We use function/relation interchangeably for $f_e/R_e$ but both mean the same.}} \tbc{We define $D = \max_{v \in \cV} |\Dom(v)|$, $N = \max_{e \in \cE}|R_e|$, $k = |\cE|$ and $r$ as the maximum arity among all functions.} 

Though our results are semiring agnostic, we mention two special problems that we consider in this paper. The first problem is when $\cF=\emptyset$ and the semiring is the {\em Boolean semiring} $(\D= \{0,1\},\vee,\wedge)$. This corresponds to the {\em Boolean Conjunctive Query} (which we will call $\BCQ$).\footnote{$\cF=\cV$ over the Boolean semiring is the natural join problem.} The other problem is when $\cF=e$ for some $e\in \cE$ and the semiring is $(\R_{\ge 0},+,\cdot)$, which corresponds to computing a {\em factor marginal} in {\em Probabilistic Graphical Models} (or PGMs) -- here we think of $f_e$ as a probability distribution. The $\FAQ$ setup (and even $\FAQS$) encompasses a large class of problems in varied domains. We refer the reader to the surveys~\cite{faq-survey,aji} for an overview of these applications.

Given a query $q=\left(\cH,\left\{f_e\right\}_{e\in\cE},\cF \right)$, we will consider the number of rounds needed to compute $q$ in a distributed environment. In particular, the underlying {\em communication topology}\footnote{Note that this is distinct from $\cH$ and is just a simple graph: see Figure~\ref{fig:example} for an example illustrating this difference.} $G=(V,E)$ is assumed to be a synchronous network and we would like to compute $q$ on $G$ with the following constraints~\cite{topology-1,topology-3}. Initially, all functions $\left\{f_e\right\}_{e\in \cE}$ are assigned to specific nodes $K \subseteq V : 1\le |K| \le k$ (called {\em players}). In each {\em round} of communication, $O(r \cdot \log_{2}(D))$ bits\footnote{\tbc{This is a natural choice since any tuple in any function can be communicated with at most $O(r \cdot \log_{2}(D))$ bits. Our bounds seamlessly generalize to the cases when each edge -- $(1)$ can transmit $B \neq r \cdot \log_{2}(D)$ bits and $(2)$ has a different capacity, but for ease of exposition, we will not consider these generalizations in this paper.}} can be simultaneously communicated on each edge in $E$ (each such edge or {\em channel} is private to the nodes at its endpoints). At the end of the protocol, a pre-determined player in $K$ knows the answer to $q$. Naturally, we would like to design protocols that minimize the total number of rounds of communication (rounds hereon) needed to compute $q$ on $G$. More generally, we would like to obtain tight bounds depending on $\cH$ and $G$ for this problem for \emph{every} query topology $\cH$ and \emph{every} network topology $G$. Note that we do not take into account the internal computation done by nodes in $G$ and we assume that all nodes in $V$ co-operatively compute the answer to $q$.

\begin{figure}
	\centering
%\begin{center}
%\hspace{-1.5cm}
\begin{tikzpicture}[scale=0.5]

%%% Star Graph
\node[circle, fill=gray!30] (A) at (-5,2) {$A$};
\node[circle,  fill=gray!30] (B) at (-8,0) {$B$};
\node[circle,  fill=gray!30] (C) at (-6,0) {$C$};
\node[circle,  fill=gray!30] (D) at (-4,0) {$D$};
\node[circle,  fill=gray!30] (E) at (-2,0) {$E$};

\draw[ultra thick, black] (A) -- (B); \node[above] at (-6.8,1) {$R$}; 
\draw[ultra thick, black] (A) -- (C); \node[above] at (-5.9,.5) {$S$};
\draw[ultra thick, black] (A) -- (D); \node[above] at (-4.1,.5) {$T$};
\draw[ultra thick, black] (A) -- (E); \node[above] at (-3.2,1) {$U$};

\node[below] at (-5,-1) {$\cH_1$};

%%% Star Hypergraph 
%\node[circle, fill=gray!30] (A) at (2,2) {$A$};
%\node[circle,  fill=gray!30] (B) at (2,0) {$B$};
%\node[circle,  fill=gray!30] (D) at (4,0) {$D$};
%\node[circle,  fill=gray!30] (C) at (4,2) {$C$};
%\node[circle, fill=gray!30] (E) at (0,2) {$E$};
%\node[circle, fill=gray!30] (F) at (6,2) {$F$};

%\draw[ultra thick, blue] (A) -- (B); \node[left] at (0,1) {$e$};
%\draw[ultra thick, green] (A) -- (C); \node[below] at (1,0) {$f$};
%\draw[rounded corners=5mm, ultra thick, black] (-1.3,2.5) -- (2.8, 2.8) -- (2.5, -1.3) --cycle; \node[left] at (1.5,1) {$U$};
%\draw[rounded corners=5mm, ultra thick, blue] (1.5,2.7) -- (5, 2.4) -- (1.5, -1.3) --cycle; \node[left] at (2.5,1) {\textcolor{blue}{$R$}};
%\draw[rounded corners=5mm, ultra thick, green] (0.9,-.5) -- (4.5, 2.9) -- (4.5, -.5) --cycle; \node[right] %at (2.8,1) {\textcolor{green}{$S$}};
%\draw[hyperedgeline, color=yellow] (4,2)--(6,2)--cycle; %\node

\node[vertex, fill=gray!30] (A) at (3, 2) {$A$};
\node[vertex, fill=gray!30] (B) at (5, 2) {$B$};
\node[vertex, fill=gray!30] (C) at (3, 0){ $C$};
\node[vertex, fill=gray!30] (D) at (6, 0) {$D$};
\node[vertex, fill=gray!30] (E) at (1, 2) (E) {$E$};
\node[vertex, fill=gray!30] (F) at (1, 0)(F) {$F$};

%\draw[hyperedgeline, color=blue] (A.center)--(B.center)--(C.center)--cycle;
%\draw[hyperedgeline, color=green, line width = 30pt] (B.center)--(C.center)--cycle;
%\draw[hyperedgeline, color=red, line width = 30pt] (C.center)--(F.center)--cycle;
%\draw[hyperedgeline, color=black, line width = 25pt] (A.center)--(B.center)--(E.center)--cycle;

\begin{pgfonlayer}{background}
\begin{scope}[transparency group,  opacity=0.5]
\draw[hyperedgeline, color=black, line width = 32pt] (A.center)--(B.center)--(C.center)--cycle;
\end{scope}

\begin{scope}[transparency group,  opacity=0.5]
\draw[hyperedgeline, color=black, line width = 29pt] (B.center)--(D.center)--cycle;
\end{scope}

\begin{scope}[transparency group,  opacity=0.5]
\draw[hyperedgeline, color=black,  line width = 30pt] (C.center)--(F.center)--cycle;
\end{scope}

\begin{scope}[transparency group, opacity=0.5]
\draw[hyperedgeline, color=black, line width = 25pt] (A.center)--(B.center)--(E.center)--cycle;
\end{scope}
\end{pgfonlayer}

%\node[elabel,color=gray!40,label=right:\(R\)]  (R) at (-0.5,2.2) {};
%\node[elabel,color=gray!70,label=right:\(S\)]  (S) at (-0.5,1.4) {};
%\node[elabel,color=gray!99,label=right:\(T\)]  (T) at (-0.5,0.6) {};
%\node[elabel,color=black,label=right:\(U\)]  (U) at (-0.5,-0.2) {};

%\draw[ultra thick, red]   (C) to[out=-20,in=-70] (F)

%\draw[rounded corners=5mm, ultra thick, red] (3, 2) -- (6.5, 2)  --cycle; %\node[right] at (,1) {\textcolor{red}{$T$}};
%\draw[ultra thick, red] (C) -- (F); \node[below] at (5,2) {\textcolor{red}{$T$}};

\node[below] at (2.8,-1) {$\cH_2$};

%%% Network

\node[circle, fill=gray!30] (R) at (-8,-3) {$R$}; \node[below] at (-8,-3.5) {$P_1$};
\node[circle, fill=gray!30] (S) at (-6,-3) {$S$}; \node[below] at (-6,-3.5) {$P_2$};
\node[circle, fill=gray!30] (T) at (-4,-3) {$T$}; \node[below] at (-4,-3.5) {$P_3$};
\node[circle, fill=gray!30] (U) at (-2,-3) {$U$}; \node[below] at (-2,-3.5) {$P_4$};
\node[below] at (-5,-4.5) {$G_1$};

\draw[thick] (R) -- (S);
\draw[thick] (T) -- (S);
\draw[thick] (T) -- (U);

\node[circle, fill=gray!30] (R) at (1,-3) {$R$}; \node[above] at (1,-2.5) {$P_1$};
\node[circle, fill=gray!30] (S) at (5,-3) {$S$}; \node[above] at (5,-2.5) {$P_2$};
\node[circle, fill=gray!30] (T) at (5,-5) {$T$}; \node[below] at (5,-5.5) {$P_3$};
\node[circle, fill=gray!30] (U) at (1,-5) {$U$}; \node[below] at (1,-5.5) {$P_4$};
\node[below] at (3,-6) {$G_2$};

\draw[thick] (R) -- (S);
\draw[thick] (R) -- (T);
\draw[thick] (R) -- (U);
\draw[thick] (S) -- (T);
\draw[thick] (S) -- (U);
\draw[thick] (T) -- (U);

\end{tikzpicture}
\caption{\tbc{Two example queries $\cH_1$ and $\cH_2$ and two topologies -- `line' $G_1$ and `clique' $G_2$. $\cH_2$ has hyperedges $R(A, B, C)$, $S(B, D)$, $T(C, F)$ and $U(A, B, E)$.}} \label{fig:example}
\end{figure}
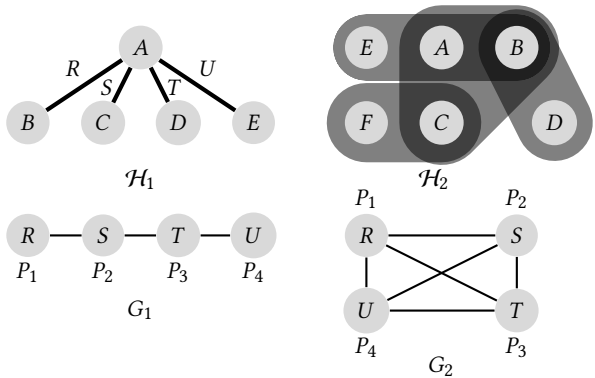

\begin{edits}
\subsection{Why this distributed model?} \label{sec:intro_model}
We believe that the strength of our model is its generality. Specifically, it captures query computation in three different paradigms, namely: $(1)$ Computing the natural join query in the Massively Parallel Computation (MPC) model~\cite{mpc,BKS13,gym,BKS14,BKS16,JR16}, $(2)$ Computing join and aggregation queries for sensor networks~\cite{sensor-db1,sensor-db2,TinyDB} and $(3)$ Computation of $\FAQ$s on arbitrary topologies using {\em software defined networks} and optical reconfigurable networks like ProjecToR~\cite{projector}. Before we discuss these in detail, we would like to mention that the $\congest$ model in distributed computing has the same setup as ours~\cite{peleg-book} with one crucial difference. Unlike our case, where we can compute $\FAQ$s on \emph{any} topology in the $\congest$ model, the topologies for computing a fixed $\FAQ$ typically depend on the query itself. 

The sequence of works in the MPC setting focus on computing the natural join $q$ (which is a special case of $\FAQS$ as mentioned earlier) on a topology $G$ with $p$ nodes, which is typically well-connected. Each round of communication has two phases -- $(1)$ internal computation among the nodes and $(2)$ communication between the nodes bounded by a node capacity $L$. The goal in MPC is to minimize the number of rounds $h$ needed for computing $q$. There are two different lines of work in this regime -- one where $p$ is fixed and the goal is to determine $h, L$~\cite{mpc,BKS13,BKS14,BKS16} and the other is when $h, L$ are fixed and the goal is to determine $p$~\cite{gym}. We compare both these classes of models with ours in Appendix~\ref{sec:comparison} and present an executive summary here.

Roughly speaking, the MPC model defined in~\cite{BKS13} is a special case of our model. We consider two different MPC models -- one with no replication (which we call $\MPC(0)$~\cite{BKS13}) and one with replication (which we dub $\MPC(\epsilon)$~\cite{BKS16, gym}). Both these models have some differences from ours and among themselves. For instance, both these models assume a specific network topology $G'$ (as opposed to any topology $G$ in our case), work on node capacities $L$ (as opposed to edge capacities in our setting) and prove bounds for the natural join problem (in contrast, our bounds apply for the more general $\FAQ$). The input functions are systematically assigned to players in $\MPC(0)$ and are uniformly distributed among players in $\MPC(\epsilon)$. The instantation of these models for the setting where $p$ is fixed and $h, L$ is to be determined is the closest to our model. In particular, when $\cH$ is a star, our protocols obtain the same guarantees as $\MPC(0)$ and are slightly worse in $\MPC(\epsilon)$. Our model does not (yet) handle the scenario when $L$ is fixed and the goal is to determine $p$. 

Sensor networks are typically tree-like topologies, where the goal is to efficiently and accurately report aggregate queries on data generated by the sensors. Since the sensors can store only little data, these queries are typically restrictive. We show in Appendix~\ref{sec:sensor} that our results imply bounds for some of these queries. Recently, Internet of Things (IoT) devices~\cite{iot1} show the promise of expanding the data storage/class of queries that can be computed on sensor networks. We believe that our model/results will find more relevance in the IoT setting since the sensors used posses more computation power than those considered in~\cite{TinyDB}. Finally, our work initiates the study of computation on general topologies to be used in emerging technologies like ProjecToR~\cite{projector}, which has been proposed for use in data centers where topologies can be changed based on the workload.
\end{edits}

\subsection{Summary of Our Contributions} \label{sec:intro_results}
Table~\ref{tab:results} lists our results and Section~\ref{sec:overview} contains a detailed overview of techniques used to obtain the results. We summarize our contributions here. For the sake of brevity, we focus on the $\BCQ$ problem. Our main result is the following. For (hyper)graphs $\cH$ with constant degeneracy\footnote{Degeneracy is defined as the smallest $d$ such that every sub(hyper)graph in $\cH$ has a vertex of degree at most $d$.} ($d$) and constant arity ($r$), we prove tight bounds (up to constant factors) for computing {\em any} $\BCQ$ on {\em any} network topology $G$. Constant treewidth implies constant $d$ and, as a result, queries having constant $d$ encompass most well-studied queries that are efficiently computable in the centralized computation model.
 
\paragraph{Upper Bounds} Our upper bounds need protocols for solving the following two basic algorithmic tasks: $(1)$ set intersection and $(2)$ sending all inputs to a single node. For $(1)$, our protocol is new in the $\FAQ$ literature and for $(2)$, we use a standard protocol from flow networks. Interestingly, our results highlight a notion of width of {\em acyclic} queries-- the number of internal nodes for a subclass of GHDs\footnote{\tbc{An internal node is a non-leaf node in a GHD.}} (defined in Section~\ref{sec:overview-lb}), which to the best of our knowledge, has not been explicitly studied in the database literature.

\paragraph{Lower Bounds} Our lower bounds follow from known lower bounds on the well-studied $\TRB$ function in two-party communication complexity literature (defined in Section~\ref{sec:overview-lb}). At a high level, we start with an arbitrary $\TRB$ instance and show that it can be reduced to a suitable $\BCQ$ instance in our model. We then prove lower bounds on the $\BCQ$ instance using known lower bounds on $\TRB$. 

We note here that the simplicity of our techniques allows us to extend our results to the general $\FAQ$ problem. Further, we would like to mention that extending our bounds to $d$-degenerate graphs with non-constant $d$ has a known bottleneck of solving $\BCQ$ of $\cH$ on $G$ when $\cH$ is a clique and $G$ is an edge. In particular, the gaps dependent on $d$ in Table~\ref{tab:results} cannot be resolved without addressing this bottleneck. 

Finally, we consider the following $\FAQS$ problem of Chain Matrix-Vector Multiplication (MCM), which is related to $k$ layer neural networks.\footnote{In neural networks, a non-linear function is applied after each matrix-vector multiplication and the multiplication is over reals instead of $\F_2$. Our lower bounds hold for this setting as well.} 
\begin{problem} [MCM] \label{prob:mcm}
Given $k$ matrices $A_i \in \F_{2}^{N \times N}$ for every $i \in [k]$ and one vector $\vx \in \F_{2}^{N}$, our goal is to compute $\vA_k\cdot \vA_{k-1}\cdot \ldots \cdot \vA_1\cdot \vx$ over $\F_2$\footnote{$\F_2$ has two elements: the additive identity $\mathbf{0}$ and multiplicative identity $\mathbf{1}$. Addition and Multiplication are all modulo $2$.} on a line $G$, where the vector $\vx$ and the matrices $(A_i)_{i \in [k]}$ are assigned in order (on $G$).
\end{problem}
We prove a tight bound for Problem~\ref{prob:mcm} in this paper. The upper bound is simple but the lower bound argument (though conceptually simple) is technically the most involved part of the paper. We use an entropy-based argument using min-entropy instead of the standard Shannon's entropy. This requires more care since we can no longer use the chain rule. Finally, we would like to note that Problem~\ref{prob:mcm} is different from the well-known Online Matrix Vector Multiplication problem. We illustrate this difference in Section~\ref{sec:oumv}.

\begin{table}[th!]
	{\small
		\centering
		{
			\hspace{0.5cm} {
				\begin{tabular}{c|c|c|c|c}
					\hline
					Query & $G$ & $d, r$ & Gap & Ref\\
					\hline
					\rowcolor{green!50}
					$\FAQ$ & L & $O(1), O(1)$ & $\tilde{O}(1)$ & Thm~\ref{cor:faq-const1} \\
					\hline
					\rowcolor{green!50}
					$\FAQ$ & A & $O(1), O(1)$ & $\tilde{O}(1)$ & Thm~\ref{cor:faq-const1}\\ 
					\hline
					\rowcolor{magenta!50}
					$\BCQ$ & A & $d, 2$ & $\tilde{O}(d)$& Thm~\ref{thm:aritytwo1}  \\
					\hline
					\rowcolor{magenta!50}
					$\FAQ$ & A & $d, r$ & $\tilde{O}(d^{2} r^{2})$& Thm~\ref{thm:faq2main}\\
					\hline
					\rowcolor{blue!50}
					MCM* & L & $1, 2$ & ${O}(1)$ & Sec~\ref{SEC:MATRIX-CHAIN}\\ 
					\hline
			\end{tabular}}
		}
		\caption{\small\tbc{The first and second columns denote the query that we compute and topology on which the query is computed. In the second column, $L$ denotes a line and $A$ denotes an arbitrary $G$. The third column denotes the degeneracy (Definition~\ref{defn:degen}) and arity conditions $(d, r)$. The fourth column denotes the gap between our upper and lower bounds ignoring polylogarithmic factors in $N$ and $G$ (denoted by $\tilde{O}$). The final column denotes the relevant result in this paper. Note that all our results except MCM (denoted by a `*') assume worst-case assignment of functions in the Query to nodes in $G$.}}
		\label{tab:results}
	}
\end{table}
\section{Our Model and Detailed Overview of Our Results} \label{sec:overview}
In this section, our goal is to provide a walk-through of our results and techniques used to prove them. We start with a formal definition of our model. Then, we illustrate with examples our results for the case when $\cH$ has arity at most two and subsequently, our new notion of width for GHDs. We conclude this section with our results on Chain Matrix-Vector Multiplication (MCM).
\subsection{Our Model}
\begin{model} \label{model:our_model}
We are given a query $q$, its underlying hypergraph $\cH = (\cV, \cE)$ with input functions $f_e$ (having at most $N$ non-zero values) for every $e \in \cE$ and a topology $G = (V, E)$. Further, each function is completely assigned to a unique node in $V$. It follows that there exists a subset $K: K \subseteq V$ that contains the players with functions and $|K| \le k = |\cE|$. We assume $N \ge |V(G)|^{2}$ and consider worst-case inputs for the functions. We would like to compute $\BCQ$ (and more generally an $\FAQ$) of $\cH$ on $G$. To design a protocol for this computation, we assume that every node in $G$ has the knowledge of $\cH$ and $G$. In each round of the protocol, at most $O(r \cdot \log_{2}(D))$ bits can be communicated over every edge in $E$. In particular, this implies any subset of edges in $G$ can communicate in the same round. Further, at the end of the protocol, a pre-determined player in $K$ has the answer to $q$. Finally, given the above setup, our goal is to design protocols that minimize the total number of rounds needed to compute $q$ assuming worst-case assignment of the functions to players in $G$. Note that we do not take into account the internal computation done by nodes in $G$ and we assume that all nodes are always available in $V$ (i.e., node failures do not happen) and they co-operatively compute the answer to $q$.
\end{model}
We prove both upper and lower bounds on the total number of rounds needed to compute $q$ on $G$ for every query hypergraph $\cH$ and every topology $G$. While our upper bounds hold for any assignment of input functions to players in $G$, our lower bounds hold for a specific class of worst-case assignments of input functions to players in $G$. In Section~\ref{sec:open}, we further discuss the assumptions on $\cH$ and $G$ in the above model.

Before we move to our results for the case when $\cH$ has arity at most two, we would like to point out that our bounds do not assume that the size of $q$ is negligible compared to $N$, which is a standard assumption for computing database queries. Thus, our results are more general and in particular, for applications in PGMs, this is necessary since the size of $q$ cannot be assumed as negligible w.r.t. $N$.
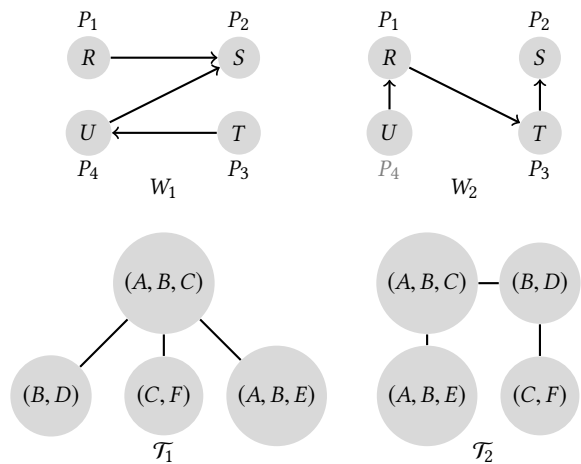
\begin{figure}[H]
	\centering
%\begin{center}
%\hspace{-1.5cm}
\begin{tikzpicture}[scale=0.5]

\node[circle, fill=gray!30] (R) at (-7,0) {$R$}; \node[above] at (-7,0.5) {$P_1$};
\node[circle, fill=gray!30] (S) at (-3,0) {$S$}; \node[above] at (-3,0.5) {$P_2$};
\node[circle, fill=gray!30] (T) at (-3,-2) {$T$}; \node[below] at (-3,-2.5) {$P_3$};
\node[circle, fill=gray!30] (U) at (-7,-2) {$U$}; \node[below] at (-7,-2.5) {$P_4$};
\node[below] at (-5,-3) {$W_1$};

\begin{scope}[every path/.style={->}]
\draw[thick] (R) -- (S);
\draw[thick] (U) -- (S);
\draw[thick] (T) -- (U);
\end{scope}  

\node[circle, fill=gray!30] (R) at (1,0) {$R$}; \node[above] at (1,0.5) {$P_1$};
\node[circle, fill=gray!30] (S) at (5,0) {$S$}; \node[above] at (5,0.5) {$P_2$};
\node[circle, fill=gray!30] (T) at (5,-2) {$T$}; \node[below] at (5,-2.5) {$P_3$};
\node[circle, fill=gray!30] (U) at (1,-2) {$U$}; \node[below] at (1,-2.5) {\textcolor{gray}{$P_4$}};
\node[below] at (3,-3) {$W_2$};

\begin{scope}[every path/.style={->}]
\draw[thick] (U) -- (R);
\draw[thick] (R) -- (T);
\draw[thick] (T) -- (S);
\end{scope}

\node[circle, fill=gray!30] (R) at (-5,-6) {$(A, B, C)$};
\node[circle, fill=gray!30] (S) at (-8,-9) {$(B, D)$};
\node[circle, fill=gray!30] (T) at (-5,-9) {$(C, F)$};
\node[circle, fill=gray!30] (U) at (-2,-9) {$(A, B, E)$}; 
\node[below] at (-5,-10) {$\cT_1$};

\draw[thick] (R) -- (S);
\draw[thick] (R) -- (T);
\draw[thick] (R) -- (U);

\node[circle, fill=gray!30] (R) at (2,-6) {$(A, B, C)$};
\node[circle, fill=gray!30] (S) at (5,-6) {$(B, D)$};
\node[circle, fill=gray!30] (T) at (5,-9) {$(C, F)$};
\node[circle, fill=gray!30] (U) at (2,-9) {$(A, B, E)$};
\node[below] at (3.5,-10) {$\cT_2$};

\draw[thick] (R) -- (S);
\draw[thick] (S) -- (T);
\draw[thick] (R) -- (U);

\end{tikzpicture}
\caption{\tbc{Two directed paths $W_1$ and $W_2$ for $G_2$ and two GHDs $\cT_1$ (with $1$ internal node) and $\cT_2$ (with $2$ internal nodes) both rooted at $(A, B, C)$ for $\cH_2$. $G_2$ and $\cH_2$ are in Figure~\ref{fig:example}.}} \label{fig:example2}
\end{figure}

\subsection{Arity Two $\cH$} \label{sec:overview_atwo}
We consider the case when $\cH$ has arity at most two and illustrate our upper and lower bound techniques through examples.
\subsubsection{Upper Bounds} \label{sec:overview_ub}
We start with a trivial protocol to compute any query $\cH$ on any $G$. We then show how to improve upon it when $\cH$ has a special structure. We use two extremal instances of $G$ for an easy exposition of our results -- a line (least connectivity) and a clique (full connectivity). We refer the reader to Figure~\ref{fig:example} for all examples (except $\cH_{0}$) considered in this section. 

\paragraph{Trivial Protocol.} There is always a \emph{trivial protocol} to solve any query $\cH$ on any $G$ in which all players send their functions to one designated player who then computes the answer.

We consider the topologies $G_1$ and $G_2$ from Figure~\ref{fig:example}. We first start by computing a toy query $\cH_{0}$ on $G_1$. 
\begin{example} \label{ex:self-loop}
Consider the query hypergraph $\cH_{0} = (\cV = \{A\}, \cE =  \{R(A), S(A), T(A), U(A)\})$ i.e., all edges are self-loops on $A$ and the line $G_1$.  We would like to solve $\BCQ$ of $\cH_0$ on $G_1$, which in Datalog format is $q_{0}() : - R(A), S(A), T(A), U(A)$.

In $G_1$, player $P_1$ gets $R$, $P_2$ gets $S$, $P_3$ gets $T$ and $P_4$ gets $U$. 
Then, solving $\BCQ$ of $\cH_{0}$ on $G_1$ is equivalent to checking if the set-intersection $R(A) \cap S(A) \cap T(A) \cap U(A)$ is empty. Let us assume that player $P_4$ needs to know the answer for this query.

We can solve this query in $N + 2$ rounds as follows. In the first round, player $P_1$ sends a value $a \in \Dom(A)$ such that there exists $R(a) = 1$ to player $P_2$ who then checks if $S(a) = 1$. More generally, in the $i$-th round, player $P_j$ for $2 \le j \le 4$ receives an $a$ from its left neighbor $(j - 1)$ and checks if $a$ is present in its table. If so, it passes $a$ to its right neighbor $(j + 1)$ (if $j \le 3$) in the next $(i + 1)$-th round. Otherwise, it does not pass anything. Notice that this protocol will terminate once all matching values of $a$ are passed from $P_1$ to $P_4$ which takes $N + 2$ rounds in the worst case. In other words, we are computing the semijoin (see Definition~\ref{defn:semijoin}) query $((R(A) \ltimes S(A)) \ltimes T(A)) \ltimes U(A)$, which is equivalent to computing $R(A) \cap S(A) \cap T(A) \cap U(A)$. Note that this is much better than the \emph{trivial protocol} for this case, which takes $3 \cdot N + 2$ rounds.

At the end of this protocol, $P_4$ knows the answer to the query. It is not too hard to see that we can extend the above protocol to the case when any other player say $P_i$ for some $i \in [3]$ is designated to know the answer. In particular, we can orient $G_1$ in such a way that all paths are directed towards $P_i$ and then run the protocol above simultaneously on all paths (there are at most two) towards $P_i$ (recall that we assume knowledge of $G$ for all nodes). Note that $P_i$ would have the answer to the query and the new protocol takes $N + x$ rounds, where $x \le 2$ depends on the choice of $P_i$.  
\end{example}
It is not too hard to see that our protocol in the above example can be extended to the case when $\cH$ is a star. We illustrate this in the following example. 
\begin{example} \label{ex:r=2}
Consider the star $\cH_1$ and the line $G_1$ in Figure~\ref{fig:example}. We would like to solve $\BCQ$ of $\cH_1$ on $G_1$, which in Datalog format is $q_{1}() : - R(A, B), S(A, C), T(A, D), U(A, E)$. In $G_1$, player $P_1$ gets $R$, $P_2$ gets $S$, $P_3$ gets $T$ and $P_4$ gets $U$. Then, $\BCQ$ of $\cH_{1}$  is $1$ iff $\pi_{A}(R) \cap \pi_{A}(S) \cap \pi_{A}(T) \cap \pi_{A}(U)$ is non-empty and $0$ otherwise. Here, $\pi_{A}(\cdot)$ denotes the projection onto attribute $A$. We assume $P_2$ needs to know the answer for this query.

We can solve this query in $N + 2$ rounds using the same protocol as in Example~\ref{ex:self-loop}. In other words, we are computing the semjoin query\footnote{\tbc{We would like to mention that casting the computation of $\BCQ$ on a star query as a semijoin is well-known~\cite{semijoin}.}} $((\pi_{A}(R) \ltimes \pi_{A}(S)) \ltimes \pi_{A}(T)) \ltimes \pi_{A}(U)$. Note that each node needs to compute $\pi_{A}(\cdot)$ internally but this doesn't need any communication between the nodes. At the end of this protocol, $P_2$ knows the answer to the query. 
\end{example}
We now show how to do the same computation (i.e., $\BCQ$ of $\cH_1$) on $G_2$.
\begin{example} \label{ex:star-algo2}
Consider the star $\cH_1$ and the clique $G_2$ in Figure~\ref{fig:example}. We would like to compute $\BCQ$ of $\cH_1$ on $G_2$, which in Datalog format is same as $q_1$ from Example~\ref{ex:r=2}. In $G_2$, player $P_1$ gets $R$, $P_2$ gets $S$, $P_3$ gets $T$ and $P_4$ gets $U$. We assume that $\Dom(A)$ is split into two halves and $P_2$ needs to know the answer for this query.
	
We can solve this query in $\frac{N}{2} + 2$ rounds as follows. We consider the two edge-disjoint directed paths $W_1$ and $W_2$ (see Figure~\ref{fig:example2}) on $G_2$ that end with $P_2$. Our protocol from Example~\ref{ex:r=2} runs on both these paths simultaneously with one caveat -- the values of $a$ in the first half of $\Dom(A)$ are sent through $W_1$ and the ones in the second half of $\Dom(A)$ are sent through $W_2$. Since both these directed paths involve the same set of nodes, our protocol is valid and takes only $\frac{N}{2} + 2$ rounds as claimed above. Note that this is better than our bound in Example~\ref{ex:r=2}. 
\end{example}
The protocols in Examples~\ref{ex:r=2} and~\ref{ex:star-algo2} can be generalized to solve any star $\cH$ on any $G$. Given the protocol for a star, there is a natural extension to $\cH$ being a tree (or more generally a forest): we handle all the stars of the tree in a bottom-up fashion (starting with the stars at the `end' of the tree) and recurse. In particular, we can apply our protocol for the star case as a black-box on each of these stars. To extend this result to general $d$-degenerate graphs $\cH$, we first decompose $\cH$ into a forest and a \emph{core} that contains the roots of all trees in the forest and all remaining vertices not in the forest. We run the above protocol on the forest and use the \emph{trivial protocol} on the core. For general $G$, note that we need to find optimal ways of applying these protocols -- for the forest part, we extend the idea of packing edge-disjoint paths from Example~\ref{ex:star-algo2} to a Steiner tree (Definition~\ref{defn:steiner-tree}) packing and for the \emph{trivial protocol}, we use standard ideas from network flows (Definition~\ref{def:mcf}). We would like to mention here that our upper bounds hold even when more than one function is assigned to a player (i.e., $|K| < k$). We will crucially exploit this fact in our lower bounds. We present more details in Section~\ref{sec:aritytwo-ub}. 

We are now ready to talk about our lower bounds.
\subsubsection{Lower Bounds}  \label{sec:overview-lb}
All our lower bounds follow from known lower bounds on the well-studied $\TRB$ function (see \cite{topology-3} and references therein) in two-party communication complexity literature. To this end, we first consider an arbitrary $\TRB$ instance of a specific size and show that it can be reduced to a suitable two-party $\BCQ$ instance. In particular, solving the two-party $\BCQ$ instance we constructed indeed solves the $\TRB$ instance we started with. Thus, known lower bounds on $\TRB$ imply lower bounds for $\BCQ$. Finally, we generalize our results from the two-party setting to general $G$ using ideas from~\cite{topology-1,topology-3} and exploit the fact that our (upper and) lower bounds are for worst-case input functions and worst-case assignments of input functions to players in $G$.

We start by defining the two-party communication complexity model as a special case of Model~\ref{model:our_model}.
\begin{model} \label{model:two-party}
Consider two players Alice ($a$) and Bob ($b$) on a graph $\cG = (V = \{a, b\}, E = \{(a, b)\})$ with strings $\bar{X} = (X_1, \dots, X_m)$ and $\bar{Y} = (Y_1, \dots, Y_m)$, where $X_i, Y_i \in \{0, 1\}^{N}$. Further, Alice gets $\bar{X}$, Bob gets $\bar{Y}$ and both have knowledge of only their inputs. The goal for these two players is to compute the boolean function $f(\bar{X}, \bar{Y}) : \{0, 1\}^{m \cdot N} \times \{0, 1\}^{m \cdot N} \rightarrow \{0, 1\}$. The randomized two-party communication complexity of computing $f$, denoted by \\ $\cR(f(\bar{X}, \bar{Y}), \cG, \{a, b\})$, is defined as the minimum worst-case number of rounds\footnote{In each round, we assume at most one bit is sent from $a$ to $b$ instead of $O(\log_{2}(r \cdot D))$ bits to be consistent with the two-party communication complexity literature.} needed by a randomized protocol that deterministically computes $f(\bar{X}, \bar{Y})$ with error at most $\frac{1}{3}$.
\end{model}
We would like to mention that considering the randomized two-party communication complexity over its deterministic counterpart makes our lower bounds only stronger. We define $\TRB$ and state the lower bound result that we will use in our arguments.
\begin{thm} [Jayram et. al~\cite{jks-03}] \label{thm:jks}
Let $\TRB_{m, N}(\bar{X}, \bar{Y}) \equiv \bigwedge_{i = 1}^{m} \SD_{N} (X_i, Y_i)$, where $\SD_{N} (X_i, Y_i)$ is $1$ if $X_i \cap Y_i \neq \emptyset$ and $0$ otherwise, $X_i, Y_i \in \{0, 1\}^{N}$ for every $i \in [m]$ and $\bar{X} = (X_1, \dots, X_m)$, $\bar{Y} = (Y_1, \dots, Y_m)$. Note that in the two-party model, Alice gets $\bar{X}$ and Bob gets $\bar{Y}$. Given this setup, we have
\begin{align*}
\cR\left(\TRB_{m, N}(\bar{X}, \bar{Y}), \cG, \{a, b\} \right) \ge \Omega(m \cdot N).
\end{align*}
\end{thm} 
We start with an arbitrary $\TRB$ instance $\TRB_{m, N} (\bar{X}, \bar{Y})$ of a suitable size and show that it can be reduced to a suitable two-party $\BCQ$ instance $\BCQ_{\cH, \bar{X}, \bar{Y}}$, where $m$ is a function of $\cH$. In particular, such a reduction would imply
\begin{align*}
\cR \left(\BCQ_{\cH, \bar{X}, \bar{Y}}, \cG, \{a, b\} \right) & \ge \cR\left(\TRB_{m, N} (\bar{X}, \bar{Y}), \cG, \{a, b\} \right)  \\ 
& \ge \Omega(m \cdot N),
\end{align*}
where the final inequality follows from Theorem~\ref{thm:jks}. The above inequality implies the following since we consider worst-case input functions for a fixed $\cH$.
%\begin{cor} \label{cor:trb-lb}
\begin{equation} \label{eq:trb-lb}
\cR \left(\BCQ_{\cH, N}, \cG, \{a, b\} \right) \ge \cR \left(\BCQ_{\cH, \bar{X}, \bar{Y}}, \cG, \{a, b\} \right),
\end{equation}
where $\BCQ_{\cH, N}$ denote the class of problems where all functions in $\cH$ have size at most $N$.
%\end{cor}
We generalize the above result to any $G$ using ideas from~\cite{topology-1,topology-3}. 
We consider an appropriate cut $C = (A, B)$ of $G$ that partitions $V$ into two vertex-disjoint subsets $A$ and $B$ and a corresponding assignment, where each function $e \in \cE(\cH)$ is assigned to a node in either $A$ or $B$. 
Since this is a valid assignment of functions in $\cH$ to players in $G$, the minimum number of rounds needed to compute an instance of $\BCQ_{\cH, N}$ on $G$ assuming worst-assignments of functions to players in $K$, denoted by $\cR \left(\BCQ_{\cH}, G, K \right)$, is at least $\frac{\cR\left(\BCQ_{\cH, \bar{X}, \bar{Y}}, \cG, \{a, b\} \right)}{\MC(G, K) \ceil{\log(\MC(G, K))}}$. We reconsider $\cH_{1}$ and $G_1$ from Example~\ref{ex:r=2} here. 
\begin{example} \label{ex:lb-1}
Recall we proved an upper bound of $N + 2$ for computing $\BCQ$ of $\cH_1$ on $G_1$. We start with an arbitrary $\TRB_{m = 1, N}(\bar{X} = (X_1), \bar{Y} = (Y_1))$ instance. With a slight abuse of notation, we treat $X_1, Y_1$ as subsets of $[N]$ (instead of elements in $\{0, 1\}^{N}$). We now construct a corresponding $\BCQ_{\cH_{1}, \bar{X}, \bar{Y}}$ instance from the $\TRB$ one as follows -- we assign $R(A, B) = X_1 \times \{1\}$, $S(A, C) = T(A, D) = [N] \times \{1\}$ and $U(A, E) = Y_1 \times \{1\}$. It is not too hard to see that $\BCQ_{\cH_{1}, \bar{X}, \bar{Y}}$ is $1$ iff $\TRB_{1, N}(\bar{X}, \bar{Y})$ is $1$, implying that solving the $\BCQ$ instance would solve the $\TRB$ instance. Finally, to obtain a lower bound for computing $\BCQ_{\cH_{1}, \bar{X}, \bar{Y}}$ on the line $G_1$, we only need a cut where $R$ and $U$ are on different sides. We consider the cut $C = (\{P_1, P_2\}, \{P_3, P_4\})$ of $G_1$ and the assignment where $P_1$ gets $R$, $P_2$ gets $S$, $P_3$ gets $T$ and $P_4$ gets $U$. Then, we can us~\eqref{eq:trb-lb} and Theorem~\ref{thm:jks} to obtain the required lower bound of $\Omega(N)$ since $\MC(G, K) = 1$. Note that the above lower bound holds for any star $\cH$. The same $\TRB$ instance can be used for Examples~\ref{ex:self-loop} and~\ref{ex:star-algo2} as well. While a similar assignment holds for Example~\ref{ex:self-loop}, Example~\ref{ex:star-algo2} requires a different assignment where $C = \{P_1\}, \{P_2, P_3, P_4\})$ and $P_1$ gets $R$ and $S$, $P_2$ gets $T$ and $P_3$ gets $U$. Note that more than one input function can be assigned to the same player in $G$.
\end{example}
For general $d$-degenerate graphs $\cH$, we start by recalling that $m$ (i.e., size of the $\TRB$ instance) is a function of $\cH$. As mentioned in Section~\ref{sec:overview_ub}, we can decompose $\cH$ into a forest and a \emph{core}. We prove three different lower bounds on $\cH$, where the size of the $\TRB$ instance $m$ used in our reduction is the maximum of three different bounds, each one on a different part of $\cH$. The first one is on $\cH$'s forest part, the second and third ones are on $\cH$'s \emph{core} part -- lower bounded by applying Moore's bound~\cite{moore} and Turan's theorem~\cite{prob-method} respectively. For each case, we show that we can reduce the $\TRB$ instance to a suitable two-party $\BCQ$ instance. Thus, known lower bounds on the $\TRB$ instance from Theorem~\ref{thm:jks} apply for the $\BCQ$ instance. Finally, to generalize our results from two-party $\BCQ$ to general $G$, we use ideas from~\cite{topology-1,topology-3} to obtain an appropriate cut for $G$ and use lower bounds from the induced two-party communication complexity problem across the cut. Note that the assignment of functions depends on the cut. We present the details in Section~\ref{sec:aritytwo-lb}.

For constant $d$, our upper and lower bounds match. However, for non-constant $d$, we have a gap of $\tilde{O}(d)$. We would like to note that there is a fundamental bottleneck in getting rid of this factor as the case of $\cH$ being a clique is an outstanding open question (even in Model~\ref{model:two-party}) and seems beyond the reach of current communication complexity techniques~\cite{arkadev}. We state this problem formally in Appendix~\ref{sec:open-clique}.

\subsection{Notion of Width} \label{sec:overview-inw} 
We start by defining the notion of $GHD$s and acyclic (hyper)graphs.
\begin{defn} [GHD] \label{Definition:ghd}
A GHD of $\mathcal{H}=(\mathcal{V},\mathcal{E})$ is defined by a triple $ \langle \cT,\chi,\lambda \rangle$, where $\cT=(V(\cT), E(\cT))$ is a tree, $\chi: V(\cT) \rightarrow 2^{\mathcal{V}}$ is a function associating a set of vertices $\chi(v) \subseteq \mathcal{V}$ to each node $v$ of $\cT$, and $\lambda: V(\cT) \rightarrow 2^{\mathcal{E}}$ is a function associating a set of hyperedges to each node $v$ of $\cT$ such that the following two properties hold. First, for each $e \in \mathcal{E}$, there is at least one node $v \in V(\cT)$ such that $e \subseteq \chi(v)$ and $e \in \lambda(v)$. Second, for every $V' \subseteq \mathcal{V}$, the set $\{v \in V(\cT) | V' \subseteq \chi(v)\}$ is connected in $\cT$, called the running intersection property (RIP hereon). We only consider rooted GHDs.

A {\em reduced-GHD} has the additional property that every hyperedge $e\in\cE$ has a unique node $v \in V(\cT)$ such that $\chi(v) = e$ (note that this is an equality).
\end{defn}
\begin{defn} [Acyclicity] \label{defn:acyclic}
A hypergraph $\cH = (\cV, \cE)$ is acyclic iff there exists a GHD $(\cT, \chi, \lambda)$ in which for every node $v \in V(\cT)$, $\chi(v)$ is a hyperedge in $\cE$.
\end{defn}
We now define the sub-classes of reduced-GHDs that we consider in this paper. In particular, we construct reduced-GHDs using the GYO algorithm~\cite{gyo-1,gyo-2,gyo-3} (GYOA, also called GYO-Elimination order) and call them GYO-GHDs. We start by defining the \emph{GYO-reduction} $\cH'$ of a hypergraph $\cH$. 
\begin{defn} [GYO-reduction and GYOA] \label{defn:gyo-r}
For any hypergraph $\cH$, the GYO-reduction $\cH'$ is defined as the leftover hypergraph after running GYOA on $\cH$. We describe GYOA here. The input to GYOA is $\cH$ and its output is a hypergraph $\cH'$. It performs the following two steps iteratively on $\cH'$ (starting with $\cH' = \cH$): $(a)$ Eliminate a vertex that is present in only one hyperedge and $(b)$ Delete a hyperedge that is contained in another hyperedge. GYOA terminates when it cannot perform any of the above two steps on $\cH'$. 
\end{defn}
We note here that running GYOA on an acyclic hypergraph $\cH$ returns an empty $\cH'$. For general hypergraphs $\cH$, we start by running GYOA on $\cH$, which returns a hypergraph $\cH'$. Note that the hyperedges removed in this process form a forest of acyclic hypergraphs. For each hypertree in the forest of acyclic hypergraphs, we can construct a reduced-GHD and root it arbitrarily. We now define $\Co(\cH)$ and $\Fo(\cH)$ based on the GYO-reduction $\cH'$ and the reduced-GHDs constructed for the original hyperedges of $\cH$ removed during GYOA. 
\begin{defn} [$\Co(\cH), \Fo(\cH)$] \label{defn:forest-core}
$\Co(\cH)$ is the union of $\cH'$ and the root in each reduced-GHD we constructed. $\Fo(\cH) = \cH \setminus \cH'$.
\end{defn}
We are now ready to construct GYO-GHDs.
\begin{construction} \label{lemma:forest-core} 
Let $\cT$ be the GYO-GHD be obtained from the following procedure. We define the root $r'$ of $\cT$ with
$\chi(r') = V(\Co(\cH))$. For each edge $e \in \cE$ with $e \subset V(\Co(\cH))$, we create a new node $v'_e$ in $\cT$ with $\chi(v'_e) = e$ and add the edge $(r', v'_e)$ to $\cT$ in order to make it a reduced-GHD. For every reduced-GHD $\cT'$ in $\Fo(\cH)$, we add the edge $(r', r'')$ to $\cT$, where $r''$ is the root of $\cT'$. We add all the remaining nodes and edges in each reduced-GHD in $\Fo(\cH)$ to $\cT$. 
\end{construction}
We argue that the above procedure produces a reduced-GHD in Appendix~\ref{app:gyo-correctness}. Our new notion of width based on GYO-GHDs, which we call $y$ (Internal Node Width), is defined as follows.
\begin{defn} \label{defn:mdw}{\tiny }
\begin{align*}
y(\cH) = \min_{\forall \cT : \cT\text{ is a GYO-GHD of } \cH} y(\cT).
\end{align*}
Here, $y(\cT)$ is the number of internal/non-leaf nodes in $\cT$.
\end{defn}	
\textbf{Unless specified otherwise, in the rest of the paper when we refer to GHDs, we are referring to GYO-GHDs.} As an example in Figure~\ref{fig:example2}, we consider two different GHDs $\cT_1$ and $\cT_2$ for the acyclic hypergraph $\cH_2$ from Figure~\ref{fig:example}. Both are outcomes of Construction~\ref{lemma:forest-core} and while $\cT_2$ has two internal nodes, $\cT_1$ has only one, implying $\mdw(\cH) = 1$. For $\cH_1$ in Figure~\ref{fig:example}, it is easy to construct a GHD with one internal node (i.e., $\mdw(\cH) = 1$) by keeping $(A, B)$ as the root and $(A, C), (A, D), (A, E)$ as leaves.  We show how this can be achieved for simple graphs $\cH$ in Section~\ref{sec:r=2}.

\subsection{Chain Matrix-Vector Multiplication}
Finally, in this work, we consider the problem of computing $\vA_k\cdots\vA_1\vx$ where the computation is over $\F_2$. The player $P_i$ gets $\vA_i$ for $i\in [k]$ and $P_0$ gets $\vx$. Player $P_{k+1}$ wants to know the answer (and does not have any input). The topology $G$ is a line with $P_i$ connected to $P_{i+1}$ for $0\le i\le k$. We show that when $k \leq N$ the natural algorithm that computes the partial product $\vA_i\cdots\vA_1\vx$ at $P_i$ taking $\Theta(k N)$ rounds is indeed optimal. By contrast, if the matrices are assigned randomly to the players then the optimal number of rounds is $\Theta(k^2N)$ (this follows from a trivial protocol). On the other extreme, if all matrices are assigned to one player, then the problem is trivial. So we are proving a tight lower-bound for arguably the simplest assignment of matrices to players that is not trivial.

We note that the existing technique of~\cite{topology-3} cannot prove a lower bound better than $\Omega(N)$ for this problem (we formally show this in Appendix~\ref{app:prev-work-mcm}). To get a better lower bound of $\Omega(kN)$, we use an entropy based inductive argument to show that at end of the $\Omega(iN)$ rounds, in player $P_i$'s view, $\vA_{i-1}\cdots\vA_1\vx$ has very high entropy. However, Shannon's entropy is too weak for this argument to go through and we use the stronger notion of min-entropy, which is omnipresent in pseudorandomness and cryptography~\cite{salil-book}. Unfortunately, this means that we can no longer appeal to the chain rule and the arguments become a bit more delicate. Finally, in the process we prove the following natural result: if $\vA$ and $\vx$ have high enough min-entropy, then $\vA\vx$ has higher min-entropy than $\vx$. To the best of our knowledge this result is new, though it follows by combining known results in pseudorandomness.\footnote{We thank David Zuckerman for showing us the high level proof idea of this result.}

\section{Preliminaries and Notation} \label{sec:notation}
\begin{edits}
In this section, we define some notions related to the query hypergraph $\cH$ and network topology $G$. We conclude the section with some asymptotic notation.
\paragraph{Query (Hyper)graph $\cH$} \begin{defn} [$n_2(\cH)$] \label{defn:garity-notation}
Using Definition~\ref{defn:forest-core}, we can decompose any $\cH$ into a core $\Co(\cH)$ and a forest $\Fo(\cH)$. We define $n_2(\cH) = |V(\Co(\cH))|$.
\end{defn}
\begin{defn} [Degree] \label{defn:degree}
The degree of a vertex $v \in \cH$ is given by $\left|\{e \ni v: e \in \cE \}\right|$. 
\end{defn}
\begin{defn} [$d$-degenerate (hyper)graph~\cite{hypergraph-degeneracy}] \label{defn:degen}
In a $d$-degenerate (hyper)graph, every sub(hyper)graph has a vertex of degree at most $d$. 
\end{defn}
\end{edits}
We now define \emph{natural join} and \emph{semijoin}.
\begin{defn} [Natural Join] \label{defn:join}
The \emph{natural join} $J = ~ \underset{e \in \cE}{\Join} R_e$ is a relation $J$ with attribute set $V(\cH)$ satisfying the following condition (where $\Join$ denotes the \emph{join} operator). A tuple $\mathbf{t} \in J$ iff for every $e \in E(\cH)$, the \emph{projection} of $\mathbf{t}$ onto attributes in $v(e)$ - denoted by $\pi_{v(e)} (\mathbf{t})$ - belongs to $R_e$. Note that $J \subseteq \prod_{v \in V(\cH)} \Dom(v)$. 
\end{defn}
\begin{edits}
\begin{defn} [Semijoin] \label{defn:semijoin}
A semijoin $J' = R_1 \ltimes R_2$ of relations $R_1$ and $R_2$ is defined as $J' = R_1 \Join \pi_{\attr(R_1) \cap \attr(R_2)} (R_2)$, where $\attr(\cdot)$ denotes the attribute set of the relations and $\ltimes$ is the semijoin operator.
\end{defn}
We show in Appendix~\ref{sec:faq_prelims} that \emph{natural join} and \emph{semijoin} are special cases of $\FAQ$.
\end{edits}
\paragraph{Network Topology $G$} We define some standard graph notions that will be used throughout the paper. 
\begin{defn} [$\MC(G, K)$]
We denote the size of the minimum cut of $G$ separating vertices in $K$ by $\MC(G, K)$.
\end{defn}
\begin{defn} [Star Graph]
A star is a tree on $n$ vertices with one internal node and $n - 1$ leaves (e.g. $\cH_1$ in Figure~\ref{fig:example}). 
\end{defn}
\begin{defn} [Steiner tree] \label{defn:steiner-tree}
Given a graph $G = (V, E)$ and a set of nodes $\cK \subseteq V$, we call a tree $T$ a Steiner tree if it connects all vertices in $\cK$ only using edges in $E$.
\end{defn}
In particular, we are interested in Steiner trees with diameter at most $\Delta$ (i.e., distance between any two nodes in $K$). Let $\cT_{\Delta,\cK}$ denote the set of all such Steiner trees.
\begin{defn} [$\ST(G, K, \Delta)$]
$\ST(G,\cK,\Delta)$ denotes the maximum number of edge disjoint Steiner trees from $\cT_{\Delta,\cK}$ that can be packed in $G$.
\end{defn}
We will need this result:
\begin{thm}[\cite{lau}] \label{thm:lau}
\[\ST(G, K,|V(G)|)=\Omega(\MC(G, K)).\]
\end{thm}
Finally, we state a recent result under Model~\ref{model:our_model} on set-intersection queries over any topology $G$ and any subset of players $K \subseteq V: |K| \le k$, which we will use frequently in our arguments.
\begin{edits}
\begin{thm} [\cite{topology-3}] \label{thm:set-intersection}
Let $\mathbf{x}_{u} \in \{0, 1\}^{N}$ for every player $u \in K$. The number of rounds taken by a protocol that deterministically computes $\bigwedge_{u \in K} \mathbf{x}_{u}$ (where the $\wedge$ is bit-wise AND) is given by $\Theta \left( \min_{\Delta \in [|V|]} \left(\frac{N}{\ST(G, K, \Delta) } + \Delta \right) \right)$.
\end{thm}
We will use the following notation for a special case of a multi-commodity flow problem:
\begin{defn} \label{def:mcf}
For every graph $G$, subset of players $K$ and integer $N'\ge 0$, let $\tau_{\MCF}(G, \cK, N')$ be the minimum number of rounds needed to route at most $N' \log_{2}(N')$ bits from all players in $K$ to any one player in $K$, assuming $\log_{2}(N')$ bits are sent in each round. \footnote{Here, we will consider the worst-case over all possible ways the $N' \log_{2}(N')$ bits are distributed over $K$. While our upper bounds can be smaller than this, we use this worst-case measure to simplify our bounds.}
\end{defn}
\end{edits}
Let the minimum number of rounds taken by a protocol to deterministically compute $\BCQ$ of $\cH$ on $G$ be denoted by $\cD(\BCQ_{\cH, N}, G, K)$, where each function in $\cH$ has size at most $N$ and is assigned to some player in $K \subseteq V, |K| \le k$. Recall that $\cR(\BCQ_{\cH, N}, G, K)$ is the most minimum worst-case number of rounds needed to deterministically compute any instance in $\BCQ_{\cH, N}$ with error at most $\frac{1}{3}$. The \emph{trivial protocol}, along with Definition~\ref{def:mcf}, implies the following.
\begin{lemma} \label{lem:naive}	
\[\cD(\BCQ_{\cH, N}, G, K) = O\left(\tau_{\MCF}(G, \cK, k\cdot r \cdot N)\right).\]
\end{lemma}
\begin{edits}
\subsection{Asymptotic Notation}\label{sec:polylog} 
For notational clarity, we will ignore a factor of $\log_{2}(N) \cdot \log_{2}(\MC(G, K))\cdot \log_{2}(n_2(\cH))$ in our lower bounds. Further, we ignore these factors while arguing for the tightness of our bounds, which we denote by $\tOm(\cdot), \tilde{O}(\cdot)$ and $\tilde{\Theta}(\cdot)$.
\end{edits}
\section{$\cH$ is a Simple Graph} \label{sec:r=2}
In this section, we consider the class of queries $\BCQ_{\cH,N}$ for a given $d$-degenerate graph $\cH$ with arity $r$ at most two and all functions have size at most $N$. We prove upper and lower bounds that are tight within a factor of $\tilde{O}(d)$ for computing any query in $\BCQ_{\cH, N}$. The following is our main result.
\begin{thm} \label{thm:aritytwo1}
For arbitrary topology $G$, subset of players $K$ and $d$-degenerate simple graph $\cH$, we have
\begin{align*}
& \mathcal{D} (\BCQ_{\cH, N}, G, K) = \\ & \tO\left(y(\cH) \cdot \min_{\Delta \in [|V|]}\left(\frac{N}{\ST(G, K, \Delta)} + \Delta \right) \right)  \\ & + \tO\left(\tau_{\MCF}(G, K, n_2(\cH) \cdot d \cdot N)\right). \tag{1.1} \label{eq:aritytwoub}
\end{align*}
Further, for all simple graphs $\mathcal{H}$, we have
\begin{align*}
& \cR(\BCQ_{\cH, N}, G, K) \ge \\ 
& \tOm\left( \frac{y(\cH) \cdot N}{\MC(G, K)} \right) + \tOm \left(\frac{n_2(\cH) \cdot N}{\MC(G, K)} \right). \tag{1.2}\label{eq:aritytwolb}
\end{align*}
\end{thm}
We would like to point out that our upper bound holds for every assignment of the functions $f_e$ to players in $K$ while our lower bound holds for some assignment of functions to players in $K$. 
We first prove the upper bound~\eqref{eq:aritytwoub}, followed by the lower bound~\eqref{eq:aritytwolb}. Finally, we argue how our bounds are tight up to a factor of $\tilde{O}(d)$. 
\subsection{Upper Bound} \label{sec:aritytwo-ub}
We first consider the case when $\mathcal{H}$ is a star, which will be a basic building block for our algorithms for general $\cH$.
\subsubsection{$\mathcal{H}$ is a star} \label{sec:star-ub-1}
Let $P = (v_0, v_1, \dots, v_k)$ be the vertices of the star with $v_0$ as it's center. In this case, $\cH$ includes $k$ relations of the form $R_{v_0, v_i}$ for every $i \in [k]$. Note that computing the corresponding $\BCQ$ query $q$ can be solved in two steps -- first, the player containing $R_{v_0, v_1}$ broadcasts it to all players in $G$. Then, the resulting problem can be solving via \emph{set-intersection} where we compute $R'_{P} = \bigcap_{i = 2}^{k} R'_{v_i}$, where $R'_{v_i} = \{ \mathbf{t} \in R_{v_0, v_1}: \exists \mathbf{t'} \in R_{v_0, v_i} \text{ s.t. } \pi_{v_0}(\mathbf{t}) = \pi_{v_0}(\mathbf{t'}) \}$.

It follows that the final output of $q$ is $1$ if $R'_{P} \neq\emptyset$ and $0$ otherwise. We can solve the resulting set intersection problem using Theorem~\ref{thm:set-intersection} to compute $R'_{P}$. The procedure to compute $R'_P$ is described in Algorithm~\ref{algo:starAlgo}, which when combined with the fact that at most $\tO\left(\log_{2}(D)\right)$ bits can be communicated in each round, implies the following result.
\begin{cor} \label{cor:star-ub}
When $\cH$ is a star, for arbitrary graphs $G$ and subset of players $K$, we have
\[ \mathcal{D}(\BCQ_{\cH, N}, G, K) 
= \tO\left(\min_{\Delta \in [|V|]}\left( \frac{N}{\ST(G, K, \Delta)} + \Delta\right)\right).\]
\end{cor}
For the case when $G$ is a line with $k$ vertices, note that $\ST(G, K, \Delta)=0$ for every $\Delta>k-1$ and $\ST(G, K,k-1)=1$, which in turn implies the following. 
\begin{cor}  \label{cor:star-ub-line}
Let $\cH$ be a star and $G$ be a line with $k$ vertices. Then, we have
\[\mathcal{D}(\BCQ_{\cH, N}, G, K) \le N + k.\]
\end{cor}
Note that the above result is a generalization of the upper bound in Example~\ref{ex:r=2}.
\begin{algorithm}
\caption{Algorithm for Star} \label{algo:starAlgo}
\begin{algorithmic}[1]
\small
\State{\textbf{Input:} A star query with attributes $P = (v_0, \dots, v_k)$ and relations $\{R_{(v_0, v_i)} : i \in [k]\}$. Note that $v_0$ is the center.}
\State{\textbf{Output:} $R'_{P}$}
\State{The player containing $R_{v_0, v_1}$ broadcasts it to all players in $G$.}
\State{Each player containing a relation $R_{v_0, v_i}$ for every $i \in [2, k]$ computes $R'_{v_i} = \{ \mathbf{t} \in R_{v_0, v_1}: \exists \mathbf{t'} \in R_{v_0, v_i} \text{ s.t. } \pi_{v_0}(\mathbf{t}) = \pi_{v_0}(\mathbf{t'}) \}$ internally.}
\State{$R'_{P} = \bigcap_{i = 2}^{k} R'_{v_i}$ is computed using Theorem~\ref{thm:set-intersection}.}
\State{\Return{$R'_{P}$}}
\end{algorithmic}
\end{algorithm}

\subsubsection{$\mathcal{H}$ is a forest} \label{sec:aritytwo-forest}
We now use the above algorithm to obtain upper bounds for the case when $\mathcal{H}$ is a forest.
\begin{lemma} \label{lemma:aritytwo-forest-ub}
For arbitrary $G$, subset of players $K$ and $\cH$ being a forest, we have
\begin{equation} \label{eq:aritytwo-forest-ub}
\mathcal{D} (\BCQ_{\cH, N}, G, K) = \tO\left(y(\cH) \cdot \min_{\Delta \in [|V|]}\left(\frac{N}{\ST(G, K, \Delta)} + \Delta \right)\right).
\end{equation}
\end{lemma}
\begin{proof}[Proof Sketch]
We keep removing stars from trees in $\mathcal{H}$ in a bottom-up fashion and solve the induced query on each removed star using Algorithm~\ref{algo:starAlgo}. Since the number of stars we remove in this process is $y(\mathcal{H})$, the stated bound follows. The complete proof is in Appendix~\ref{app:atwo-forest-ub}.
\end{proof}

\subsubsection{The general case: $d$-degenerate graphs}
We now state our upper bound when $\mathcal{H}$ is a $d$-degenerate simple graph:
\begin{lemma} \label{lemma:aritytwo3}
For arbitrary $G$, subset of players $K$, and any $d$-degenerate simple graph $\cH$, we have
\begin{align*}
\label{eq:aritytwoubgeneral}
& \mathcal{D} (\BCQ_{\cH, N}, G, K)  = \\
& \tO\left(y(\cH) \cdot \min_{\Delta \in [|V|]}\left(\frac{N }{\ST(G, K, \Delta)} + \Delta \right) \right) \\ & + \tO\left(\tau_{\MCF}(G, K, n_2(\cH) \cdot d \cdot N)\right). \tag{2.1} \end{align*}
\end{lemma}
\begin{proof}[Proof Sketch]
We decompose $\mathcal{H}$ into two components using Definition~\ref{defn:forest-core} -- forest ($\Fo(\cH)$) and core ($\Co(\cH)$). We then use Lemma~\ref{lemma:aritytwo-forest-ub} to solve the induced query on $\Fo(\cH)$. For the core, we use the \emph{trivial protocol} of sending all the remaining relations to one player. The complete proof is in Appendix~\ref{app:atwo-ub}.
\end{proof}

\subsection{Lower Bound} \label{sec:aritytwo-lb}
We start with an overview, followed by lower bounds for the case when $\mathcal{H}$ is a forest and conclude with lower bounds for all simple graphs $\mathcal{H}$.

\subsubsection{Overview}
As we showed in Section~\ref{sec:overview-lb}, we start by considering an arbitrary $\TRB$ instance of size $m$ where $m$ is a function of $\cH$. We then show that it can be reduced to a suitable two-party $\BCQ$ instance, which is functionally equivalent to the $\TRB$ instance we started with. In particular, solving the $\BCQ$ instance we constructed indeed solves the $\TRB$ instance we started with. We denote this reduction succinctly by $\TRB_{m, N}\le \BCQ_{\cH, N}$. Finally, we generalize our results from the two-party setting to general $G$ using ideas from~\cite{topology-1,topology-3}. We crucially exploit the fact that our lower bounds are for worst-case assignment of input functions to players in $G$ and show a very specific class of assignments that achieves the required lower bound.

\subsubsection{$\mathcal{H}$ is a forest}
We prove the following lemma.
\begin{lemma} \label{lemma:lower-forest}
When $\mathcal{H}$ is a forest, we have 
\[\TRB_{\frac{y(\cH)}{2}, N} \le \BCQ_{\cH, N}.\]
\end{lemma}
\begin{proof}
For notational simplicity, define $y=y(\cH)$.
Given $\mathcal{H}$ and a $\TRB_{\frac{y}{2}, N}$ instance we design a corresponding $\BCQ_{\cH, N}$ instance. As $\mathcal{H}$ is bipartite, let $(L,R)$ be the node partition of $\mathcal{H}$ and consider the set $O_L$ ($O_R$ resp.) consisting of all nodes of degree at least two included in $L$ ($R$ resp.). Let $O$ equal the largest of $O_L$ and $O_R$ (i.e., $O$ consists of nodes of odd or even distance from the roots of the forest).
Note that $|O| \geq \frac{y}{2}$,\footnote{Note that in the arity two case, it is easy to construct a GYO-GHD with $y$ internal nodes using the structure of $\cH$. In particular, we can make $\Co(\cH)$ the root and the remaining structure of the GYO-GHD mimics the structure of the forest $\Fo(\cH)$.} and assume w.l.o.g. that the size of $O$ is exactly $\frac{y}{2}$ (otherwise we take a subset of $O$).
We associate a pair of sets $(S_o,T_o)$ from $\TRB_{\frac{y}{2}, N}$ with each node $o \in O$, such that 
\begin{equation} \label{eq:tribestree}
\TRB_{\frac{y}{2}, N}(\hat{S},\hat{T}) = \bigwedge_{o \in O} \SD_{N}(S_o, T_o),
\end{equation}
where $\SD_{N}(S_o, T_o) = 1$ if $S_o \cap T_o \neq \emptyset$ and $0$ otherwise.

We now construct a corresponding $\BCQ_{\cH, N}$ instance in detail. 
We start by defining a pair of relations corresponding to each pair $({S_o}, {T_o})$. 
Let $o \in O$. 
If $o$ has a parent in $\mathcal{H}$, let $o_p$ be its parent.
Let $o_c$ be a child of $o$.
We consider the relations $R_{S_o} = S_{o} \times \{1\}$ and $R_{T_o} = T_{o} \times \{1\}$, where the attribute set of $R_{S_o}$ is $(o,o_c)$ and that of $R_{T_o}$ is $(o,o_p)$. Here we treat $S_o$ and $T_o$ as subsets of $[N]$ (instead of elements in $\{0,1\}^N$).
In the case that $o$ does not have a parent node, it is a root in $\mathcal{H}$ with at least two children, and thus we can set $o_p$ to be a child of $o$ that differs from $o_c$.
Thus, $\TRB_{\frac{y}{2}, N}(\hat{S},\hat{T}) = 1$ iff for each $o \in O$, the join $R_{S_o}  \Join R_{T_o}$ is not empty.
To complete the description of the $\BCQ$ instance, for each $o \in O$, we associate all additional edges $(o,v)$ adjacent to $o$ in $\mathcal{H}$ with the relation $[N] \times \{1\}$ on attributes $(o,v)$; and remaining edges $(u,v)$ that are not adjacent to any $o \in O$ with the relation $\{1\} \times \{1\}$. Note that no two vertices $o_1,o_2 \in O$ are adjacent in $\mathcal{H}$. Let us denote the $\BCQ$ instance constructed above by $q_{\cH,\hat{S},\hat{T}}$.

To complete the proof, we show that $q_{\cH,\hat{S},\hat{T}}=1$ iff \\ $\TRB_{\frac{y}{2}, N}(\hat{S},\hat{T})=1$.
If $q_{\cH,\hat{S},\hat{T}}= 1$ then there exists a tuple $\mathbf{t} \in \prod_{v\in V(\cH)}\Dom(v)$ that satisfies all relations in $q_{\cH,\hat{S},\hat{T}}$, i.e. $\mathbf{t}_e\in R_e$ for every $e\in\cE$.
Specifically, for each $o \in O$, $R_{S_o}  \Join R_{T_o}$ is not empty which implies that  $\TRB_{\frac{y}{2}, N}(\hat{S},\hat{T}) = 1$.
Alternatively, if $\TRB_{\frac{y}{2}, N}(\hat{S},\hat{T}) = 1$, we can find a tuple $\mathbf{t} \in \prod_{v\in V(\cH)}\Dom(v)$ that satisfies all relations in $q_{\cH,\hat{S},\hat{T}}$. For each $o \in O$ we set $\pi_o(\mathbf{t})$ to be any element in the intersection of $S_o$ and $T_o$, and for all remaining nodes $v$ we set $\pi_v(\mathbf{t})=1$. It holds that the relations corresponding to edges of the form $(o,o_p)$, $(o,o_c)$, $(o,v)$, and $(u,v)$ described above are all satisfied. This concludes our proof.
\end{proof} 
Note that the above argument was independent of $G$. We now use the structure of $G$ to obtain a lower bound on \\$\mathcal{R}(\BCQ_{\cH, N}, G, K)$ using known results for $\TRB_{\frac{y}{2}, N}$.

\paragraph{Lower bounds dependent on $G$} \label{sec:lb-dep-g}
We show the following lower bound for arbitrary $G$, assuming worst-case assignment of relations to players in $K$.
\begin{lemma} \label{lemma:tree-lb-1}
For any topology $G$ and $\cH$ being a forest, 
\[\mathcal{R}(\BCQ_{\cH, N}, G, K) \ge \tOm\left(\frac{y(\cH) \cdot N}{\MC(G, K)}\right). \]
\end{lemma}
\begin{proof}
We first consider a min-cut $(A, B)$ of $G$ that separates $K$, where $A$ and $B$ denote the set of vertices in each partition ($A \cup B = V$). Using the notation given in the proof of Lemma~\ref{lemma:lower-forest}, let $q_{\cH,\hat{S},\hat{T}}$ be the query corresponding to a given instance $\TRB_{\frac{y}{2}, N}(\hat{S},\hat{T})$. We assign relations $\{R_{S_o}\}_{o\in O}$ to vertices in $A$ and relations \\ $\{R_{T_o}\}_{o \in O}$ to vertices in $B$. The other relations in $q_{\cH,\hat{S},\hat{T}}$ can be assigned arbitrarily. Note that any protocol to compute $q_{\cH,\hat{S},\hat{T}}$ on $G$ gives a two-party protocol (Alice, Bob) for $\TRB_{\frac{y}{2}, N}$. In particular, Alice gets the sets $\{S_o\}_{o \in O}$ (corresponding to $R_{S_o}$) assigned to vertices in $A$ and Bob gets the sets $\{T_o\}_{o \in O}$ (corresponding to $R_{T_o}$) assigned to vertices in $B$ (ignoring the additional relations). \tbc{It follows that if there exists a $\cR(\BCQ_{\cH, N}, G, K)$ round protocol on $G$, then we have a two-party protocol (i.e., on a graph $\cG = (\{a, b\}, (a, b))$) with at most $\cR(\BCQ_{\cH, N}, G, K) \cdot \MC(G,K) \cdot \ceil{\log_{2}(\MC(G, K))}$ rounds. Indeed, we can simulate the two-party protocol on $G$ across the cut $(A,B)$, where Alice is responsible for $A$ and Bob for $B$. In particular, if Alice needs to send a message to Bob (or vice-versa), it will be sent across edges crossing the cut. Note that in each round, at most $\MC(G,K) \ceil{\log_{2}(\MC(G, K))}$ bits will be exchanged between Alice and Bob. We need $\ceil{\log_{2}(\MC(G, K))}$ bits in order to know the edge on which the message was sent. We can now invoke \eqref{eq:trb-lb} to obtain a lower bound of
\[\cR(\BCQ_{\cH, N}, G, K) \ge \tOm \left( \frac{y(\cH) \cdot N}{\MC(G, K)}\right).\]
}
\end{proof}

\subsubsection{General $\cH$}
We are now ready to prove our general lower bound for all simple graphs $\mathcal{H}$.
\begin{thm} \label{thm:aritytwo2}
For arbitrary $G$, $K \subseteq V$, and graph $\cH$, we have
\[ \cR(\BCQ_{\cH, N}, G, K) \ge \tOm\left(\frac{(y(\cH) + n_2(\cH)) \cdot N}{\MC(G, K)} \right).\]
\end{thm}
\begin{proof} [Proof Sketch]
We present a proof sketch here. For notational convenience, define $y=y(\cH)$ and $n_2=n_2(\cH)$.
Let $m = \max\left(\frac{y}{2}, \frac{n_2}{2\log(n_2)}\right)$.
In general, as in the proof of Lemma~\ref{lemma:lower-forest}, given $\mathcal{H}$ and a TRIBES instance $\TRB_{m, N}(\hat{S},\hat{T})$ we construct a $\BCQ$ instance $q_{\cH,\hat{S},\hat{T}}$ such that $q_{\cH,\hat{S},\hat{T}}=1$ iff $\TRB_{m, N}(\hat{S},\hat{T})=1$. 
To this end we need to `embed' the $m$ pairs of sets $(S_i,T_i)$ from $\TRB_{m, N}(\hat{S},\hat{T})$ as relations in $q_{\cH,\hat{S},\hat{T}}$. 
For $m = \frac{y}{2}$, we embed the pairs $(S_i, T_i)$ in the forest $\Fo(\cH)$ as done in Lemma~\ref{lemma:lower-forest}. 
For $m = \frac{n_2}{2 \cdot \log(n_2)}$, we consider $\Co(\cH)$. 
We then show that it must be the case that $\Co(\mathcal{H})$ either includes $\left(\frac{n_2}{2\log(n_2)}\right)$
vertex-disjoint cycles (or) it has an independent set of size $\Omega(n_2)$. In both cases, we show how one can embed $\frac{n_2}{2\log(n_2)}$ pairs $(S_i,T_i)$ of $\TRB_{m, N}(\hat{S},\hat{T})$ in $C(\mathcal{H})$. We defer the proof to Appendix~\ref{app:atwo-lb}.

Assuming the above embeddings, we conclude that $q_{\cH,\hat{S},\hat{T}}=1$ iff $\TRB_{m, N}(\hat{S},\hat{T})=1$, where $m = \max\left(\frac{y}{2}, \frac{n_2}{2\log(n_2)}\right)$. Since sum and max are within a factor $2$ of each other, we can write $m \ge \frac{y}{4} +  \frac{n_2}{4\log(n_2)}$. We can now apply ideas from the proof of Lemma~\ref{lemma:tree-lb-1} to obtain the required lower bound $\tOm\left( \frac{(y+n_2) \cdot N}{\MC(G, K)}\right)$.
\end{proof}
Note that in Theorem~\ref{thm:aritytwo1}, the upper bound follows from Lemma~\ref{lemma:aritytwo3} and the lower bound from Theorem~\ref{thm:aritytwo2}.
We conclude this section by noting that when $N\ge |V|^2$, our upper and lower bounds differ by $\tilde{O}(d)$ factor (for worst-case assignments of relations to players). In particular, Theorem~\ref{thm:lau} implies that the first two terms in the upper and lower bounds match up to an $\tilde{O}(1)$ factor. In Appendix~\ref{app:mcf-mc}, we show that for worst-case assignment of relations, the second terms in the upper and lower bounds also differ by a $\tilde{O}(d)$ factor, as desired.
\section{Hypergraphs $\cH$ and General $\FAQ$} \label{sec:general_faq}

Our results generalize fairly seamlessly to hypergraphs $\cH$. For constant $d, r$, our upper and lower bounds match. However, for non-constant $d$, we have a gap of $\tilde{O}(d^{2} \cdot r^{2})$, which is worse than our gap of $\tilde{O}(d)$ for the arity two case. The technical details are deferred to Appendix~\ref{app:gen-arity}.

We extend our results from $\BCQ$ to the general $\FAQ$ problem. We define the general $\FAQ$ problem here, which is a generalization of $\FAQS$. We are given a multi-hypergraph $\cH=(\cV,\cE)$ where for each hyperedge $e\in\cE$, we also have an input function $f_e:\prod_{v\in e} \Dom(v)\to \D$. In addition, we are given a set of {\em free variables} $\cF\subseteq \cV : |\cF| = \ell$ and\footnote{\tbc{For a fixed $\cF$, the vertices in $\cV$ can be renumbered so that $\cF = [\ell]$ w.l.o.g.}} we would like to compute the function:
\begin{equation} \label{eq:faqgeneraldef}
\phi\left(\vx_{[\ell]}\right) =\underset{x_{\ell + 1} \in \Dom(x_{\ell + 1})}{\tenssum{}^{(\ell + 1)}}\dots \underset{x_{n} \in \Dom(x_{n})}{\tenssum{}^{(n)}}  \underset{S \in \cE} {\tensprod{}} f_{S}(\vx_{S}),
\end{equation}
where $\vx=(x_u)_{u\in\cV}$ and $\vx_S$ is $\vx$ projected down to co-ordinates in $S\subseteq \cV$. 
The variables in $\cV \setminus \cF$ are called \emph{bound variables}.
For every bound variable $i > \ell$, $\tenssum{}^{(i)}$ is a binary (aggregate) operator on the domain $\D$. Different bound variables may have different aggregates. 
Finally, for each bound variable $i > \ell$ either $\tenssum{}^{(i)} = \tensprod{}$ (\emph{product aggregate}) or $(\D,\tenssum{}^{(i)},\tensprod{})$ forms a commutative semiring (\emph{semiring aggregate}) with the same additive identity $\mathbf{0}$ and multiplicative identity $\mathbf{1}$.
As with $\FAQS$, we assume that the functions are input in the {\em listing} representation, i.e. the function $f_e$ is represented as a list of its non-zero values: $R_e=\{(\vy,f_e(\vy))|\vy\in\prod_{v\in e} Dom(v):f_e(\vy)\neq \vzero\}$. Note that when $\tenssum{}^{(i)} = \tenssum{}$ is the same semiring aggregate for every $\ell < i \le n$, we have the $\FAQS$ problem. 

For any $\D$, let $\FAQ_{\D, \cH, N, \cF}$ denote the class of $\FAQ$ problems, where each function in $\cH$ has at most $N$ non-zero entries. (Note that we are not explicitly stating the operators for the bound variables $(\tenssum{}^{(\ell + 1)}, \dots, \tenssum{}^{(n)})$ since our upper and lower bounds hold for all such operators.) Let $\cR(\FAQ_{\D, \cH, N, \cF}, G, K)$ denote the minimum worst-case number of rounds needed by a randomized protocol with error at most $\frac{1}{3}$ that computes any query in $\FAQ_{\D, \cH, N, \cF}$ on $G$ with functions assigned to nodes in $K$. We prove the following results.
\begin{thm} \label{cor:faq-const1}
For $O(1)$-degenerate hypergraphs $\cH$ with $O(1)$-arity, we have 
\[\cR(\FAQ_{\D, \cH, N, \cF}, G, K) =  \tilde{\Theta}\left(\frac{ (y(\cH)  + n_2(\cH)) \cdot N}{\MC(G, K)}\right) \]
for any $\D$, specific choices of $\cF$, arbitrary $G$ and $K$. When $G$ is a line (as in the first row in Table~\ref{tab:results}), $\MC(G, K) = 1$.
\end{thm}
\begin{thm} \label{thm:faq2main}
For general degenerate hypergraphs $\cH$ with arity at most $r$, we have
\begin{align*}
& \cD \left(\FAQ_{\D, \cH, N, \cF}, G, K \right)  \\ 
& = \tO \left( y(\cH) \cdot \min_{\Delta \in [|V|]}\left(\frac{N \cdot r}{\ST(G, K, \Delta)} + \Delta \right) \right)\\
&  + \tO \left(\tau_{\MCF}(G, K, n_2(\cH) \cdot d \cdot r \cdot N)\right)
\end{align*}
and
\[\cR \left(\FAQ_{\D, \cH, N, \cF}, G, K \right) \ge \tOm\left(\frac{ \left(d \cdot y(\cH)  + n_2(\cH) \right)\cdot N}{d \cdot r \cdot \MC(G, K)} \right).\]
The lower bound differs from the upper bound by a factor of $O(r^2 d^2)$ in the worst case.
\end{thm}
We would like to mention here once again that our upper bound is a deterministic protocol and the lower bound is for randomized protocols. Details are in Appendix~\ref{app:general-faq}.

\section{Matrix Chain Multiplication} \label{SEC:MATRIX-CHAIN}
We consider the following $\FAQS$ problem. The network topology $G$ has $k+2$ players $P_0,\dots,P_{k+1}$ such that $(P_i,P_{i+1})$ is an edge (i.e. $G$ is a line) where $P_0$ receives $\vx\in\F_2^N$ and $P_i$ for $i\in [k]$ receives $\vA_i\in\F_{2}^{N\times N}$. Player $P_{k+1}$ wants to compute $\vA_k\cdot \vA_{k-1}\cdots \vA_1\cdot \vx$.
Alternatively, for every $i \in [k]$, define
$\vy_i=\vA_i\cdot \vy_{i-1}$,
with $\vy_0=\vx$. Note that we want to compute $\vy_k$.
Note that this is an $\FAQS$ problem since we can re-write the above as
\begin{equation} \label{eq:matrix-chain-faqs}
\phi(z_{k})=\sum_{(z_i)_{i=0}^{k-1} \in [N]^k} \left(\prod_{j=1}^k A_j(z_{j},z_{j-1})\right) X(z_0),
\end{equation}
where the functions satisfy $A_j(x,y)=\vA_j[x,y]$ and $X(z)=\vx[z]$ for every triple of indices $x,y,z\in [N]$. \footnote{Note that $N$ here is the dimension of the matrices as opposed to the number of non-zero entries used in the previous sections.}

We note that this problem can be solved in $O(kN)$ rounds.\footnote{Note that the trivial algorithm of sending all inputs to a single player takes $\Omega(kN^2)$ rounds.} 
\begin{prop}
\label{prop:matrix-chain}
The $\FAQS$ problem from~\eqref{eq:matrix-chain-faqs} can be computed in $O(kN)$ rounds.
\end{prop}
\begin{proof}
We describe our algorithm here. We start by computing $\vy_{1} = A_1 \cdot \vx$, which can be done in $O(N)$ rounds. We then successively compute $\vy_{i} = A_{i} \cdot \vy_{i -1}$ for every $i \in [2, k]$. Note that this takes $O(k \cdot N)$ rounds in total and we would get the final answer in $\vy_{k}$. 
\end{proof}
We remark that when $k$ is large, a bottom-to-top fashion merge algorithm can achieve $O(N^2 \log k + k)$ rounds (details are deferred to Appendix~\ref{app:k>=N}). In the next section, we prove a tight lower bound of $\Omega(kN)$ for the case $k \leq N$. 

\subsection{Difference from Online Matrix Vector Multiplication} \label{sec:oumv}
We state the Online Matrix Vector Multiplication problem formally here.
\begin{problem} [Online Matrix Vector Multiplication]
Given an $N \times N$ Boolean matrix $M$, we receive $N$ Boolean $N \times 1$ vectors $\mathbf{v_1} \dots, \mathbf{v_n}$ one at a time, and are required to output $M \cdot \mathbf{v_i}$ (over the Boolean semiring $(\D= \{0,1\},\vee,\wedge)$) before seeing the vector $\mathbf{v_{i+1}}$, for all $i \in [n - 1]$.
\end{problem}
Note that this problem and our problem are somewhat dual problems and our results do not imply anything for this problem. 

\subsection{The Lower Bound}
We will argue that the upper bound of $O(kN)$ rounds in Proposition~\ref{prop:matrix-chain} is tight if $k \leq N$. Before we do that we collect some definitions and results related to the min-entropy of a random variable.
\subsubsection{Background}
The {\em min-entropy} of a random variable $X$ is defined as 
\begin{align*}
\hinf{X}:=-\log \max_{x \in \text{supp}(X)}\Pr[X = x].
\end{align*} 
For a random variable $X$ and an event $\mathcal{E}$ that is possibly correlated with $X$, define 
\begin{align*}
\hinf{X\mathcal{E}} = -\log \max_{x \in \text{supp}(X)}\Pr[X = x, \mathcal{E}].
\end{align*}
Notice that in the above definition, we do not `normalize' $\Pr[X = x, \mathcal{E}]$ by a factor of $\Pr[\mathcal{E}]$.

For random variables $X$ and $Y$, the conditional smooth min-entropy $\hinfe{X|Y}{\eps}$ is defined as
\begin{align*}
\hinfe{X|Y}{\eps} & = \sup_{\mathcal{E}}\min_{y \in \ssupp(Y)}\hinf{X\mathcal{E}|Y = y} \\
&= \sup_{\mathcal{E}}\left(-\log \max_{(x, y) \in \ssupp(X, Y)} \Pr[\mathcal{E}, X = x|Y = y]\right)
\end{align*}
where the quantification over $\mathcal{E}$ is over all events $\mathcal{E}$ (which can be correlated with $X$ and $Y$) with $\Pr(\mathcal{E}) \geq 1-\eps$. When $Y$ is a deterministic variable (in other words, we are not conditioning on any randomized variable), then we simply use $\hinfe{X}{\eps}$: 
\begin{equation}
\label{def:h-inf-eps-1}
\hinfe{X}{\eps} = \sup_{\mathcal{E}} \hinf{X\mathcal{E}},
\end{equation}
where again the quantification over $\mathcal{E}$ is over all events $\mathcal{E}$ with $\Pr(\mathcal{E}) \geq 1-\eps$. 

The following results will be useful in our analysis.
\begin{lemma}[Lemma $4$ and Lemma $7$ of~\cite{smooth-min-ent}] \label{lem:cond-entropy}
Let $Y$ be a random variable with support size at most $2^{\ell}$. Then we have for any $\eps \geq 0, \eps' > 0$ and random variable $X$, that
\begin{align*}
\hinfe{X|Y}{\eps + \eps'} \ge \hinfe{X}{\eps}-\ell-\log(1/\eps').
\end{align*}
\end{lemma}
\begin{thm} \label{thm:CG-ext}
Let the constant $\gamma>0$ be small enough. Let $\vx\in\F_2^N$, $\vA\in\F_2^{N\times N}$ and $\vY$ be random variables such that for every $y \in \ssupp(\vY)$, $\vx$ and $\vA$ are independent conditioned on $\vY = y$. Moreover, for some reals $\eps_1,\eps_2 \geq 0$, we have 
\begin{align*}
\hinfe{\vA|\vY}{\eps_1} \ge (1 - \gamma) \cdot N^{2} \text{ and } \hinfe{\vx|\vY}{\eps_2} \ge \alpha \cdot N,
\end{align*}
where $\alpha\eqdef  3\gamma +\sqrt{2\gamma} + h(\sqrt{2\gamma})$ and $h(p)=-p\log_2{p}-(1-p)\log_2(1-p)$ for any $p : 0 < p \le 1$.
Then, we have 
\begin{align*}
\hinfe{\vA\vx|\vY}{\eps_1+\eps_2+2^{-\Omega(\gamma N)}}\ge \left(1-\sqrt{2\gamma}\right)\cdot N.
\end{align*}
\end{thm}
The proof of Theorem~\ref{thm:CG-ext} follows from known results in pseudorandomness~\cite{sse,david} and is deferred to Appendix~\ref{sec:matrix-vec-main}.

\subsubsection{Showing Proposition~\ref{prop:matrix-chain} is tight for $k \leq N$}
At a high level, we will prove by induction that for player $P_i$ at time about $\gamma iN$, the min-entropy of $\vy_{i-1}$ is at least $\alpha\cdot N$ (and the situation at $P_{i+1}$ should be similar). Since by this time $P_{i+1}$ would have received at most $O(\gamma iN)\le O(\gamma N^2)$ bits, this means $\vA_{i}$ has min-entropy at least $(1-\gamma)N^2$. Thus, we can apply Theorem~\ref{thm:CG-ext} to argue that at $P_{i+1}$ the min-entropy of $\vy_i=\vA_i\cdot \vy_{i-1}$ is large. To finish the inductive argument we have to wait for $\gamma N$ more steps but by Lemma~\ref{lem:cond-entropy}, even then $\vy_i$ will still have high enough min-entropy.
It is natural to wonder if we can make the same argument using Shannon entropy instead of min-entropy. In Appendix~\ref{app:shannon}, we show that this is not possible.

We define some useful notation before we prove the lower bound. At any given time $t$, let ${\vm}^i(t)$  denote the transcript of messages exchanged on the link between $P_{i-1}$ and $P_i$ till time $t$. 
For $i\in [k+1]$, define
$t_i= \frac{\gamma}{4}\cdot iN$, and $\tilde {\vm}^i = {\vm}^i(t_i)$.
For a random variable ${\vm}$, we will use
${m}$ to denote a specific value of the random variable ${\vm}$.
In addition, we use $\tilde {\vm}^{[i]}$ and ${\tilde {m}}^{[i]}$ to denote the tuples  $(\tilde {\vm}^1, \tilde {\vm}^2, \cdots, \tilde {\vm}^i)$ and $({\tilde {m}}^1, {\tilde {m}}^2, \cdots, {\tilde {m}}^i)$ respectively. 

Let ${\eps^*} = 2^{-\Omega(\gamma N)}$ be at least thrice the maximum of $2^{-\gamma N/4}$ and the $2^{-\Omega(\gamma N)}$ term in Theorem~\ref{thm:CG-ext}. We will argue the following.
\begin{lemma} \label{lem:min-ent-comm}
Let $\vA_i$ for every $i\in [k]$ and $\vx$ be all uniformly and independently distributed. Let $\gamma>0$ be such that\footnote{There exists a value $\gamma \ge 0.01$ (for large enough $N$) that satisfies the required conditions.}
\begin{equation} \label{eq:gamma-cond}
4\gamma+\sqrt{2\gamma}+h(\sqrt{2\gamma})\le 1,
\end{equation}
and $\gamma N/4$ is an integer.
Then we have the following for every $i\in [k+1]$:
\begin{align} \label{eq:star}
\hinfe{\vy_{i-1}|\tilde {\vm}^{[i]}}{i\eps^*}\ge N(1-\gamma-\sqrt{2\gamma}).
\end{align}
\end{lemma}
\begin{proof}
The proof is by induction. For the base case of $i=1$, Lemma~\ref{lem:cond-entropy} implies that (recall that $\tilde {\vm}^1 = {\vm}^1(t_1) ={\vm}^1(\gamma N/4)$ and $\vy_0 = \vx$):
\begin{align*}
\hinfe{\vy_0|\tilde {\vm}^1}{\eps^*} & \geq \hinf{\vx} - \gamma N /4 - \log(1/\eps^*) \\
& \geq  N(1-\gamma/4-\gamma/2) \\ 
& \ge N(1-\gamma-\sqrt{2\gamma}).
\end{align*}
Thus,~\eqref{eq:star} holds for $i = 1$.
	
We assume~\eqref{eq:star} holds for some $i \geq 1$; we prove that it also holds with $i$ replaced by $i + 1$.  For any interval $[\ell,r]$ we use $\vA_{[\ell:r]}$ to denote the tuple $(\vA_{\ell},\dots,\vA_r)$. Conditioned on $\tilde {\vm}^i ={\tilde {m}}^i$, since all communication between $P_1,\dots,P_{i-1}$ and $P_i,\dots,P_{k+1}$ are independent, we have
\begin{enumerate}[topsep=3pt,itemsep=0pt,label=(B\arabic*)]
\item $(\vx, \vA_{[1:i-1]})$ is independent of $\vA_{[i, k]}$.
\item $\vy_{i-1}$ and $\tilde {\vm}^{[i-1]}$ are determined by $(\vx, \vA_{[1:i-1]})$.
\item ${\vm}^{i+1}(t_i)$ is determined by $\vA_{[i:k]}$.
\end{enumerate}
The above properties imply the following, which will be used many times in our analysis:
\begin{itemize}[topsep=3pt,itemsep=0pt,label=(B\arabic*)]
\item[(C)] Conditioned on $\tilde {\vm}^i ={\tilde {m}}^i$, $(\vx, \vA_{[1:i-1]}, \vy_{i-1}, \tilde {\vm}^{[i-1]})$ and $(\vA_{[i, k]}, {\vm}^{i+1}(t_i))$ are independent.
\end{itemize}
	
By (C), $\vy_{i-1}\big|\big(\tilde {\vm}^{[i]}, {\vm}^{i+1}(t_i)\big)=\big({\tilde {m}}^{[i]}, m^{i+1}(t_i)\big)$ has the same distribution as $\vy_{i-1}\big|\tilde {\vm}^{[i]}={\tilde {m}}^{[i]}$. By the inductive hypothesis,  we have 
\begin{align}
\hinfe{\vy_{i-1}\Big|\big(\tilde {\vm}^{[i]}, {\vm}^{i+1}(t_i)\big)}{i\eps^*} &= \hinfe{\vy_{i-1}\Big|\tilde {\vm}^{[i]}}{i\eps^*} \nonumber\\
&\ge N(1-\gamma-\sqrt{2\gamma}). \label{eq:x-req-ent}
\end{align}
We show in Appendix~\ref{app:together} that equality can be achieved in the above lemma. Further, Lemma~\ref{lem:cond-entropy} (with $\eps=0$ and $\eps'=\eps^*/3$) implies that 
\begin{align}
\hinfe{\vA_i\Big|\big(\tilde {\vm}^{i}, {\vm}^{i+1}(t_i)\big)}{\eps^*/3}& \ge N^2-2i\cdot\frac{\gamma}{4}\cdot N - \log(\eps^*/3) \notag\\
	\label{eq:A-req-ent}
&\ge N^2(1-\gamma).
\end{align}
	
Again by (C), $\vA_i\big|\big(\tilde {\vm}^{i}, {\vm}^{i+1}(t_i)\big)=\big(\tilde {m} ^i, m^{i+1}(t_i)\big)$ has the same distribution as $\vA_i\big|\big(\tilde {\vm}^{[i]}, {\vm}^{i+1}(t_i)\big)=\big(\tilde m^{[i]}, m^{i+1}(t_i)\big)$. Equality in~\eqref{eq:x-req-ent} and~\eqref{eq:A-req-ent} imply that
\begin{equation} \label{eq:A-req-ent1}
\hinfe{\vA_i\Big|\big(\tilde {\vm}^{[i]}, {\vm}^{i+1}(t_i)\big)}{\eps^*/3}\ge N^2(1-\gamma).
\end{equation}
	
By~\eqref{eq:x-req-ent},~\eqref{eq:A-req-ent1}, and (by (C)) the fact that $\vA_i$ and $\vy_{i-1}$ are independent conditioned on $(\tilde {\vm}^{[i]}, {\vm}^{i+1}(t_i))= $\\$(\tilde {m} ^{[i]}, m^{i+1}(t_i))$, we have the following via Theorem~\ref{thm:CG-ext}:
\begin{align*}
\hinfe{\vy_{i}=\vA_i\vy_{i-1}\Big|\big(\tilde {\vm}^{[i]}, {\vm}^{i+1}(t_i)\big )}{(i+2/3) {\eps^*}}\ge N(1-\sqrt{2\gamma}),
\end{align*}
as long as $1-\gamma-\sqrt{2\gamma}\ge 3\gamma+\sqrt{2\gamma}+h(\sqrt{2\gamma}),$ which follows from~\eqref{eq:gamma-cond}. By applying Lemma~\ref{lem:cond-entropy} again (with $\eps=(i+2/3) {\eps^*}$ and $\eps'=\eps^*/3$), we get\footnote{Note that we are not conditioning on ${\vm}^{i+1}(t_{i+1})={\vm}^{i+1}(t_i+\gamma N/4)$ instead of the earlier ${\vm}^{i+1}(t_i)$.}
\begin{align*}
& \hinfe{\vy_{i}\big|\tilde {\vm}^{[i+1]}}{(i+1) {\eps^*}} \\
& \geq \hinfe{\vy_i\Big|\big(\tilde {\vm}^{[i]}, {\vm}^{i+1}(t_i)\big )}{(i+2/3)\eps^*} - N\gamma/4 - \log(\eps^*/3) \\
& \ge N \left(1-\sqrt{2\gamma}-\frac{\gamma}{4}-\frac{\gamma}{2})\ge N(1-\gamma-\sqrt{2\gamma} \right),
\end{align*}
as desired.
\end{proof}
The above immediately gives us our lower bound.
\begin{thm} \label{thm:matrix-chain-lb}
Any protocol that solves the $\FAQS$ problem from~\eqref{eq:matrix-chain-faqs} with $k \leq N$ and large enough $N$, with success probability at least $\frac{1}{2}$, takes $\Omega(kN)$ rounds.
\end{thm}
We will need the following result in proving the above theorem.
\begin{lemma} \label{lemma:correct-probability}
Assume $\hinfe{X|Y}{\epsilon} \geq L$. Then for every function $f:\ssupp(Y) \to \ssupp(X)$, we have $\Pr[f(Y) = X] \leq \epsilon + 2^{-L}$.
\end{lemma}
The proof of this lemma is deferred to Appendix~\ref{app:matrix-chain-lb}. We are now ready to prove Theorem~\ref{thm:matrix-chain-lb}.
\begin{proof} [Proof of Theorem~\ref{thm:matrix-chain-lb}]
Let $\Pi$ be any protocol with at most $t_{k+1} = \gamma(k+1)N/4$ rounds. Lemma~\ref{lem:min-ent-comm} implies that at the end of the protocol, we have 
\begin{align*}
\hinfe{\vy_k\big|\tilde {\vm}^{[k+1]}}{(k+1){\eps^*}}\ge N(1-\gamma-\sqrt{2\gamma}).
\end{align*}
This implies that even if the player $k+1$ is given $\tilde {\vm}^{[k+1]}$ (instead of only $\tilde {\vm}^{k+1} = {\vm}^{k+1}(t_{k+1})$), it can only output the correct answer with probability at most
\begin{align*}
(k+1){\eps^*}+2^{-N(1-\gamma-\sqrt{2\gamma})},
\end{align*}
by Lemma~\ref{lemma:correct-probability} (here $f(Y)$ is the output at $P_{k+1}$ for $Y=\tilde {\vm}^{[k+1]}$ and $X=\vy_k$). For large enough $N$, the above quantity is less than $\frac{1}{2}$.
\end{proof}
\section{Related Work} \label{sec:related}
We survey relevant related work here.
\paragraph{Parallel Database Query Computation.} The MPC model has seen a lot of research activity in the last few years~\cite{BKS13,BKS14,BKS16,JR16,gym,mpc}. We compare these models with ours in Section~\ref{sec:intro_model} and Appendix~\ref{sec:comparison}.

\paragraph{Distributed Computing and Communication Complexity.} As stated earlier, our model is similar to (and different from) the $\congest$ model in distributed computing~\cite{peleg-book}. Recently, there has been work on the same model as ours but instead of minimizing the number of rounds, they focus on minimizing the total communication of the protocols~\cite{T87,PVZ12,WZ12,BEOPV13,topology-1,topology-2}. Finally,~\cite{topology-3} obtained results on minimizing the number of rounds of protocols in our setup for some well-studied functions in two-party communication complexity literature.

\paragraph{Widths of GHDs.} The Internal Node Width $\mdw(\cH)$ of a GHD focuses on minimizing the number of internal (non-leaf) nodes in GHDs of acyclic hypergraphs. There is a related notion for Tree Decompositions called \emph{Lean} Tree Decompositions (LTDs)~\cite{mdw1,mdw2,mdw3}. The LTDs minimize the number of internal nodes in the following way -- they try to retain only pairs of connected internal nodes whose intersection forms a bridge in the original graph $\cH$. The other nodes are forced to become leaves of one of the internal nodes. While our construction procedure of GYO-GHDs tries to convert existing internal nodes to become leaves of some internal nodes, we do not (yet) see an exact one-to-one mapping from GYO-GHDs to LTDs. We would like to mention here that both the goals of GYO-GHDs and LTDs are the same i.e., to minimize the number of internal nodes. We would like to note here that $\mdw(\cH)$ can potentially reduce the depth of the GHD as well. Reducing depth of GHDs (sometimes by increasing the treewidth) has been considered before~\cite{gym,ghddepth1,ghddepth2}. 

For GHDs, the problem of computing GHDs that minimize certain cost functions of the HDs are studied in the framework of Weighted GHDs~\cite{WeightedGHD1,WeightedGHD2}. For a given hypergraph $\cH$, one way to map our notion of width to their setting is to consider a vertex aggregating function on every candidate HD $\cT$ for $\cH$. In particular, we can write
\begin{equation} \label{eq:wghd}
\Lambda_{H}^{f'}(\cT) = \sum_{v' \in V(\cT)} f'_{\cH}(v'),
\end{equation}
where $f'_{\cH} = 1$ if $v'$ is a internal node and $0$ otherwise. It follows that $f'_{\cH}$ can be computed in linear time in size of $\cT$. Given this setup, Theorem $3.4$ in~\cite{WeightedGHD1} proves that computing Minimal GHDs over HDs for arbitrary vertex aggregation functions is NP-Hard. 

However, this does not hold in our case since there is always a GHD with one internal node (containing all the variables in $\cH$). As a result, considering the minimization over all GHDs for our case is trivial and doesn't give tight results since we minimize over GYO-GHDs (Construction~\ref{lemma:forest-core}). For the tightness of our bounds for $\cH$ with constant degeneracy and constant arity, we only need an $O(1)$-factor approximation of Internal-Node-Width, which we achieve (details in Appendix~\ref{app:gen-arity}). 

We refer the reader to~\cite{width-survey} for a recent survey on widths for GHD.

\paragraph{Entropy in Communication Complexity.} Information complexity by now is a well-established sub-field of communication complexity that uses Shannon entropy to measure the amount of information exchanges in a two-party communication protocol. Information complexity was introduced in the work of Chakrabarti et al.~\cite{ent-1} and later used in a systematic way to tackle multiple problems in~\cite{ent-2}. To the best of our knowledge, min-entropy has only been used very recently in communication complexity~\cite{min-ent-1,min-ent-2} though it has found numerous applications in pseudorandomness and cryptography for at least two decades~\cite{salil-book}. Our work adds to the recently growing body of work that uses min-entropy to prove communication complexity results~\cite{min-ent-3}.

\section{Future Work} \label{sec:open}
There are many interesting and challenging open questions arising from this paper. We list them here and all of them are related to limitations in Model~\ref{model:our_model}. The questions arising out of a comparison of our model with the state-of-the art models are in Appendix~\ref{sec:model-open}.
\begin{itemize}
%\item {Handling node failures, which is a common occurence in practical distributed systems. We discuss this in detail in Appendix \yell{add ref}}
\item {Finding optimal assignments (instead of worst-case) of input functions to players in $G$ for a given query.}
\item {Given a query, identifying the optimal topology to compute it in our model.}
\item {Closing the gap between our upper and lower bounds for $d$-degenerate graphs for super-constant $d$. This has a known communication complexity bottleneck, which we discuss in Appendix~\ref{sec:open-clique}.}
\item {Extending our model (and results) to the case when $N$ is smaller than $G$. }
\item {Obtaining tight bounds for Steiner tree packing and multi-commodity flow for arbitrary $G$ for smaller values of $N$ in the $\congest$ model~\cite{ghaffari-thesis}.}
\item {Expanding the choice of free variables with new algorithms and techniques (details in Appendix~\ref{sec:free-vars}).}
\end{itemize}

%\subsection{Input Functions are Completely Assigned to Players in $G$}
%We discuss the assumptions on the knowledge of $q$ and $G$. We note that if $N$ is much larger than size of $G$, then all players can send their information (either about $\cH$ or their location in $G$) to one player who can then broadcast the common knowledge back to all players. Unfortunately, for smaller values of $N$, the state-of-the-art results in the $\congest$ model do not give tight bounds for Steiner tree packing and multi-commodity flow for arbitrary $G$ as we need~\cite{ghaffari-thesis}. We consider our work to provide further motivation to solve these two fundamental problems in the $\congest$ model
\section*{Acknowledgments}
We thank the anonymous reviewers of PODS`$19$ for their helpful comments. We are greatly indebted to Arkadev Chattopadhyay and David Zuckerman for their insights that led to the results in Section~\ref{SEC:MATRIX-CHAIN}. We thank Martin Grohe, Oliver Kennedy, Hung Ngo and Dan Suciu for helpful discussions. This work was supported by NSF grant CCF-$1717134$. 
\bibliography{main,atri}
\bibliographystyle{abbrv}
\appendix
\onecolumn
\section{Comparison with Relevant Models} \label{sec:comparison}

\subsection{Basic MPC model}
We formally define the MPC model used in~\cite{BKS13} and the model adopted by~\cite{gym} here, both in the language of our model (Model~\ref{model:our_model}). We consider the MPC with no replication, which is known as the basic MPC model in the literature.
\begin{model} [$\MPC(0)$] \label{model:mpc1}
We are given a query $q$ and its underlying hypergraph $\cH = (\cV, \cE)$ with input functions $f_e$ having at most $N$ non-zero values for every $e \in \cE$. We consider the network topology $G'$ with $p + k$ nodes, defined as follows.  
There are $k$ nodes, each assigned a function $f_e$ for every $e \in \cE$.
We call this set $K$. 
There are no edges between any pair of nodes in $K$.
All nodes in $K$ are directly connected by an edge to every node in a clique with $p$ nodes that are disjoint from $K$ (also a part of $G'$). 
Each node in $K$ has capacity $N$ and all the remaining nodes have capacity $L$. The capacity of a node bounds the number of bits it can receive in each round.
Given this setup, we would like to compute $\BCQ$ (and more generally an $\FAQ$) of $\cH$ on $G'$. 

To design a protocol for this computation, we can assume that every node in $G'$ has the knowledge of $q$ and $G'$. At the end of the protocol, a pre-determined player in $K$ knows the answer to $q$. 

Finally, given the above setup, our goal is to design protocols that minimize the number of rounds to 
compute $q$ on $G$. This model does not take into account the internal computation done by the $p + k$ nodes and assumes the nodes co-operatively compute the answer to $q$.
\end{model}
We now summarize the differences of this model from ours. 
\subsubsection{Differences from our Model} \label{sec:diff} 
\begin{itemize}
\item {$\MPC(0)$ assumes a specific choice of a network topology $G'$ as opposed to general topology $G$ in our model.}
\item {$\MPC(0)$ assumes a specific assignment of functions in $\cE$ to players in $G'$. Our upper bound techniques can handle any assignment but our lower bounds are for a specific class of assignments. We would like to mention here that this is true for the models in~\cite{BKS13,gym} as well (i.e., upper bounds can handle any assignment whereas lower bounds are for a specific class of assignments) and we consider one such assignment in Model~\ref{model:mpc1}.}
\item {$\MPC(0)$ assumes node capacities whereas ours assumes edge capacities.}
\item {The models in~\cite{BKS13,gym} design protocols wherein the number of rounds is either constant or a function of $k$. The number of rounds in our model are a function of $N$.}
\item {The models in~\cite{BKS13,gym} generally prove results for computing natural join whereas we look at $\BCQ$ (and more generally $\FAQ$s). We note that the results of~\cite{BKS13,gym} for natural join apply to $\BCQ$ as well. This is true for both upper and lower bounds in~\cite{BKS13} and only for upper bounds in~\cite{gym}.}
\end{itemize}

We consider two instantations of this model -- one by~\cite{BKS13} and the other by~\cite{gym}.
\subsubsection{Fixing $p$ and Determining $L$~\cite{BKS13}} \label{sec:mpc1}  This model assumes $N$ is larger than the size of $G'$ and all the functions $f_e$ are matchings (i.e., skew-free). 
In other words, for each variable $v \in e$, each of the values $x_v \in \Dom(v)$ can occur in at most one tuple in $f_e$. 
Using Proposition $3.2$ and Theorem $3.3$ in~\cite{BKS13}, it can be shown that there exists an optimal one round protocol to solve $\BCQ$ of any star $\cH$ on $G'$ with $L = \Omega \left(\frac{k \cdot N}{p}\right)$. Further, when $\cH$ is a forest, $\BCQ$ of $\cH$ on $G'$ takes $\Theta(\log(D'))$ rounds for the same $L$ (where $D'$ is the diameter of $\cH$). 
We would like to mention here that a follow up work~\cite{BKS14} handled input functions with specific types of skew and proved upper and lower bounds for the queries considered above. 
Since each node in the $p$-clique can have different capacities in this scenario, we do not discuss it further here.

\subsubsection{Fixing $L$ and Determining $p$~\cite{gym}} \label{sec:gym} This model assumes the size of $G'$ is much larger than $N$. Assuming $L = (k \cdot N)^{\frac{1}{\delta}}$ for a fixed small constant $\delta > 1$, we can use the Main Results $1$ and $2$ from~\cite{gym} and show that there exists a protocol to solve $\BCQ$ of any star $\cH$ in: $(1)$ $\tO(k)$ rounds with $p = (k \cdot N)^{2 - \frac{2}{\delta}}$ and $(2)$ $\tO(\log_{2}(k))$ rounds with $p = (k \cdot N)^{6 - \frac{2}{\delta}}$.

\paragraph{\nopunct} Before we instantiate our model for a comparison with the above models, we would like to state that while our model can handle the constraint where the size of $G'$ can be larger than $N$, our techniques cannot. Hence, we restrict our comparison to the model in Section~\ref{sec:mpc1}. We now instantiate our model (Model~\ref{model:our_model}) with $G'$ and assume that each edge in $G'$ has capacity
\begin{equation} \label{eq:ell}
L' = \frac{L}{k} = \frac{N}{p}.
\end{equation}
Note that this is a weaker version of Model~\ref{model:mpc1} since node capacities don't necessarily translate to equal edge capacities when the goal is to compute $q$ on $G'$. We take this route as it helps us make a fair comparison with Section~\ref{sec:mpc1}.

\subsubsection{Our Results in Model~\ref{model:mpc1}} \label{sec:comp1} We show how our upper bound techniques apply for solving $\BCQ$ of any star $\cH$ on $G'$. 
We can instantiate Corollary~\ref{cor:star-ub} with capacity $\tilde{\Theta}(1)$ to get $\tO\left( \min_{\Delta \in [|V(G')|]} \left(\frac{N}{\ST(G', K, \Delta)}  + \Delta\right) \right)$ rounds. 
We claim that 
\begin{align*}
\min_{\Delta \in [|V(G')|]} \left(\frac{N}{\ST(G', K, \Delta)}  + \Delta\right) = \tO\left(\frac{N}{p}\right).
\end{align*} 
To see this, we show such a Steiner tree packing containing $p$ trees with diameter $2$ -- each node in the $p$-clique in $G'$ along with all its $k$ edges incident on $K$ forms a Steiner tree. Since there are $p$ such nodes, we can obtain such a packing. 
Recall that when each edge in $G'$ has capacity $L'$ (instead of the $O(\log_2(D))$ capacity in Model~\ref{model:our_model}), our upper bound gets divided by $L'$.
Thus, we have an upper bound of $\tO\left(\frac{N}{L' \cdot p}\right) = O(1)$ (using~\eqref{eq:ell}) i.e., a constant number of rounds.

Note that a lower bound of one round on the number of rounds is trivial. Hence, we can obtain a tight bound of $\tilde{\Theta}(1)$ for any star $\cH$, resulting in an one round protocol matching results in Section~\ref{sec:mpc1}. 

Given the tight results for the star case, there is a natural generalization for our protocol and bounds when $\cH$ is a forest using ideas from the proof of Lemma~\ref{lemma:aritytwo-forest-ub}. We start by noting that all stars at the same level in $\cH$ can be computed simultaneously since each node in $K$ is directly connected to each node in the $(p)$-clique. In particular, we can run the star protocol used above on all these stars simultaneously but we still need to be able to uniquely identify the stars computed. It's not too hard to see that this can be done with $O(\log(y(\cH)))$ additional information for each internal node $v$. This results in an upper bound of 
\begin{align*}
O\left(D' \cdot \log(y(\cH)) \cdot \min_{\Delta \in [|V|]} \left( \frac{N}{\ST(G', K, \Delta)} + \Delta\right) \right) = O \left( \frac{D' \cdot \log(y(\cH)) \cdot N}{p} \right),
\end{align*}
where $\ST(G', K, 2) = p$ and $D'$ is the diameter of $\cH$.
If we divide our upper bound by $L'$ and substitute its value from~\eqref{eq:ell},
we can use ideas similar to those used in the star case to obtain a protocol with $O(D')$ rounds.
However, our lower bound techniques do not work for the assignment of functions to $K$ in Model~\ref{model:mpc1}.
We would like to mention that the model in Section~\ref{sec:mpc1} takes $\Theta(\log(D'))$ rounds for this case (though the upper bound only holds for the special case of matching databases).

Finally, for general simple graphs $\cH$, we decompose $\cH$ into a forest and a core by Definition~\ref{defn:forest-core}. We use the \emph{trivial protocol} on the core, which is basically sending all functions to one player in $K$ and is independent of the induced query in the core. We would like to mention that this is worse than existing protocols~\cite{BKS13,gym} for $\cH$ with non-constant degeneracy $d$ since we do not exploit any information about the query.

Before we move to the next model, we would like to mention here that the results of Section~\ref{sec:mpc1} and our results match up to a constant factor for the case when $\cH$ is a star. The upper bounds match since the protocols in both cases split the input functions the same way -- the model in Section~\ref{sec:mpc1} uses hashes to achieve this and we use Steiner tree packings for the same. The results however start diverging even when $\cH$ is a tree of small depth.

\subsection{General MPC model}

\paragraph{\nopunct} We now perform our second and final comparison. We formally define the model from~\cite{BKS16}, which is a followup of~\cite{BKS13,BKS14} and performs a \emph{worst-case analysis} of the communication cost for join queries. All the three models are described in~\cite{mpc}. We define it in the language of our model like we did for Model~\ref{model:mpc1}.
\begin{model} [$\MPC(\epsilon)$] \label{model:mpc2}
Let $\epsilon$ be a fixed value s.t. $0 \le \epsilon < 1$.
We are given a query $q$ and its underlying hypergraph $\cH = (\cV, \cE)$ with input functions $f_e$ having at most $N$ non-zero values for every $e \in \cE$. We consider the network topology $G''$, which is a clique on $p$.  
The input of size $k \cdot N$ is uniformly partitioned across the $p$ nodes. 
Let $K = V(G'')$. It follows that $|K| = p$.
All nodes in $G$ have capacity $L(\epsilon)$. The capacity of a vertex bounds the number of bits it can receive in each round.
Given this setup, we would like to compute $\BCQ$ (and more generally an $\FAQ$) of $\cH$ on $G''$. 
	
To design a protocol for this computation, we can assume that every node in $G''$ has the knowledge of $q$ and $G''$. At the end of the protocol, a pre-determined player in $K$ knows the answer to $q$. 
	
Finally, given the above setup, our goal is to design protocols that minimize the number of rounds to 
compute $q$ on $G''$. This model does not take into account the internal computation done by the $p$ nodes and assumes the nodes co-operatively compute the answer to $q$.
\end{model}

We now summarize the differences of this model from ours. 
\subsubsection{Differences from our Model} \label{sec:diff2} 
\begin{itemize}
\item {$\MPC(\epsilon)$ assumes a specific choice of a network topology $G''$ as opposed to general topology in our model.}
\item {$\MPC(\epsilon)$ assumes a uniform distribution of the input across the $p$ nodes instead of one function being completely assigned to a specific node in $G''$ in our model.}
\item {$\MPC(\epsilon)$ works with node capacities like Model~\ref{model:mpc1}, whereas ours works on edge capacities.}
\item {The model in~\cite{BKS16} designs protocols wherein the number of rounds either constant or a function of $k$. The number of rounds in our model are a function of $N$.}
\item {The model in~\cite{BKS16} proves results for computing natural join whereas we look at $\BCQ$ (and more generally $\FAQ$s). Note that their upper results for natural join apply for $\BCQ$ as well (but lower bound results do not transfer).}
\end{itemize}

We consider the instantiation of this model by~\cite{BKS16}. We would like to mention here that the models studied in~\cite{BKS13,BKS14,gym} can all be instantiated in this setting only for proving upper bounds.
\subsubsection{Fixing $p$ and Determining $L$~\cite{BKS16}} \label{sec:mpc2} This model assumes $N$ is larger than the size of $G'$ and there are no restrictions in the input functions. Using Theorems $3.1$ and $3.3$ of~\cite{BKS16}, it can be shown that there exists an optimal one round protocol to solve $\BCQ$ of any star $\cH$ on $G''$ with $L\left(\epsilon=1-\frac{1}{k}\right) = \Omega \left(\frac{N}{p^{1 - \epsilon}}\right) = \Omega \left(\frac{N}{p^{\frac{1}{k}}}\right)$. Further, when $\cH$ is a forest, $\BCQ$ of $\cH$ on $G''$ takes $O(k)$ rounds\footnote{\tbc{For the case when we are interested in computing the join query of $\cH$, then there is also a matching $\Omega(k)$ lower bound.}} with $L\left(\epsilon=1-\frac{1}{\rho^{*}(\cH)}\right) = \Omega \left(\frac{N}{p^{1 - \epsilon}}\right) = \Omega \left(\frac{N}{p^{\frac{1}{\rho^{*}(\cH)}}}\right)$ using ideas in Section $4$ of~\cite{BKS16}. Here, $\rho^{*}(\cH)$ denotes the edge cover number of $\cH$ (i.e., size of the minimum edge cover of $\cH$).

We now instantiate our model (Model~\ref{model:our_model}) with $G''$ and assume that each edge in $G''$ has capacity
\begin{equation} \label{eq:ell2}
L'' = \frac{L(\epsilon)}{p}.
\end{equation}
Further, we assume that the input functions are not distributed uniformly but rather based on some pre-determined hash functions. Note that this certainly makes our model (Model~\ref{model:mpc1}) more restrictive since node capacities don't necessarily translate to equal edge capacities when the goal is to compute $q$ on $G''$ and the hash-based split (see Appendix~\ref{app:hash-based-split}) restricts the way in which input functions can be distributed across nodes in $G''$. We opt for this since it helps us make a fair comparison with Section~\ref{sec:mpc2}.

\subsubsection{Our Results in Model~\ref{model:mpc2}} \label{sec:comp2} We now show how our upper bound techniques apply in this model for solving $\BCQ$ of any star $\cH$ on $G'$.  We do not compare lower bounds here since $(1)$~\cite{BKS16} lower bounds do not hold for $\BCQ$ (or at least it does not follow immediately from their lower bounds for the join queries) and $(2)$ Our lower bounds for the case when the functions are uniformly distributed over the players are quantitatively very weak.
For the upper bound, we can instantiate Corollary~\ref{cor:garity-star-ub} in with capacity $\tilde{\Theta}(1)$ to get $\tO\left( \min_{\Delta \in [p]} \left(\frac{N}{\ST(G'', K, \Delta)}  + p \cdot \Delta\right) \right)$ rounds.
We claim that $\min_{\Delta \in [p]} \left(\frac{N}{\ST(G'', K, \Delta)}  + p \cdot \Delta\right) = \tO\left(\frac{N}{p} + p\right)$. 
To see this, we show a Steiner tree packing containing $\frac{p - 1}{2}$ trees with diameter $1$ -- 
we can greedily keep picking and throwing out paths of length $p - 1$ from $G''$ that contain all the $p$ vertices. Each such path forms a Steiner tree. 
Since we can identify $\frac{p - 1}{2}$ such paths, we can obtain such a packing. 
Recall that when each edge in $G''$ has capacity $L(\epsilon)$ instead of the standard $O(\log_2(N))$, our upper bound gets divided by $L''$.
Thus, we have an upper bound of $\tO \left(\frac{\frac{N}{p} + p} {L''}\right)$. Using~\eqref{eq:ell2} and the fact $N \ge p^{2}$ (from Model~\ref{model:our_model}), we get a $\tO \left( p^{\frac{1}{k}} \right)$ round protocol. Note that this is worse than the one round protocol by Section~\ref{sec:mpc2}.

For the case when $\cH$ is a forest, we can instantiate Lemma~\ref{lemma:garity-acyclic-ub} with capacity $\tilde{\Theta}(1)$ to get a bound of 
\begin{align*}
\tO\left( y(\cH) \cdot \min_{\Delta \in [|V|]}\left(\frac{N \cdot r}{\ST(G'', K, \Delta)} + p\cdot\Delta \right)\right)  = \tO\left(\frac{D' \cdot \log(y(\cH)) \cdot N}{p} \right), 
\end{align*}
where $\ST(G'', K, 1) = p$ and $D'$ is the diameter of $\cH$.
We can use ideas from Section~\ref{sec:comp1} and from those used in the star case to obtain a protocol with $\tO \left(D' \cdot p^{\frac{1}{\rho*(\cH)}} \right)$ rounds. In particular, to get this bound, we divide our upper bound by $L''$ and substitute its value form~\eqref{eq:ell2}. Note that this is worse than the $O(k)$ round protocol by Section~\ref{sec:mpc2}.

Finally, for general simple graphs $\cH$, we decompose $\cH$ into a forest and a core by Definition~\ref{defn:forest-core}. We use the \emph{trivial protocol} on the core, which is basically sending all functions to one player in $K$ and is independent of the induced query in the core. As stated in Section~\ref{sec:comp1}, this is worse than existing protocols~\cite{BKS16} for $\cH$ with non-constant degeneracy $d$ since we do not exploit any information about the query.

\subsection{Scope for Future Work} \label{sec:model-open}
Many open questions arise out of this comparison. We summarize them here and leave them for future work.
\begin{itemize} 
	\item {Can we modify our model to handle node failures like Models~\ref{model:mpc1} and~\ref{model:mpc2} do, using replication?}
	\item {Can we improve over our \emph{trivial protocol} for cyclic queries using ideas from~\cite{BKS13,BKS14,BKS16,gym,JR16,mpc}?}
	\item {Can our algorithmic ideas for set intersection be plugged into the Models~\ref{model:mpc1} and~\ref{model:mpc2}?}
	\item {Can we extend our techniques to handle arbitrary distributions of input functions to nodes in the topology?}
\end{itemize}

\subsection{Connection to Sensor Networks} \label{sec:sensor}
Sensor networks are typically tree-like topologies, where the goal is to efficiently and accurately report aggregate queries on data generated by the sensors. Since the sensors can traditionally store only little data, they stream their data (as they generate them) to designated points in the topology called storage points. There is a server that has more computational power and initiates these queries, collects the query answers, reports them and so on. Join/Aggregate queries are computed either between the storage points or between the server and a storage point~\cite{TinyDB}.

We now restate this setting in our language. The server and the storage points are the nodes in $G$ and the edges are defined based on the sensor network. The query to be computed on $G$ is a $\FAQ$ $q$ (Joins/Aggregates are a special cases of $\FAQ$), whose underlying query hypergraph is $\cH$. The input functions in $\cH$ are assigned to a subset of nodes $K$ in $G$. The upper and lower bounds that we obtain for computing $q$ on $G$ assuming all input functions have size at most $N$ apply for this setting in sensor networks. Further, in our setup, we can make any pre-determined node in $K$ (say the server) know the answer to $q$. In particular, this implies our model captures query computation in Sensor Networks for a specific class of queries. 

Due to the theoretical nature of our results, the potential applications of our model/results in the IoT setting are somewhat speculative. We hope that our work motivates more study of our general model in these applications areas.

\subsection{Which Distributed Computing Model to Use?} \label{sec:usage}
We believe that different models could be used for different settings. For instance, if all the nodes are interconnected to each other in a compute farm (i.e., $G$ is a clique) and each node can receive only a certain amount of data in a communication round and we are interested in computing the join query corresponding to $\cH$, then the MPC-based models~\ref{model:mpc1} and~\ref{model:mpc2} are more suitable. On the other hand, if we are looking at more general topologies $G$, the capacities are on the edges and we are interested in computing $\BCQ$ of $\cH$, then using our model might make more sense. 

\section{The Clique Open Problem} \label{sec:open-clique}
Consider the case where $\cH$ is a $k$-clique with all input functions having size at most $N$ and $G$ is an edge $e = (a, b)$. The goal is to compute $\BCQ$ of $\cH$ on $G$ assuming worst-case assignment of functions in $\cH$ to players in $G$.

We can prove an upper bound of $O(k^{2} \cdot N)$ as follows. Consider an assignment where half of the functions (i.e., $\frac{k \cdot (k - 1)}{4}$ of them) are assigned to $a$ and the other half of them is assigned to $b$. In particular, $a$ can send all its functions to $b$ to compute the $\BCQ$ of $\cH$ on $G$, the upper bound of $O(k^{2} \cdot N)$ follows. Since we consider worst-case assignments, we can't prove a better upper bound.

The best lower bound known so far for this query is $\Omega(k \cdot N)$, which is worse than the upper bound by a factor of $O(k)$. Going beyond this bound seems beyond the reach of current two-party communication complexity techniques~\cite{arkadev}. We believe that our work will provide more motivation to solve this outstanding open question in two-party computational complexity.
\section{Missing Details in Section~\ref{sec:overview}}
\subsection{GYO-GHD is a reduced GHD} \label{app:gyo-correctness}
The correctness of Construction~\ref{lemma:forest-core} follows from the facts that the \emph{GYO-reduction} of any $\cH$ is unique~\cite{gyo-4} and the hyperedges removed while running GYOA form an acyclic forest (Lemma $4.8$ in~\cite{gyo-5}). We define $\Fo(\cH)$ as the union of all vertices in all hyperedges in the acyclic forest excluding the roots (as they are included in $\Co(\cH)$). To complete our construction, we need to argue that $\cT$ is a reduced-GHD. This follows from our construction i.e., edge $e \in \cE$ satisfies either $e \subseteq V(\Co(\cH))$ (or) $e \subseteq V(\Fo(\cH))$ and in both these cases, there always exists a node $v$ in $\cT$ such that $\chi(v) = V(e)$. We argued this already for $\Co(\cH)$ and for $\Fo(\cH)$, this follows from the definition of acyclicity.

\subsection{Example for Construction~\ref{lemma:forest-core}} \label{sec:ex-core}
Consider a hypergraph $\cH_{3}$ with nodes $\cV(\cH_{3}) = \{A, B, C, D, E, F, G, H\}$ and hyperedges \\ \[\cE(\cH_{3}) = \{e_1 = (A, B, C), e_2 = (B, C, D), e_3 = (A, C, D), e_4 = (A, B, E), e_5 = (A, F), e_6 = (B, G), e_7 = (G, H)\}.\] 
We now apply the GYO algorithm (GYOA)~\cite{gyo-1,gyo-2,gyo-3} on $\cH$, which basically keeps performing the following two steps until it cannot. 
First, it checks if there is a node that is present in one hyperedge and if so, eliminates it. Second, it deletes a hyperedge that is contained in another. 
We document the execution of GYOA on $\cH$ here. Let $\cE'(\cH_{3}) = \cE(\cH_3)$. 
\begin{itemize}
	\item {Choose $H$ as it is present in only one hyperedge $(G, H)$. Remove it and the reduced hypergraph now is $\cE'(\cH{3}) = \{e_1, e_2, e_3, e_4, e_5, e_6, (G)\}$. Since the edge $(G)$ is subsumed by more than one hyperedge we can remove it from $\cE'(\cH_3)$.}
	\item {Choose $G$ as it is present in only one hyperedge $(B, G)$. Remove it and the reduced hypergraph now is $\cE'(\cH{3}) = \{e_1, e_2, e_3, e_4, e_5, (B)\}$. Since the edge $(B)$ is subsumed by more than one hyperedge we can remove it from $\cE'(\cH_3)$.}
	\item {Choose $F$ as it is present in only one hyperedge $(A, F)$. Remove it and the reduced hypergraph now is $\cE'(\cH{3}) = \{e_1, e_2, e_3, e_4,(A)\}$. Since the edge $(A)$ is subsumed by more than one hyperedge we can remove it from $\cE'(\cH_3)$.}
	\item {Choose $E$ as it is present in only one hyperedge $(A, F)$. Remove it and the reduced hypergraph now is $\cE'(\cH{3}) = \{e_1, e_2, e_3,(A, B)\}$. Since the edge $(A, B)$ is subsumed by more than one hyperedge we can remove it from $\cE'(\cH_3)$.}
\end{itemize}
The GYOA terminates after the final step since it cannot find a variable that is contained in only one hyperedge. 
Let $\cT$ be the GYO-GHD obtained from this procedure. 
The final edge set $\cE'$ returned by GYOA is $e_c = \{e_1, e_2, e_3\}$ and the acyclic forest removed in this process contains the edges $e_{f} = \{e_4, e_5, e_6, e_7\}$ and is rooted at $e_4$. 
The forest $\Fo(\cH_{3})$ is the union of all vertices in the set $e_{f} \setminus e_{4}$. 
We build the core $\Co(\cH_3)$ now with vertices that are union of edges in $e_c$ and $e_4$ (i.e., the root of the acyclic forest). 
$\cT$ is rooted at $r'$ with $\chi(r') = \cup_{i \in [3]} v(e_i) \cup v(e_4)$. 
We create new nodes $v'_{e_i}$ for every $i \in [4]$ and all of them are directly connected to $r'$ (i.e., edge $(r', v'_{e_i})$ is added to $E(\cT)$). 
Thus, $\cT$ contains the nodes $r', v'_{e_1}, v'_{e_2}, v'_{e_3}, v'_{e_4},  e_5, e_6, e_7$. 
Since we do not enforce constraints on the remaining edges in $\cT$ as long as it is a valid GHD, we show two sample GYO-GHDs that can be constructed out of this. The first has edge set $E(\cT) \cup \{(r', e_5), (r', e_6), (e_6, e_7)\}$ (having two internal nodes) and the second has edge set $E(\cT) \cup \{(e_4, e_5), (e_5, e_6), (e_6, e_7)\}$ (having three internal nodes). It's not too hard to see that both these are reduced-GHDs by Definition~\ref{Definition:ghd}.
\section{Missing Details in Section~\ref{sec:notation}}

\subsection{Connecting $\tau_{\MCF}$ and $\MC(G,K)$} \label{app:mcf-mc}
In this section, we show that under worst-case assignment of relations to players, the bounds $\tau_{\MCF}(G,K,N')$ and $\frac{N'}{\MC(G,K)}$ are within an $\tilde{O}(1)$ factor of each other. Then, we argue that our upper and lower bounds are tight within a $\tilde{O}(d)$ factor for a larger class of assignments.

Let $(A,B)$ be a cut that separates $K$ of size $\MC(G,K)$.
First, consider the assignment where half of the relations are assigned to one player $a$ in $A$ and the rest to another player $b$ in $B$. Note that in this case, $\tau_{\MCF}(G,K,N')$ is upper bounded by number of rounds needed to send $N'$ bits from (say) $a$ to $b$. By the max-flow-min-cut theorem, we know that we can send $N'$ bits from $a$ to $b$ in $\frac{N'}{\MC(G,K)}+d(a,b)$ rounds, where $d(a,b)$ is the distance between $a$ and $b$.

We now somewhat extend the class of assignments so that our upper and lower bounds are still within a factor $\tilde{O}(d)$ of each other. Let $(A,B)$ be the cut as above.
Now, let's assume we distribute the relations that embed the $\TRB$ instance so that the $m=\frac{n_2}{2\cdot \log(n_2)}$ pairs $(S_i,T_i)$ of $\TRB_{m,N}(\hat{S},\hat{T})$, the $S_i$ are assigned to some players in $A$ and the $T_i$'s to players in $B$. The remaining relations are divided equally among $A$ and $B$. Note that our lower bound still holds.
 
For the upper bound, we have to look at the multi-commodity flow that needs to send $\tO(n_2 \cdot d \cdot N)$ bits of flow from all but one player to a designated player (who is assigned at least one of the $(S_i,T_i)$ in the hard instance for the lower bound). Each of the at most $n_2 \cdot d$ relations denote one `demand' of size $N$. The sparsity $S$ of the cut $(A, B)$ is defined as the ratio of the number of cut edges and the size of the maximum demand separated by cut. We have $S \ge \frac{\MC(G, K)}{n_2 \cdot d \cdot N}$ since $\MC(G, K)$ is the smallest cut and the maximum demand separated by any cut is at most $n_2 \cdot d \cdot N$. Using the celebrated result of Leighton and Rao~\cite{LR}, one can schedule this multi-commodity flow in $\tilde{O}\left(\frac{n_2 \cdot d \cdot N}{\MC(G,K)}+\Delta(G,K)\right)$ rounds, where $\Delta(G,K)$ is the largest distance between any two players in $K$.

\begin{table*}[!htbp]
	{\small
		\centering
		{\renewcommand{\arraystretch}{1.5}
			\begin{tabular}{|l|l|}
				\hline
				Notation & Meaning\\
				\hline
				$q$  & Join query $\{R_i, A_i\}_{i \in [k]}$\\ 
				\hline
				$R_i/R_e/R_{v(e)}$ & Function/Relation\\
				\hline
				$r$ & Upper bound on arity of $\{R_i\}_{i \in [k]}$\\
				\hline
				$A_i$ & Attribute set of relation $R_i$\\
				\hline
				$A(q)$ & All attributes of $q$\\
				\hline
				$n$ & Size of $A(q)$\\
				\hline
				$k$ & Size of $q$ (number of relations)\\
				\hline
				$N$ & Upper bound on size (i.e., number of tuples) of $R_i$ \\
				\hline
				$\cH = (\cV, \cE)$ & Underlying (multi)-hypergraph of $q$\\
				\hline
				$\Fo(\cH),\Co(\cH)$ & decomposition of $\cH$ into forest $\Fo(\cH)$ and core $\Co(\cH)$ using Construction~\ref{lemma:forest-core}.\\
				\hline 
				$y(\cH)$ & $y(\cH)  = \min_{\forall \cT : \cT\text{ is a GYO-GHD of } \cH} y(\cT)$. \\
				\hline
				$n_2(\cH)$ & $|V(\Co(\mathcal{H}))|$ \\
				\hline
				$ \langle \cT,\chi,\lambda \rangle$ & Generalized hypertree decompositions (GHD) of $\cH$\\
				\hline
				$\chi(v) \subseteq \cV$ & Subset of vertices of $\cV$ associated to each node $v \in V(\cT)$\\
				\hline
				$\lambda(v) \subseteq \cV$ & Subset of hyperedges of $\cE$ associated to each node $v \in V(\cT)$\\
				\hline
				$G = (V, E)$ & Communication graph\\
				\hline
				$K$ & At most $k$ terminal nodes in $G$ \\
				\hline
				%$\cR$ &  Minimum worst-case number of rounds needed by a protocol to deterministically compute $q$. \\
				%\hline
				%$\cD$ & Deterministic round complexity\\
				%\hline
				$\tau_{\MCF}(G, \cK, N')$ & Round complexity of routing $N' \log_{2}(N')$ bits from all players in $\cK$ to any one player in $\cK$. \\
				\hline
				$\ST(G,\cK,\Delta)$ & $\Delta$ diameter Steiner tree packing\\
				\hline
			\end{tabular}
		}
		\caption{Notations used in the paper.}
		\label{tab:notation}
	}
\end{table*}

\section{Missing Details in Section $4$} \label{app:aritytwo}
\subsection{Proof of Lemma~\ref{lemma:aritytwo-forest-ub}} \label{app:atwo-forest-ub}
\begin{proof}
For simplicity, we assume that $\mathcal{H}$ has only one tree with $y=y(\cH)$ internal nodes. Next we show that we can solve the $\BCQ$ problem on $\cH$ by solving another $\BCQ$ problem on $\cH'$ with $y - 1$ internal nodes defined as follows. We remove the bottom most star $P = (v_1, \dots, v_{|P|})$ (where $v_1$ is the center and $(v_2, \dots, v_{|P|})$ are the leaves) from $\mathcal{H}$. We define $V(\mathcal{H}') = V(\mathcal{H}) \setminus (v_2, \dots, v_{|P|})$ and $E(\mathcal{H}') = E(\mathcal{H}) \cup \{(v_1)\} \setminus \{(v_1, v_i) : i \in [2, |P|]\}$. Using arguments in Section~\ref{sec:star-ub-1}, we can process $P$ in $\tO\left(\min_{\Delta \in [|V|]}\left(\frac{N}{\ST(G, K, \Delta)} + \Delta \right)\right)$ rounds and the result computed is $R'_{P}= \bigcap_{i = 2}^{k} R'_{v_i}$, where $R'_{v_i} = \pi_{v_1}(R_{v_1, v_i})$. Finally, we set $R'_{v_1}=R'_P$ (while the remaining surviving relations remains the same). It is easy to see that $\BCQ$ on $\cH$ is $1$ iff $\BCQ$ on $\cH'$ is $1$.

Note that $\cH'$ is also a tree, which implies we can continue this process recursively until $\cH'$ has only one node left. Thus, the final answer is given by $(R'_{P} \stackrel{?}{\neq} \emptyset)$ and the number of recursive calls is bounded by the number of internal nodes $y$. Further, if $\mathcal{H}$ is a forest, our argument can be applied individually on each tree, resulting in the upper bound~\eqref{eq:aritytwo-forest-ub}. This completes the proof.
\end{proof}

\subsection{Proof of Lemma~\ref{lemma:aritytwo3}} \label{app:atwo-ub}
\begin{proof}
We start by considering $\Fo(\cH)$ (via Construction~\ref{lemma:forest-core}). Using the protocol in the Proof of Lemma~\ref{lemma:aritytwo-forest-ub} (stated above), we know that $\tO\left(y(\cH) \cdot \min_{\Delta \in [|V|]}\left(\frac{N}{\ST(G, K, \Delta)} + \Delta \right) \right)$ rounds suffice to reduce $\cH$ to an updated hypergraph $\cH'=(V(\Co(\cH)),E(\cH)\cup \{(r')|r' \text{ is a root  in }\Fo(\cH))$. Further, for each root $r'$ in $\Fo(\cH)$, the corresponding relation $R_{r'}$ is the set computed by the algorithm in the proof of Lemma~\ref{lemma:aritytwo-forest-ub}. It is easy to check that $\BCQ$ on $\cH$ has the same answer as $\BCQ$ on $\cH'$.

We can now use the \emph{trivial protocol} to solve $\BCQ_{\cH', N}$ on $G$, which by Lemma~\ref{lem:naive} gives the upper bound of~\eqref{eq:aritytwoubgeneral}, completing the proof.
\end{proof}

\subsection{Proof of Theorem~\ref{thm:aritytwo2}} \label{app:atwo-lb}
We start by stating some standard results that we use in our proof and then prove our general lower bound.
\paragraph{Existing Results: } We state two standard graph theory results that we will use in our lower bound arguments. 
\begin{lemma} [Moore's Bound~\cite{moore}] \label{lemma:moorebound1}
Every graph with $p > 2 |V|$ edges has a cycle of length at most $\frac{2 \cdot \log(|V|)}{\log\left(\frac{p}{|V|} - 1\right)}$.
\end{lemma}
\begin{thm} [Turan's Theorem~\cite{prob-method}] \label{theorem:turan1}
If a graph $\mathcal{H}$ has $n'$ vertices and at most $n' \cdot d$ edges, then there always exists an independent set of size at least $\frac{n'}{d + 1}$ in $\mathcal{H}$.
\end{thm}
\paragraph{Our Results.} We are now ready to prove our general lower bound for all simple graphs $\mathcal{H}$.
\begin{proof}
For notational convenience, define $y=y(\cH)$ and $n_2=n_2(\cH)$.
Let $m = \max\left(\frac{y}{2}, \frac{n_2}{2\log(n_2)}\right)$.
In general, as in Lemma~\ref{lemma:lower-forest}, given $\mathcal{H}$ and a TRIBES instance $\TRB_{m, N}(\hat{S},\hat{T})$ we construct a $\BCQ$ instance $q_{\cH,\hat{S},\hat{T}}$ such that $q_{\cH,\hat{S},\hat{T}}=1$ iff $\TRB_{m, N}(\hat{S},\hat{T})=1$. 
To this end we need to ``embed'' the $m$ pairs of sets $(S_i,T_i)$ from $\TRB_{m, N}(\hat{S},\hat{T})$ as relations in $q_{\cH,\hat{S},\hat{T}}$. 
For $m = \frac{y}{2}$, we embed the pairs $(S_i, T_i)$ in the forest $\Fo(\cH)$ as done in Lemma~\ref{lemma:lower-forest}. 
For $m = \frac{n_2}{2 \cdot \log(n_2)}$, we consider $\Co(\cH)$. 
We then show that it must be the case that $\Co(\mathcal{H})$ either includes $\left(\frac{n_2}{2\log(n_2)}\right)$
vertex-disjoint cycles (referred to as {\em Case 1}), or that it has an independent set of size $\Omega(n_2)$ (referred to as {\em Case 2}). 
In both cases, we show how one can embed $\frac{n_2}{2\log(n_2)}$ pairs $(S_i,T_i)$ of $\TRB_{m, N}(\hat{S},\hat{T})$ in $C(\mathcal{H})$. In particular, we prove the following lemma.
\begin{lemma}
$\Co(\mathcal{H})$ either includes $\left(\frac{n_2}{2\log(n_2)}\right)$ vertex-disjoint cycles (Case 1) or it has an independent set of size $\Omega(n_2)$ (Case 2).
\end{lemma}
\begin{proof}
By definition, the average degree of $C(\mathcal{H})$ is at least two (because if there is a vertex in $\Co(H)$ of degree at most one, then it should be part of $\Fo(\cH)$, which would contradict Construction~\ref{lemma:forest-core}). As long as the average degree is greater than $10$, we can use Lemma~\ref{lemma:moorebound1} to prove that there exists a cycle in $\Co(\mathcal{H})$ of length at most $\log(n_2)$. We can remove this cycle from $\Co(\mathcal{H})$ and recurse until the average degree is below $10$. Let $w$ be the number of vertex-disjoint cycles we have collected. If $w \geq \left(\frac{n_2}{2\log(n_2)}\right)$ we are in Case 1. Otherwise, at some point we are left with an induced subgraph of $\Co(\mathcal{H})$ of size at least $\frac{n_2}{2}$ and average degree at most $10$. In this case, by Theorem~\ref{theorem:turan1}, we can find an independent set in the induced subgraph (and thus in $\Co(\mathcal{H})$) of size at least $\Omega(n_2)$, which is Case 2.
\end{proof}
We now show separately for each case how to embed $\frac{n_2}{2\log(n_2)}$ pairs $(S_i,T_i)$ of $\TRB_{m, N}$ in $\Co(\mathcal{H})$. 
We start with Case 2.
In this case, for large enough $n_2$, $\Co(\mathcal{H})$ has an independent set of size at least $\frac{n_2}{2\log(n_2)}$ consisting of nodes of degree at least two.
We can thus use a proof identical to that given in Lemma~\ref{lemma:lower-forest} to construct the remaining relations of $\BCQ_{\cH, N}$ corresponding to $\Co(\mathcal{H})$. Namely, the independent set of $\Co(\mathcal{H})$ will play the role of the set $O$ in Lemma~\ref{lemma:lower-forest}.

We now address Case 1.
Consider a cycle $C$ in $C(\mathcal{H})$ and a pair of sets $(S_i,T_i)$ from $\TRB_{m, N}$. Let $c_1,c_2,\dots,c_\ell$ be the nodes in $C$. To embed  $(S_i,T_i)$ in $C$ we first present $S_i$ not as a subset of $[N]$ but rather as a subset of $[\sqrt{N}]\times[\sqrt{N}]$ or alternatively as a relation $R_{S_i}$ over two attributes with domain $[\sqrt{N}]$. Similarly, we associate $T_i$ with a relation $R_{T_i}$ over two attributes with domain $[\sqrt{N}]$. 
We define the relation corresponding to edge $(c_1,c_2)$ in the cycle as $R_{S_i}$, the relation corresponding to edge $(c_3,c_2)$ as $R_{T_i}$ (note that we reverse the order of attributes for $R_{T_i}$), and the relations corresponding to the remaining cycle edges as $\{(i, i) : i \in [\sqrt{N}]\}$. Notice, with this assignment of relations to the edges it holds that $\SD(S_i,T_i)=1$ iff there is an assignment that satisfies all relations in the cycle. Indeed, if tuple $\mathbf{t}$ satisfies all relations on the cycle, then 
the pair $(\pi_{c_1}(\mathbf{t}),\pi_{c_2}(\mathbf{t}))$ is in $S_i$, the pair $(\pi_{c_3}(\mathbf{t}),\pi_{c_2}(\mathbf{t}))$ is in $T_i$, and it holds that $\pi_{c_3}(\mathbf{t})=\pi_{c_{4}}(\mathbf{t})=\dots=\pi_{c_\ell}(\mathbf{t})=\pi_{c_{1}}(\mathbf{t})$. Thus, we conclude that the pair $(\pi_{c_1}(\mathbf{t}),\pi_{c_2}(\mathbf{t}))$ is in $S_i$ and the pair $(\pi_{c_2}(\mathbf{t}),\pi_{c_1}(\mathbf{t}))$ is in $T_i$, which in turn implies that $\SD(S_i,T_i)=1$.
Alternatively, if $\SD(S_i,T_i)=1$ then there exists a pair $(\alpha,\beta)$ such that $(\alpha,\beta) \in S_i \cap T_i$. We can now set $\mathbf{t}$ with $\pi_{c_3}(\mathbf{t})=\pi_{c_{4}}(\mathbf{t})=\dots=\pi_{c_\ell}(\mathbf{t})=\pi_{c_{1}}(\mathbf{t})=\alpha$ and $\pi_{c_2}(\mathbf{t})=\beta$ to satisfy all relations corresponding to the cycle.

We continue in a similar manner for each cycle $C$ in our collection of cycles. Namely, for each cycle, we define relations corresponding to a pair of sets from $\TRB_{m, N}(\hat{S},\hat{T})$ such that the sets intersect iff there is an assignment that satisfies the relations in $C$. Notice that the collections of pairs of sets $(S,T)$ corresponding to the cycle collection have pair-wise intersection iff there is an assignment that satisfies all the relations in the cycle collection. To complete the definition of  $q_{\cH,\hat{S},\hat{T}}$ we still need to assign a relation to all edges in $\Co(\mathcal{H})$ that do not appear in any of the cycles in the collection. All such edges are assigned the complete relation $[\sqrt{N}] \times [\sqrt{N}]$ over 2 attributes of domain $[\sqrt{N}]$. Note that the complete relations assigned do not impose any restrictions on the possible tuples $\mathbf{t}$ that satisfy the relations corresponding to the collection of cycles, and thus we have successfully embedded $\frac{n_2}{2\log(n_2)}$ pairs $(S_i,T_i)$ of $\TRB_{m, N}(\hat{S},\hat{T})$ in $\Co(\mathcal{H})$.

Thus, we can conclude that $q_{\cH,\hat{S},\hat{T}}=1$ iff $\TRB_{m, N}(\hat{S},\hat{T})=1$, where $m = \max\left(\frac{y}{2}, \frac{n_2}{2\log(n_2)}\right)$. Since sum and max are within a factor $2$ of each other, we can write $m \ge \frac{y}{4} +  \frac{n_2}{4\log(n_2)}$. We can now apply ideas from the proof of Lemma~\ref{lemma:tree-lb-1} to obtain the required lower bound $\tOm\left( \frac{(y+n_2) \cdot N}{\MC(G, K)}\right)$, as desired.
\end{proof}
\section{Queries when $\mathcal{H}$ is a $d$-degenerate hypergraph of arity at most $r$} \label{app:gen-arity}
In this section, we consider $\BCQ$s whose underlying hypergraph $\mathcal{H}$ is $d$-degenerate with arity at most $r$. We prove upper and lower bounds that are tight within a factor of $\tilde{O}(d^{2} \cdot r^{2})$ for computing any query in $\BCQ_{\cH, N}$. We assume $N \ge |V|^2$ and consider worst-case assignment of input functions to players in $G$.

\subsection{Main Theorem}
We state our main theorem here.
\begin{thm} \label{thm:garitymain}
For arbitrary $G$, subset of players $K$ and $d$-degenerate hypergraphs with arity at most $r$, we have
\begin{equation} \label{eq:garityub}
\cD (\BCQ_{\cH, N}, G, K) = \tO \left( y(\cH) \cdot \min_{\Delta \in [|V|]}\left(\frac{N \cdot r}{\ST(G, K, \Delta)} + \Delta \right) + \tau_{\MCF}(G, K, n_2(\cH) \cdot d \cdot r \cdot N)\right).
\end{equation}
Further, for $d$-degenerate hypergraphs $\mathcal{H}$, we have 
\begin{equation} \label{eq:garitylb}
\cR(\BCQ_{\cH, N}, G, K) \ge \tOm \left(\frac{\frac{y(\cH) \cdot N}{r} +  \frac{n_2(\cH) \cdot N}{d r}}{\MC(G, K)} \right).
\end{equation}
\end{thm}
We would like to point out that as in the arity two case, our upper bound holds for every assignment of the functions $f_{e}$ to players in $K$ and our lower bound holds for a specific assignment of functions of players in $K$. We prove this theorem in two steps. We first prove the upper bound~\eqref{eq:garityub}, followed by the lower bound~\eqref{eq:garitylb}. Finally, we argue that our bounds are tight within a gap of $\tilde{O}(d^{2} \cdot r^{2})$.

\subsection{Upper Bound} \label{sec:garity-ub}
Our proof is similar in nature to the proof for the arity two case. For the rest of this section, unless specified otherwise, let $\cT$ be a GHD of $\cH$ with the root bag being $\Co(\cH)$. Recall from Definition~\ref{defn:mdw} that $y(\cT)$ is the number of internal nodes in $\cT$. In this section, we will prove our upper bounds in terms of $y(\cT)$. Since we do not assume anything about $\cT$ beyond the fact that it has $\Co(\cH)$ as its root, our bound holds for the smallest $y(\cT)$ over all such GHDs $\cT$. This by Definition~\ref{defn:mdw} is {\em exactly} $y(\cH)$. We prove our bounds in terms of $y(\cT)$ since it makes the exposition simpler.

We compute each query in $\BCQ_{\cH, N}$ on GHD $\cT$. We first consider the case when $\cT$ is a star, which will be a basic building block for our algorithms for more general $\cH$.
 
\subsubsection{$\cT$ is a star} \label{sec:garity-star-ub}
Let $\cH$ be an acyclic hypergraph whose GHD $\cT$ is a star of the form $P = (v_1, \dots, v_k)$ with $v_1$ as the center. By Definition~\ref{defn:acyclic}, $\cH$ includes $k$ relations of the form $R_{\chi(v_i)}$ for every $i \in [k]$. Note that computing the corresponding BCQ query $\BCQ_{\cH, N}$ can be solved via a \emph{set-intersection problem} of computing $R'_{P} = \cap_{i = 2}^{k} R'_{v_i}$, where $R'_{v_i} = \{\mathbf{t} \in R_{\chi(v_1)} : \exists \mathbf{t}' \in R_{\chi(v_i)} \text{ s.t. } \pi_{\chi(v_1) \cap \chi(v_i)} \left(\mathbf{t} \right) =  \pi_{\chi(v_i) \cap \chi(v_1)}\left( \mathbf{t}' \right)\}$. It is easy to see that the final output of the $\BCQ_{\cH, N}$ instance is $1$ if $R'_{P} \neq \emptyset$ and $0$ otherwise. 
We would like note here that the set intersection here is computed on the attribute set $\chi(v_1)$ (each entry in the sets is a $r$-dimensional vector) as opposed to a single attribute (as was the case for arity two)

We describe our algorithm (Algorithm~\ref{algo:starAlgo-garity}) here. We first broadcast the relation $R_{\chi(v_1)}$ to all the remaining players containing relations $R_{\chi(v_i)}$ for every $i \in [2, k]$ in $G$. Then, each player containing $R_{\chi(v_i)}$ for every $i \in [2, k]$ computes $R'_{v_i}$. Finally, $R'_{P} = \bigcap_{i = 2}^{k} R'_{v_i}$ is computed using known upper bounds on set intersection using Theorem~\ref{thm:set-intersection}. Using the fact that at most $\tO\left(r \cdot \log_{2}(D)\right)$ bits are communicated in each round, we have the following result.
\begin{cor} \label{cor:garity-star-ub}
For acyclic hypergraphs $\cH$ where $\cT$ is a star, arbitrary graphs $G$ and subset of players $K$, we have
\[ \cD(\BCQ_{\cH, N}, G, K) = \tO\left(\min_{\Delta \in [|V|]}\left( \frac{r \cdot N}{\ST(G, K, \Delta)} + \Delta \right)\right).\]
\end{cor}
For the case when $G$ is a line with $k$ vertices, note that $\ST(G, K, \Delta) = 0$ for every $\Delta > k - 1$ and $\ST(G, K, k - 1) = 1$, which in turn implies the following.
\begin{cor}
Let $\cH$ be acyclic, $\cT$ be a star and $G$ be a line with $k$ vertices. Then,
\[\cD(\BCQ_{\cH, N}, G, K) = \tO(r \cdot N + k).\]
\end{cor}

\begin{algorithm}
\caption{Algorithm for $\cT$ is a Star} \label{algo:starAlgo-garity}
\begin{algorithmic}[1]
\small
\State{\textbf{Input:} A star $P = (v_1, \dots, v_k) \in \cT$ and relations $\{R_{\chi(v_i)} : i \in [k]\}$. Note that $v_1$ is the center.}
\State{\textbf{Output:} $R'_{P}$}
\State{The player containing $R_{\chi(v_1)}$ broadcasts it to all players in $G$.}
\State{Each player containing a relation $R_{\chi(v_i)}$ for every $i \in [2, k]$ computes $R'_{v_i} = \{\mathbf{t} \in R_{\chi(v_1)} : \exists \mathbf{t}' \in R_{\chi(v_i)} \text{ s.t. } \pi_{\chi(v_1) \cap \chi(v_i)} \left(\mathbf{t}\right) =  \pi_{\chi(v_i) \cap \chi(v_1)} \left(\mathbf{t}'\right) \}$ internally.}
\State{$R'_{P} = \bigcap_{i = 2}^{k} R'_{v_i}$ is computed using Theorem~\ref{thm:set-intersection}.}
\State{\Return{$R'_{P}$}}
\end{algorithmic}
\end{algorithm}

\subsubsection{$\cH$ is an acyclic forest}
Similar to the proof of Lemma~\ref{lemma:aritytwo-forest-ub}, we use the analysis on the case when $\cT$ is a star to obtain better upper bounds for the case when $\cH$ is a forest of acyclic hypergraphs (i.e., $\cH = \Fo(\cH)$).
\begin{lemma} \label{lemma:garity-acyclic-ub}
For arbitrary graphs $G$, subset of players $K$ and any GHD $\cT$, we have
\begin{equation} \label{eq:garity-acyclic-ub}
\cD(\BCQ_{\cH, N}, G, K) = \tO\left(y(\cT) \cdot \min_{\Delta \in [|V|]}\left(\frac{N}{\ST(G, K, \Delta)} + \Delta \right)\right).
\end{equation}
Here, $y(\cT)$ denotes the number of internal nodes in $\cT$.
\end{lemma}
\begin{proof}
We start with a proof sketch. We keep removing stars from $\cT$ in a bottom-up fashion. We solve the induced query on each removed star $P$ using the analysis of Section~\ref{sec:garity-star-ub}. Since the number of stars we remove in this process is $y(\cT)$, the stated bound follows.

We now formalize our idea using a recursive algorithm. For simplicity, we assume that $\cH$ has only one acyclic hypergraph with its corresponding $\cT$ having $y(\cT)$ internal nodes. Next, we show that we can solve $\BCQ$ problem on $\cH$ by solving another $\BCQ$ problem on $\cH'$ with $y(\cT) - 1$ internal nodes defined as follows.
We remove the bottom-most star $P = (v_1, \dots, v_{|P|})$ (where $v_1$ is the center and $(v_2, \dots, v_{|P|})$ are the leaves) from $\cT$. We define $V(\cT') = V(\cT) \setminus (v_2, \dots, v_{|P|})$ and $E(\cT') = E(\cT) \setminus \{(v_1, v_i) : i \in [2, |P|]\}$. This implies that $\cH'$ is updated as follows -- 
\[V(\cH') = V(\cH) \setminus \{ \chi(v_i) \cap \chi(v_1) : i \in [2, |P|] \} \text{ and } E(\cH') = E(\cH) \setminus \{ \chi(v_i) : i \in [2, |P|] \}.\] 
Using the arguments of Section~\ref{sec:garity-star-ub}, we can process $P$ in $\tO\left(\min_{\Delta \in [|V|]}\left(\frac{N}{\ST(G, K, \Delta)} + \Delta \right)\right)$ rounds and compute $R'_{P} = \cap_{i = 2}^{|P|} R'_{v_i}$, where 
\[R'_{v_i} = \{\mathbf{t} \in R_{\chi(v_1)} : \exists \mathbf{t}' \in R_{\chi(v_i)} \text{ s.t. } \pi_{\chi(v_1) \cap \chi(v_i)}\left (\mathbf{t}\right) =  \pi_{\chi(v_1) \cap \chi(v_i)} \left(\mathbf{t}'\right)\}.\] 
Finally, we set $R'_{\chi(v_1)} = R'_{P}$ (while the remaining surviving relations remain the same). It is easy to see that $\BCQ$ on $\cH$ is $1$ iff $\BCQ$ on $\cH'$ is $1$.

Note that $\cH'$ is also acyclic, which implies that we can continue this process recursively until $\cT'$ has only one node left (in which case we just check if its relation is empty or not). Thus, the final answer is given by $(R'_{P} \stackrel{?}{\neq} \emptyset)$ and the number of recursive calls is bounded by the number of internal nodes $y(\cT)$. Further, if $\mathcal{H}$ is a forest of acyclic hypergraphs, our argument can be applied individually on every acyclic hypergraph in the forest, resulting in the upper bound~\eqref{eq:garity-acyclic-ub}. This completes the proof.
\end{proof}

\subsubsection{$d$-degenerate Hypergraphs $\cH$ with arity at most $r$}
In this section, we prove our general upper bound result when $\mathcal{H}$ is a $d$-degenerate graph of arity at most $r$.
\begin{lemma} \label{lemma:garity-ub}
For arbitrary $G$, subset of players $K$, and any GHD $\cT$, we have
\begin{equation}\label{eq:garity-ubl}
\cD (\BCQ_{\cH, N}, G, K) = \tO \left( y(\cT) \cdot \min_{\Delta \in [|V|]}\left(\frac{r \cdot N}{\ST(G, K, \Delta)} + \Delta \right) + \tau_{\MCF}(G, K, n_2(\cH) \cdot d \cdot r \cdot N)\right).
\end{equation}
Here, $y(\cT)$ denotes the number of internal nodes in $\cT$.
\end{lemma}
\begin{proof}
We start with a proof sketch. We decompose $\mathcal{H}$ into two components -- an acyclic forest ($\Fo(\cH)$) and a core ($\Co(\cH)$). Then, we use Lemma~\ref{lemma:garity-acyclic-ub} to solve the induced query on $\Fo(\cH)$. For $\Co(\cH)$, we use the \emph{trivial protocol} that sends all the remaining relations to one player.

More formally, consider $\Fo(\cH)$ (via Construction~\ref{lemma:forest-core}). Using the protocol in Lemma~\ref{lemma:garity-acyclic-ub}, we know that \\ $\tO\left(y(\cT) \cdot \min_{\Delta \in [|V|]}\left(\frac{N \cdot r}{\ST(G, K, \Delta)} + \Delta \right) \right)$ rounds suffice to reduce $\cH$ (with corresponding GHD $\cT$) to $\Co(\cH)$. Further, the protocol returns relations of the form $R'_{r}$ for every root $r$ in $\Fo(\cH)$. In particular, since $\chi(r)$ resides with $\Co(\cH)$ for each such root, it is easy to check that $\BCQ$ on $\cH$ has the same answer as $\BCQ$ on $\Co(\cH)$.

We can now use the \emph{trivial protocol} to solve $\BCQ_{\Co(\cH), N}$ on $G$ with $\tau_{\MCF}(G, K, n_2(\cH) \cdot d \cdot r \cdot N)$ via Lemma~\ref{lem:naive}. Recall (from Definition~\ref{defn:garity-notation}) that 
\[y(\cH) = \min_{\forall \cT : \cT\text{ is a GYO-GHD of } \cH} y(\cT) \]
and as a result, for the optimal $\cT$, we have an upper bound of~\eqref{eq:garityub}. This completes the proof.
\end{proof}

\subsection{Lower Bounds} \label{sec:garity-lb}
We start with an overview of our lower bound. Then, we prove lower bounds for the case when $\cH$ is a forest of acyclic hypergraphs (i.e., $\cH=\Fo(\cH)$). Finally, we use the argument for $\Fo(\cH)$ to obtain our lower bounds for general $\cH$. As in Section~\ref{sec:aritytwo-lb}, our lower bounds follow from a reduction from the well-studied TRIBES function.
\subsubsection{Preliminaries and Notation}
We define the concept of Strong Independent Sets (including a lower bound on their size) and introduce a specific construction of GHDs, which we will use in our lower bound arguments.
\begin{defn} [Strong Independent Set] \label{defn: sis}
Given a hypergraph $\cH$, a strong independent set $I \subseteq V(\cH)$ satisfies the following property. For any pair of vertices $u, w \in I, u \neq w$, there exists no hyperedge $e \in E(\cH)$ with $\{u, w\} \subseteq e$.
\end{defn}
\begin{thm} [Size of Strong Independent Set~\cite{sis}] \label{thm:turan2}
Any $d$-degenerate hypergraph $\cH$ with arity at most $r$ has a strong independent set of size at least $\frac{|V(\cH)|}{d \cdot (r - 1)}$. 
\end{thm}
\begin{construction}[MD-GHD] \label{cons:mdghd}
Let $\cT'$ be a GHD of $\cH$ (recall that we mean GYO-GHDs obtained by Construction~\ref{lemma:forest-core} when we say GHD). We now construct our GHD $\cT$ from $\cT'$ as follows. We first set $V(\cT) = V(\cT')$, $E(\cT) = E(\cT')$ and modify $\cT'$ as follows. Consider any parent-child pair $(u, v) \in E(\cT)$, where $u$ is the parent and $v$ is the child. If there exists a node $w \in V(\cT)$ that occurs above $u$ in $\cT$ (i.e., $w$ is a predecessor of $u$) such that $\chi(v) \cap \chi(u) \subseteq \chi(w)$,\footnote{If there are multiple choices, we pick the topmost $w$ among them.} we perform a modification as follows. We delete the edge $e_1 = (u, v)$ from $E(\cT)$ and add the edge $e_2 = (w, v)$ to it. Note that the subtree rooted at $v$ is still preserved. We continue this process until this operation cannot be performed.
\end{construction}
It follows that $\cT$ is still a valid GHD according to Definition~\ref{Definition:ghd}. We now argue that the construction above terminates after a certain number of steps.
\begin{cor} \label{cor:mdghd}
Construction~\ref{cons:mdghd} terminates after at most $|E(\cT)| \cdot y(\cT)$ steps, where $y(\cT)$ is the number of internal nodes in $\cT$.
\end{cor}
\begin{proof}
The proof is by observation. We start by noting that the {\em maximum depth} of $\cT$ is at most $y(\cT) + 1$. In each step of Construction~\ref{cons:mdghd}, note that we pick a parent-child pair $(u, v) \in E(\cT)$ and if there exists a predecessor $w$ of $u$ in $\cT$ satisfying $\chi(v) \cap \chi(u) \subseteq \chi(w)$, we replace $(u, v)$ by $(w, v)$ in $\cT$. Note that there can be only at most $y(\cT)$ such predecessors in the worst-case. Since we can replace any edge in $\cT$ in this way, the total number of replacements that can occur is at most $|E(\cT)| \cdot y(\cT)$.     
\end{proof}
Next, we state and prove a useful property of MD-GHDs, which we will invoke in our arguments.
\begin{lemma} \label{lemma:mdghd}
Let $\cT$ be a MD-GHD and $U=\{u_1,\dots,u_{y(\cT)}\}$ be the set of internal nodes in $\cT$ indexed in a bottom-up fashion (i.e., leaves to root). In particular, if $u_i$ is a {\em descendant} of $u_j$ in $\cT$ for any $i, j \in y(\cT)$, then $j > i$. For each $u_i \in U$, there exists an attribute $p_i$ that does not occur in any bag of $\mathcal{T}$ except descendants of $u_i$. Further, there are at least two $R_{p_i, 1}, R_{p_i, 2}$ distinct relations on hyperdeges $p_{i, 1} \neq p_{i, 2} \in E(\cH)$ incident on $p_i$.
\end{lemma}
\begin{proof}
%The proof is by induction on $i \in [y(\cT)]$. 
We start with the case when $i = 1$ and consider the star $P_1 = (v_1 = u_1, v_2, \dots, v_{|P_i|})$, where $v_1$ is the center and $(v_2, \dots, v_{|P_i|})$ are direct descendants of $v_1 \in \cT$. Next, we claim that each set $\chi(v_1) \cap \chi(v_j)$ for every $j \in [2, |P_1|]$ contains at least one attribute $p^1_{j}$ such that $p^1_{j}$ does not occur anywhere in $\mathcal{T} \setminus P_1$. 

Assume otherwise i.e., $v_j$ for some $j \in [2, |P_1|]$ violates our claim. Then, we have that there exists some node $w \in \cT \setminus P_1$ such that $\chi(v_1) \cap \chi(v_j) \subseteq \chi(w)$. For this to be true, we need $w$ to be the parent of $v_1$ by the running intersection property of GHDs. If $w$ is indeed the parent of $v_1$, then this scenario contradicts the fact that $\cT$ is a MD-GHD by Construction~\ref{cons:mdghd}, which would replaced the edge $(v_1, v_j) \in \cT$ by $(u_i, v_j)$. Thus, we have argued that for each set $\chi(v_1) \cap \chi(v_j)$ for every $j \in [2, |P_1|]$, there exists at least one attribute $p^{1}_{j}$ that does not occur anywhere in $\mathcal{T} \setminus P_1$. Finally, recall that $\cH$ is acyclic and as a result, we have that for each node $v \in V(\cT)$, $\chi(v) \in E(\cH)$. In particular, this implies the hyperedges $R_{p_i, 1} = \chi(v_1)$ and $R_{p_i, 2} = \chi(v_j)$ are in $E(\cH)$. 

To complete the proof, we recursively make the $i = 1$ argument on $\cT' = \left( V(\cT) \setminus \{v2, \dots, v_{|P_i|}\}, E(\cT) \setminus \{ (v_i, v_j) : j \in [2, |P_i|] \} \right )$ until $\cT'$ is empty. Since $\cT'$ is a tree as well, note that the above argument goes through without any issues.  
\end{proof}

\subsubsection{Lower Bounds for $\Fo(\cH)$}
\begin{thm} \label{thm:garity-acyclic-lb}
When $\cT$ is a MD-GHD for $\Fo(\cH)$, we have
\[ \TRB_{\frac{y(\cT)}{r}, N} \le \BCQ_{\cH, N}.\]
Here, $y(\cT)$ is the number of internal nodes in $\cT$.
\end{thm}
\begin{proof}
Given $\mathcal{H}$, $\cT$ and a $\TRB_{\frac{y(\cT)}{r}, N}$ instance we design a corresponding $\BCQ_{\cH, N}$ instance. Using Lemma~\ref{lemma:mdghd}, we have shown that for each internal node $u_i \in U$, there exists an attribute $p_i$ such that there are at least two relations $R_{p_{i,1}}, R_{p_{i,2}}$ on hyperedges $p_{i,1}\ne p_{i,2} \in E(\cH)$ incident on $p_i$, and in addition, $p_i$ does not occur in (any bag of) $\cT = V(\cT) \setminus (\cup_{\ell<i}P_\ell)$, where $P_\ell$ is a star with $u_{\ell}$ as center and its direct descendants in $\cT$ as leaves.

We now consider the set $P=\left\{p_1,\dots,p_{|y(\cT)|} \right\}$ of attributes and claim that $P$ includes a strong independent set $I$ of size at least $\frac{|y(\cT)|}{r}$.
We construct such a set greedily.
We use the following observation in our analysis. For any $p_i$, we have
\begin{align*}
\left|\left \{p_j \mid j > i \text{ and } \exists e \in E(\cH)\  \text{ s.t. } \ p_i,p_j \in e \right \} \right| \leq r-1.
\end{align*}
Assume otherwise, then there exists $p_i$ that shares edges with $r$ attributes $p_j$ for $j>i$. By the discussion above, as such edges include $p_i$ they must be associated with $\cup_{\ell<i}P_\ell$. This implies, via the running intersection property of $\cT$, that these $r$ attributes together with $p_i$ are in $\chi(u_i)$ in contradiction to $|\chi(u_i)|\le r$.

Next, we start the greedy construction with $I=\{p_1\}$, and remove $p_1$ from $P$ together will all $p_j$ that share an edge with $p_1$. We have removed at most $r$ attributes from $P$. We continue recursively. At step $\ell$, we consider the smallest index $i$ for which $p_i$ has not been removed from $P$. We add $p_i$ to $I$ and remove $p_i$ and any $p_j$ that shares an edge with $p_i$ from $P$. As all $p_\ell$ for $\ell < i$ have been removed from $P$, we only remove $r$ additional nodes from $P$. By this greedy process, the final $I$ is an independent set of size at least $\frac{|y(\cT)|}{r}$.

Assume WLOG that the strong independent set $I$ satisfies $|I| = \frac{y(\cT)}{r}$ (otherwise we take a subset of $I$ with size $\frac{y(\cT)}{r}$). Associating a pair of sets $(S_i,T_i)$ from $\TRB_{\frac{y(\cT)}{r}, N}(\hat{S},\hat{T})$ with each node $p_i \in I$, we have (by definition) $\TRB_{\frac{y(\cT)}{r}, N} (\hat{S}, \hat{T})= \bigwedge_{p_i \in I} \SD_{N}({S_i}, {T_i})$.
We now construct a corresponding $\BCQ_{\cH, N}$ instance in detail. 
We start by defining a pair of relations corresponding to each pair $({S_i}, {T_i})$. 
Recall that each $p_i \in I$ corresponds to a $u_i \in U$, such that $p_i \in \chi(u_i)$ and $\exists u_i' \in \children(u_i) : p_i \in \chi(u_i')$. 
We define relations $R_{S_i} = R_{p_i, 1} = S_{i} \times_{i = 2}^{|\chi(u_i)|} \{1\}$ and $R_{T_i} = R_{p_i, 2} = T_{i} \times_{i = 2}^{|\chi(u_i')|} \{1\}$, where both $R_{S_i}$ and $R_{T_i}$ have their first attribute as $p_i$ and are cartesian products where all attributes except $p_i$ have a trivial domain of $\{1\}$.
In particular, we have $\attr(R_{S_i}) = \chi(u_i)$ and $\attr(R_{T_i}) = \chi(u_i')$. Further, we treat $S_i$ and $T_i$ as subsets of $[N]$ (instead of elements in $\{0,1\}^N$).
Thus, $\TRB_{\frac{y(\cT)}{r}, N}(\hat{S}, \hat{T}) = 1$ iff for each $p_i \in I$, the join $R_{S_i}  \Join R_{T_i}$ is not empty. 
To complete the description of $\BCQ_{\cH, N}$, we need to define the other relations in $\cH$ as well. Note that all the remaining relations $R' = \{ \{R_e : e \in E(\cH)\} \setminus \{R_{S_i} \cup R_{T_i} : p_i \in I\} \}$ can be incident on only at most one $p_i \in I$ (as $I$ is a strong independent set).
If $R_e \in R'$ is incident on $p_i \in I$, we define $R_{e} = \{(\ell, 1, \dots, 1) : \ell \in [N]\}$ (where $p_i$ is the first attribute in $e$). Otherwise, we set $R_{e} = \{(1, \dots, 1)\}$ (note that the order of attributes does not matter here). Let us denote the BCQ instance constructed above by $q_{\cH, \hat{S}, \hat{T}}$.

To complete the proof, we show that $q_{\cH, \hat{S}, \hat{T}}=1$ iff $\TRB_{\frac{y(\cT)}{r}, N}(\hat{S}, \hat{T}) =1$.
In particular, if $q_{\cH, \hat{S}, \hat{T}} = 1$, there exists a tuple $\mathbf{t} \in \prod_{v \in V(\cH)}\Dom(v)$ that satisfies all relations in $q_{\cH, \hat{S}, \hat{T}}$ i.e., $\mathbf{t}_{e} \in R_{e}$ for every $e \in E(\cH)$.
Specifically, for each $p_i \in I$, $R_{S_i}  \Join R_{T_i}$ is not empty, which implies  $\TRB_{\frac{y(\cT)}{r}, N}(\hat{S}, \hat{T}) = 1$.
Alternatively, if $\TRB_{\frac{y(\cT)}{r}, N}(\hat{S}, \hat{T}) = 1$, we can find tuple $\mathbf{t} \in \prod_{v \in V(\cH)} \Dom(v)$ that satisfies all relations in $q_{\cH, \hat{S}, \hat{T}}$.
For each $p_i \in I$, we set $\pi_{p_i}(\mathbf{t})$ to be in the intersection of $S_i$ and $T_i$, and for all remaining nodes $v \in V(\cH) \setminus I$ we set $\pi_{v}(\mathbf{t})=1$.
Note that this implies all the relations in $q_{\cH, \hat{S}, \hat{T}}$ are satisfied. This concludes our proof.
\end{proof}
Note that the above argument was independent of $G$. We now use the structure of $G$ to obtain a lower bound on $\cR(\BCQ_{\cH, N}, G, K)$ using known results for $\TRB_{\frac{y(\cT)}{r}, N}$.
\subsubsection{Lower Bounds dependent on $G$}
We show the following lower bound for arbitrary $G$ assuming worst-case assignment of relations to players in $K$.
\begin{lemma} [Arbitrary $G$] \label{lemma:acyclic-lb-1} If $\cH=\Fo(\cH)$, then
\[\cR(\BCQ_{\cH, N}, G, K) \ge \tOm\left(\frac{y(\cH) \cdot N}{r \cdot \MC(G)}\right).\]
\end{lemma}
\begin{proof}
We first consider a min-cut $(A, B)$ of $G$ that separates $K$, where $A$ and $B$ denote the set of vertices in each partition ($A \cup B = V(G)$). Using the notation used in the proof of Theorem~\ref{thm:garity-acyclic-lb}, let $q_{\cH, \hat{S}, \hat{T}}$ be the query computed on a MD-GHD $\cT$ corresponding to a given instance $\TRB_{\frac{y(\cT)}{r}, N}(\hat{S}, \hat{T})$. We assign relations $\{R_{S_i}\}_{p_i \in I}$ to vertices in $A$ and relations $\{R_{T_i}\}_{p_i \in I}$ to vertices in $B$. The other relations in $q_{\cH, \hat{S}, \hat{T}}$ can be assigned arbitrarily. Note that any protocol to compute $q_{\cH, \hat{S}, \hat{T}}$ on $G$ gives a two-party protocol (Alice, Bob) for $\TRB_{\frac{y(\cT)}{r}, N}(\hat{S}, \hat{T})$. In particular, Alice gets the sets $\{S_{i}\}_{p_i \in I}$ (corresponding to $R_{S_i}$) assigned to vertices in $A$ and Bob gets the sets $\{T_{i}\}_{p_i \in I}$ (corresponding to $R_{T_i}$) assigned to vertices in $B$ (ignoring the additional relations). It is not too hard to see that if there exists a $\cR(\BCQ_{\cH, N}, G, K)$ round protocol for $\TRB_{\frac{y(\cT)}{r}, N}$ on $G$, then we have a two-party protocol (i.e., on a graph $\cG= (\{a, b\}, (a, b))$) with at most $\cR(\BCQ_{\cH, N}, G, K) \cdot \MC(G,K) \cdot \ceil{\log_{2}(\MC(G, K))}$ rounds (see Proof of Lemma~\ref{lemma:tree-lb-1} for a detailed discussion). Since $\cR(\BCQ_{\cH, N}, G, K) \cdot \MC(G,K) \cdot \ceil{\log_{2}(\MC(G, K))}$ is lower bounded from Theorem~\ref{thm:jks} by $\tOm\left(\frac{y(\cT) \cdot N}{r}\right)$ and (since by definition of $y(\cH)$) $y(\cT)\ge y(\cH)$, we conclude our assertion.
\end{proof}

\subsubsection{Lower Bounds for $d$-degenerate hypergraphs $\cH$} \label{sec:garity-lb-general}
We are now ready to prove our general lower bound for all $d$-degenerate hypergraphs $\cH$.
\begin{thm} \label{thm:garity-lb-1}
For arbitrary $G$, subset of players $K$ and $d$-degenerate hypergraphs $\cH$ with a MD-GHD $\cT$, we have
\begin{equation} \label{eq:garity-lb-1}
\cR (\BCQ_{\cH, N}, G, K) \ge \tOm \left( \frac{\left( \frac{y(\cT)}{r} + \frac{n_2(\cH)}{d \cdot r}\right) \cdot N}{\MC(G, K)} \right).
\end{equation}
Here, $y(\cT)$ denotes the number of internal nodes in $\cT$.
\end{thm}
\begin{proof}
Let $m_1 =\frac{y(\cT)}{r}$ and $m_2 = \frac{n_2(\cH)}{d \cdot r}$. We obtain two independent lower bounds on $\cH$ and our final bound is the maximum between them (which is at least half their sum). In general, as in Theorem~\ref{thm:garity-acyclic-lb}, given $\mathcal{H}$ and a TRIBES instance $\TRB_{m_j, N}(\hat{S}, \hat{T})$ for every $j \in [2]$, we construct a BCQ query $q^{(j)}_{\cH, \hat{S}, \hat{T}}$ on $\mathcal{H}$ such that $q^{(j)}_{\cH, \hat{S}, \hat{T}}=1$ iff $\TRB_{m_j, N}=1$.
To this end, we need to ``embed'' the $m_j$ pairs of sets $(S_i,T_i)$ from $\TRB_{m_j, N}(\hat{S}, \hat{T})$ as relations in $q^{(j)}_{\cH, \hat{S}, \hat{T}}$. 
Recall that $\Co(\cH)$ is present at the root of $T$.
It is easy to check that one can apply the reduction on Theorem~\ref{thm:garity-acyclic-lb} to construct $q^{(1)}_{\cH, \hat{S}, \hat{T}}$ with the required properties.

Finally for $m_2$, we apply Theorem~\ref{thm:turan2} on the root of $T$ (i.e., $\Co(\cH)$) to obtain a strong independent set of size at least $\frac{n_2(\cH)}{d \cdot (r - 1)} \ge \frac{n_2(\cH)}{d \cdot r}$ (since $r \ge 2$).
We then use a proof identical to that given in Theorem~\ref{thm:garity-acyclic-lb} to embed the TRIBES instance $\TRB_{m_2, N}$ onto $\Co(\cH)$.

Let $m = \max(m_1, m_2) = \max\left(\frac{y(\cT)}{r},  \frac{n_2(\cH)}{d \cdot r} \right)$. Since sum and max are within a factor $2$ of each other, we can write $m \ge \frac{y(\cT)}{2 \cdot r} + \frac{n_2(\cH)}{2 \cdot d \cdot r}$. We can now apply ideas from the proof of Lemma~\ref{lemma:acyclic-lb-1} to obtain the required lower bound $\tOm\left( \frac{\left( \frac{y(\cT)}{r} + \frac{n_2(\cH)}{d \cdot r}\right) \cdot N}{\MC(G, K)} \right)$. This concludes our proof.
\end{proof}
Finally, we prove Theorem~\ref{thm:garitymain}.
\begin{proof}[Proof of Theorem~\ref{thm:garitymain}]
The upper bound~\eqref{eq:garity-ubl} follows from Lemma~\ref{lemma:garity-ub}.  For the lower bound, note that our bounds depend on an arbitrary MD-GHD $\cT$ for $\cH$. By definition~\ref{defn:garity-notation}, we have that $y(\cT)\ge y(\cH)$ and the lower bound~\eqref{eq:garitylb} follows. Using Definition~\ref{defn:mdw}, we have that upper and lower bounds match for the GHD that achieves the internal-node-width $\mdw(\cH)$. 
\end{proof}
We conclude this section by noting that when $N \ge |V|^{2}$ our upper and lower bounds differ by $\tilde{O}(d^{2} \cdot r^{2})$ factor (for worst-case assignments of relations to players). In particular, Theorem~\ref{thm:lau} implies that the first two terms in the upper and lower bounds match up to an $\tilde{O}(r^2)$ factor. Using the same arguments as in Appendix~\ref{app:mcf-mc}, we can show that for worst-case assignment of relations, we have the second terms in the upper and lower bounds differ by a $\tilde{O}(d^{2} \cdot r^{2})$ factor, as desired. 
\section{Bounds for General FAQs and Assumptions in Model~\ref{model:our_model}} \label{app:general-faq}
In this section, we prove Theorem~\ref{cor:faq-const1} and address assumptions in Model~\ref{model:our_model}. We start with the redefinition of the $\FAQ$ problem and state some known results.
\subsection{Preliminaries and Existing Results} \label{sec:faq_prelims}
We define the general FAQ problem here. We are given a multi-hypergraph $\cH=(\cV,\cE)$ where for each hyperedge $e\in\cE$, we are given an input function $f_e:\prod_{v\in e} \Dom(v)\to \D$. In addition, we are given a set of {\em free variables} $\cF\subseteq \cV : |\cF| = \ell$. We would like to note that our results hold only for specific choices of $\cF$. For a fixed $\cF$, the vertices in $\cV$ can be renumbered so that $\cF = [\ell]$ WLOG. We would like to compute the function:
\begin{equation} \label{eq:faqgeneraldef2}
\phi\left(\vx_{[\ell]}\right) =\underset{x_{\ell + 1} \in \Dom(X_{\ell + 1})}{\tenssum{}^{(\ell + 1)}}\dots \underset{x_{n} \in \Dom(X_{n})}{\tenssum{}^{(n)}}  \left( \underset{S \in \cE} {\bigotimes} f_{S}(\vx_{S}) \right),
\end{equation}
where we use $\vx=(x_u)_{u\in\cV}$ and $\vx_S$ is $\vx$ projected down to co-ordinates in $S\subseteq \cV$. 
The variables in $\cV \setminus \cF$ are called \emph{bound variables}.
For every bound variable $i > \ell$, $\tenssum{}^{(i)}$ is a binary (aggregate) operator on the domain $\D$. Different bound variables may have different aggregates. 
Finally, for each bound variable $i > \ell$ either $\tenssum{}^{(i)} = \tensprod{}$ (\emph{product aggregate}) or $(\D,\tenssum{}^{(i)},\tensprod{})$ forms a commutative semiring (\emph{semiring aggregate}) with the same additive identity $\mathbf{0}$ and multiplicative identity $\mathbf{1}$.
As with database systems, we assume that the functions are input in the {\em listing} representation i.e., the function $f_e$ is represented as a list of its non-zero values: $R_e=\{(\vy,f_e(\vy))|\vy\in\prod_{v\in e} \Dom(v):f_e(\vy)\neq \vzero\}$. Let $\FAQ_{\D, \cH, N, \cF = [\ell]}$ denote the class of $\FAQ$ problems, where each function $f_e$ for $e\in\cE$ has at most $N$ non-zero entries. Note that we are not explicitly stating $(\tenssum{}^{(\ell + 1)}, \dots, \tenssum{}^{(n)})$ since our results hold for all such choice of operators for the bound variables.

When $\tenssum{}^{(i)} = \tenssum{}$ for every $i \in [\ell + 1, n]$ and $(\D,\tenssum{}^{(i)},\tensprod{})$ forms a commutative semiring, we have the $\FAQS$ problem. We have already seen that $\BCQ$ and computing some Factor Marginals in PGMs are special cases of $\FAQS$. We restate them in the language of $\FAQ$ for completeness. When $\cF = \emptyset$ and $\D= \{0, 1\}$ (i.e., the {\em Boolean semi-ring}), $\FAQ_{\{0, 1\}, \cH, N, \emptyset}$ corresponds to the {\em Boolean Conjunctive Query}, which we denote by $\BCQ_{\cH,N}$. Further, if $\cF = \cV$ and $\D = \{0 ,1\}$, we have the {\em natural join} problem in Definition~\ref{defn:join} and if $\cF = e$ for any $e \in \cE$ with $\D = \{0, 1\}$, we have the {\em semijoin} problem in Definition~\ref{defn:semijoin}. We would like to mention that $R_e$ can be equivalently represented as $R_{V(e)}$. 

We state a result from~\cite{ajar} here.
\begin{thm} \label{thm:ajar-push-down} [Theorem $9$ from~\cite{ajar}]
Suppose we are given a relation $R_{e}$ for some $e \in E(\cH)$ such that $Z_1, Z_2 \in e$ and two operators $\tenssum{}^{(Z_1)},\tenssum{}^{(Z_2)}$. Then, we have
\begin{align*}
\underset{z_{1} \in \Dom(Z_{1})}{\tenssum{}^{(Z_1)}} \underset{z_{2} \in \Dom(Z_{2})}{\tenssum{}^{(Z_2)}}  R_{e} \quad =  \quad \underset{z_{2} \in \Dom(Z_{2})}{\tenssum{}^{(Z_2)}} \underset{z_{1} \in \Dom(Z_{1})}{\tenssum{}^{(Z_1)}} R_{e} 
\end{align*}
if one of the following conditions hold:
\begin{itemize}
\item {$\tenssum{}^{(1)} = \tenssum{}^{(2)}$}
\item {There exists two relations $R_{e'}$ and $R_{e''}$ such that $Z_1 \not \in e'$, $Z_2 \not \in e''$ and $R_{e'} \Join R_{e''} = R_{e}$.}
\end{itemize}
\end{thm}
We use the above theorem to obtain the following result.
\begin{cor} \label{cor:op-swap}
If there exists a function $R_{e'}$ for $e' \in E(\cH)$ and other functions $R_{e}$ for every $e \in E(\cH) \setminus \{e'\}$ such that the set of attributes $V(e') \supseteq (Z_1, \dots Z_w)$ satisfies $Z_1, \dots, Z_w \not \in e$ for every $e \in E(\cH) \setminus \{e'\}$. Then, we have
\begin{equation} \label{eq:op-swap}
\left(\underset{z \in \Dom(Z)} {\tenssum{}^{(z)}} \right)_{Z \in V(\cH)} \left( \underset{e \in E(\cH)}{\bigotimes} R_{e} \right) = \left( \left( \underset{z \in \Dom(Z)}{\tenssum{}^{(z)}} \right)_{z \in \cup_{v \in V(\cH) \setminus \{Z_1, \dots, Z_w\}} } \underset{e \in E(\cH) \setminus \{e'\}}{\bigotimes} R_{e} \right) \bigotimes \left(\underset{z_1 \in \Dom(Z_1)}{\tenssum{}^{(Z_1)}} \dots \underset{z_w \in \Dom(Z_w)}{\tenssum{}^{(Z_w)}} R_{e'} \right).
\end{equation}
\end{cor}
We now argue the above corollary, which is a direct application of Theorem~\ref{thm:ajar-push-down}. In particular, $R_{e'}$ be a relation as described above. We can now rewrite the LHS of~\eqref{eq:op-swap} as follows:
\begin{align*}
\left(\underset{z \in \Dom(Z)} {\tenssum{}^{(z)}} \right)_{Z \in V(\cH)} \left( \underset{e \in E(\cH)}{\bigotimes} R_{e} \right) & =  \left(\underset{z \in \Dom(Z)} {\tenssum{}^{(z)}} \right)_{Z \in V(\cH)} \left(\left( \underset{e \in E(\cH) \setminus e'}{\bigotimes} R_{e} \right) \bigotimes R_{e'} \right) \\
& = \left( \left( \underset{z \in \Dom(Z)}{\tenssum{}^{(z)}} \right)_{z \in \cup_{v \in V(\cH) \setminus \{Z_1, \dots, Z_w\}} } \underset{e \in E(\cH) \setminus \{e'\}}{\bigotimes} R_{e} \right) \bigotimes \left(\underset{z_1 \in \Dom(Z_1)}{\tenssum{}^{(Z_1)}} \dots \underset{z_w \in \Dom(Z_w)}{\tenssum{}^{(Z_w)}} R_{e'} \right).
\end{align*}
Here, the first equality follows from the definiton of $\underset{e \in E(\cH)}{\bigotimes} R_{e}$ and the second equality follows by invoking Theorem~\ref{thm:ajar-push-down}. In particular, we have $(Z_1, \dots, Z_w) \not \in \cup_{e \in E(\cH) \setminus \{e'\}} V(e)$ and $\left( \underset{e \in E(\cH) \setminus e'}{\bigotimes} R_{e} \right) \bigotimes R_{e'}  = \underset{e \in E(\cH)}{\bigotimes} R_{e}$ (by definition) and we can combine both to get~\eqref{eq:op-swap} as required. In words, since the variables $(Z_1, \dots, Z_w)$ are {\em private} to $R_{e'}$, we can `push down' the aggregations $\left(\underset{z_1 \in \Dom(Z_1)}{\tenssum{}^{(Z_1)}} \dots \underset{z_w \in \Dom(Z_w)}{\tenssum{}^{(Z_w)}}\right)$ inside every tuple in $R_{e'}$, as required.

We consider the model, where there is only one player $P$ in $G$, having all the $k$ input relations. Our goal is to measure the time complexity for $P$ to compute $\FAQ_{\D, \cH, N, \emptyset}$ in the RAM model, which focuses on the number of steps for computation (ignoring operations like memory access). We prove the following result, which follows from~\cite{faq,ajar}. We will use this proof subsequently in our distributed algorithm in Section~\ref{sec:faqub2}.
\begin{thm} \label{thm:faqthm}
When $\cH$ is acyclic, for an arbitrary set of operators $(\tenssum{}^{(i)})_{i \in [n]}$, the time complexity of computing $\FAQ_{\D, \cH, N, \emptyset}$ in the RAM model is $\tilde{O}(N)$. 
\end{thm}
\begin{proof}
For any input function $f$ such that $f_{e} : \prod_{v \in e} \Dom(v) \rightarrow \D$ and an arbitrary set of operators $(\tenssum{}^{(i)})_{i \in [n]}$, we can write
\begin{equation}
Q = \underset{x_{1} \in \Dom(X_{1})}{\tenssum{}^{(1)}}  \dots \underset{x_{n} \in \Dom(X_{n})}{\tenssum{}^{(n)}} \underset{e \in E(\cH)}{\tensprod{}} R_{e}
\end{equation}
using~\eqref{eq:faqgeneraldef}. Here, $Q$ is an instance of $\FAQ_{\D, \cH, N, \emptyset, (\tenssum{}^{(i)})_{i \in [n]}}$.
Recall that $R_{e}$ is the {\em listing representation} of $f_{e}: \{(\vy,f_e(\vy))|\vy\in\prod_{u\in e} Dom(u):f_e(\vy)\neq \vzero\}$ for every $e \in E(\cH)$. 
We now use Construction~\ref{lemma:forest-core} on $\cH$ obtaining a GHD $\cT$ where each node $v \in \cT$ corresponds to a hyperedge $\chi(v) \in E(\cH)$ (see Definition~\ref{defn:acyclic}).

We describe the algorithm here, which uses a message-passing algorithm (upward pass) on a GHD $\cT$. In particular, in a bottom-up fashion, every node $v \in \cT$ performs two computations -- first, it updates the relation $R_{\chi(v)}$ based on the messages received from all its neighbors. Second, if it is not a root, then it computes the message $v$ needs to send to its parent $v' = \parent(v)$. We obtain the final answer for $Q$ in the root. 

We now formalize the algorithm above. Since we are considering that one player contains all the $k$ input relations, we can assume that all relations $R_{\chi(v)}$, where  $v$ is a node in $\cT$ can be accessed at any point in time without any additional communication. Let $v$ be the current node under consideration in our algorithm. We update the relation $R_{\chi(v)}$ as follows:
\begin{equation} \label{eq:message-up}
R_{\chi(v)} \leftarrow R_{\chi(v)} \tensprod{} \left(\underset{u \in \Gamma(v)}{\tensprod{}} m_{u, v} \right), 
\end{equation}
where $\Gamma(v)$ and $m_{u, v}$ denote the neighborhood of $v \in \cT$ and the message sent from $u$ to $v$ respectively. 
Initialize $w = \mathbf{1}$. 
For every tuple $\mathbf{t} \in R_{\chi(v)}$ and for all tuples $\mathbf{t}' \in m_{u, v}$ with $\pi_{\chi(u) \cap \chi(v)} (\mathbf{t}') = \pi_{\chi(u) \cap \chi(v)} (\mathbf{t})$ for every $u \in \Gamma(v)$, we compute the running product $w \leftarrow w \cdot f(\mathbf{t'})$. 
Then, the tuple $\mathbf{t}$ in $R_{\chi(v)}$ is updated as $(\mathbf{t}, f(\mathbf{t}) \cdot w)$. Recall by definition that $|R_{\chi(v)}| \le N$. We claim the following:
\begin{equation} \label{eq:up-msg-size}
|m_{u, v}| \le N \quad \forall u \in \cT \text{ s.t. } u \text{ is the parent of } v \in \cT.
\end{equation}
Assuming~\eqref{eq:up-msg-size} is true (we will prove at the end of this section), we prove that~\eqref{eq:message-up} can be computed in $\tilde{O}(N)$ time. First, we observe that for a fixed tuple $\mathbf{t} \in R_{\chi(v)}$, there exists at most one tuple $\mathbf{t}' \in m_{u, v}$ such that $\pi_{\chi(u) \cap \chi(v)} (\mathbf{t}') = \pi_{\chi(u) \cap \chi(v)} (\mathbf{t})$ for every $u \in \Gamma(v)$. 
Then, we traverse through all tuples in $R_{\chi(v)}$ in the worst case and our stated claim follows. We call this \emph{Step 1}.

If $v$ is not the root of $\cT$, the message $m_{v, v'}$ that $v$ needs to send to its parent $v' = \parent(v) \in \cT$ is computed on the updated $R_{\chi(v)}$ (from~\eqref{eq:message-up}) as follows. Notice that the variables in the set $\chi(v) \setminus \chi(v') = (Z_1, \dots, Z_w)$ are \emph{private} to the node $v$. In particular, all variables in $\chi(v) \setminus \chi(v')$ are not present anywhere apart from the subtree of $\cT$ rooted at $v$ (follows from the running intersection property of $\cT$). Notice that the attributes $(Z_1, \dots, Z_w)$ are present in $\subseteq (X_{1}, \dots, X_{n})$. Consider the \emph{reduced} FAQ query at $v$ given by
\begin{align*}
Q_{v} = \left(\underset{z \in \Dom(Z)}{\tenssum{}^{(Z)}}\right)_{Z \in \underset{y \in \cT'}{\cup} \chi(y)} \quad \underset{y \in V(\cT')}{\tensprod{}} R_{\chi(y)},
\end{align*}
where $\cT'$ denotes the set of nodes that haven't been processed in $\cT$ so far in the message-up algorithm (which includes $v$). We can rewrite $Q_{v}$ by invoking Corollary~\ref{cor:op-swap} as follows:
\begin{equation} \label{eq:op-swap-2}
Q_v = \left( \left(\underset{z \in \Dom(Z)}{\tenssum{}^{(Z)}}\right)_{Z \in \cup_{y \in \cT' \setminus \{v\}} \chi(y)} \prod_{y \in \cT' \setminus \{ v\}} R_{\chi(y)} \right) \tensprod{} \left(\underset{z_1 \in \Dom(Z_1)}{\tenssum{}^{(Z_1)}} \dots \underset{z_w \in \Dom(Z_w)}{\tenssum{}^{(Z_w)}} R_{\chi(v)}\right).
\end{equation}
In particular, we are ``pushing down" the aggregations  
 $\left(\underset{z_1 \in \Dom(Z_1)}{\tenssum{}^{(z_1)}}, \dots, \underset{z_{w} \in \Dom(Z_w)}{\tenssum{}^{(z_w)}}\right)$ inside every tuple in the relation $R_{\chi(v)}$ since they are not contained in any relation $R_{\chi(y)}$ for every $y \in \cT' \setminus \{v\}$. 
 In other words, the attributes $(Z_1, \dots, Z_w)$ belong to only relations in the subtree of $\cT$ rooted at $v$. 
 Further, observe that this computation is performed at node $v$.
Note that the aggregations are computed on the annotated values of the relations as follows.
For every tuple $\mathbf{t} \in \pi_{\chi(v) \cap \chi(v')} R_{\chi(v)}$, the tuple 
\begin{align*}
\left(\mathbf{t},  \quad \left( \left( \underset{\pi_{Z_{i}}(\mathbf{t'})  }{ \tenssum{}^{(Z_{i})}}  f(\mathbf{t}') \right)_{i \in [w]} \right)_{\pi_{\chi(v) \cap \chi(v'), \forall \mathbf{t}' \in R_{\chi(v)}} \mathbf{t}'  = \mathbf{t}, \forall \mathbf{t'} \in R_{\chi(v)}} \right)
\end{align*}
is added to the message $m_{v, v'}$. 
In particular, for each tuple $\mathbf{t} \in R_{\chi(v)}$, we aggregate the annotated values of all tuples $\mathbf{t}'$ that satisfy $\pi_{\chi(v) \cap \chi(v'), \forall \mathbf{t}' \in R_{\chi(v)}} \mathbf{t}'  = \mathbf{t},$. 
If $v$ is the root of $\cT$, we have $\chi(v) = (Z_1, \dots, Z_w)$ and as a result,~\eqref{eq:op-swap-2} will have only the right hand side of the product. 
Thus, the final answer for $Q$ can be computed from $v$.
Notice that this computation can be done in $\tilde{O}(N)$ time since $|m_{v, v'}| \le |R_{\chi(v)}|$ (which, combined with $|R_{\chi(v)}| \le N$ proves~\eqref{eq:up-msg-size} as required) and we might traverse through all tuples in $R_{\chi(v)}$ in the worst case. We call this \emph{Step 2}.

Finally, when the algorithm terminates, we need to argue that we obtain the correct result for $Q$. Consider the first node $v \in \cT$ considered in our message up process. The \emph{reduced} FAQ query $Q_{v} = Q$'s correctness follows from Corollary~\ref{cor:op-swap}. Since we repeatedly apply the same procedure for all other nodes in $v \in \cT \setminus \{v\}$, the correctness follows. Since both \emph{Step 1} and \emph{Step 2} take only $\tilde{O}(N)$ time and our choice of $(\tenssum{}^{(i)})_{i \in [n]}$ was arbitrary, this completes our proof.
\end{proof}

\subsection{Main Theorem}
We prove the following theorem in our model (Model~\ref{model:our_model}) assuming that any hypergraph can be decomposed into a forest $\Fo(\cH)$ and a core $\Co(\cH)$ using Construction~\ref{lemma:forest-core}.
\begin{thm} \label{thm:faq2}
For arbitrary $G$, subset of players $K$, any $\cF \subseteq V(\Co(\cH)) : |\cF| = \ell$, and $d$-degenerate hypergraphs $\cH$ with arity at most $r$, we have
\begin{equation}\label{eq:faqgeneral1}
\mathcal{D} \left(\FAQ_{\D, \cH, N, \cF}, G, K\right) = \tO\left( y(\cH) \cdot \min_{\Delta \in [|V|]}\left(\frac{N \cdot r}{\ST(G, K, \Delta)} + \Delta \right) + \tau_{\MCF}(G, K, n_2(\cH) \cdot d \cdot r \cdot N)\right).
\end{equation}
Further, we have
\begin{equation} \label{eq:faqlb2}
\cR \left(\FAQ_{\D, \cH, N, \cF}, G, K \right) \ge \tOm\left(\frac{ y(\cH) \cdot N}{r \cdot \MC(G, K)} + \frac{n_2(\cH) \cdot N}{d \cdot r\cdot \MC(G, K)}\right).
\end{equation}
Both the results hold for any $\D$ and any choice of operators $(\tenssum{}^{\ell + 1}, \dots, \tenssum{}^{n})$ over $\D$ as defined in Section~\ref{sec:faq_prelims}.
\end{thm}
We would like to note here that for simple graphs $\cH$, we can overcome the factor of $d$ in the lower bound (see Theorem~\ref{thm:aritytwo2}). In particular, we can use similar ideas from there to prove Theorem~\ref{cor:faq-const1}.
\subsection{Upper Bound for General FAQs} \label{sec:faqub2}
The upper bound follows from a slight modification of our algorithm to compute $\BCQ_{\cH, N}$ (Algorithm~\ref{algo:starAlgo-garity}) and uses ideas from the proof of Theorem~\ref{thm:faqthm} to `push down' a specific subset of operators in $(\tenssum{}^{\ell + 1}, \dots, \tenssum{}^{n})$. We present a proof sketch here. Let's fix an input function $f$ such that $f_e : \prod_{v \in e} \Dom(v) \rightarrow \D$. Using~\eqref{eq:faqgeneraldef}, we can write
\begin{align*}
\FAQ_{\D, \cH, N, \cF} &= \underset{x_{\ell + 1} \in \Dom(X_{\ell + 1})}{\tenssum{}^{(\ell + 1)}}  \dots \underset{x_{n} \in \Dom(X_{n})}{\tenssum{}^{(n)}} 
	\left( \prod_{e \in E(\cH) \setminus \Fo(\cH)} R_{V(e)}  \prod_{e' \in \Fo(\cH)} R_{V(e')}\right)
\end{align*}
since $\Fo(\cH)$ is a sequence of hyperedges in $\cH$. Recall that $R_{e} : e \in E(\cH)$ is the {\em listing representation} of $f_{e}: \{(\vy,f_e(\vy))|\vy\in\prod_{u\in e} \Dom(u):f_e(\vy)\neq \vzero\}$.

We use the same ideas from the Proof of Lemma~\ref{lemma:garity-acyclic-ub} but for each removed star $P$, we use Algorithm~\ref{algo:starAlgo-faq-1} to compute $P$. In particular, we show that computing $\FAQ_{\D, \cH, N, \cF}$ can be solved by computing the product $R'_{P} = \prod_{i = 2}^{k} R'_{\chi(v_i)}$. Note that this product can be computed on Steiner Tree packing like Theorem~\ref{thm:set-intersection}, where along with the set-intersection, the product of annotated values in the corresponding tuples is computed in $\tilde{O}(1)$ time as well. We perform two steps -- compute the intersection of tuples in each $R'_{\chi(v_i)}$ and multiply the annotated values if there is a tuple present in every $R'_{\chi(v_i)}$. It follows that $\FAQ_{\D, \cH, N, \cF} = R'_{P}$.

We describe our algorithm here. We perform a message-passing algorithm (upward pass) starting with a broadcast of the function $R_{\chi(v_1)}$ to all players in $G$. For every $i \in [2, k]$, the player containing $R_{\chi(v_i)}$ computes the up message $m_{v_i, v_1}$ it needs to send $v_1$ internally (Step~\ref{step:faq-internal-comp} in Algorithm~\ref{algo:starAlgo-faq-1}). Notice that the variables in the set $\Gamma(v_{i, 1}) = \{\chi(v_i) \setminus \chi(v_1)\} =\{ Z_1, \dots, Z_w \}$ are \emph{private} to the node $v_i$. In particular, the variables in $\Gamma(v_{i, 1})$ are not present anywhere in the remaining hypergraph. Further, note that $\{Z_1, \dots, Z_w \} \subseteq \{x_{\ell + 1}, \dots, x_{n}\}$. We `push down' the aggregations $\left(\underset{z_1 \in \Dom(Z_1)}{\tenssum{}^{(Z_1)}}, \dots, \underset{z_w \in \Dom(Z_w)}{\tenssum{}^{(Z_w)}}\right)$ inside every tuple in the function $R_{\chi(v_i)}$ since these variables do not occur anywhere in the remaining hypergraph. Further, the aggregations are computed on the annotated values of the relations as follows. In particular, for every tuple $\mathbf{t} \in \pi_{\chi(v_i) \cap \chi(v_1)} R_{\chi(v_i)}$, the tuple $\left(\mathbf{t},  \forall \mathbf{t}' \in R_{\chi(v_i)} : \underset{\pi_{Z_{1}}(\mathbf{t'})}{ \tenssum{}^{(Z_{1})}} \dots  \underset{\pi_{Z_{w}}(\mathbf{t'})}{\tenssum{}^{(Z_{w})}} f(\mathbf{t}') \text{ if } \pi_{\chi(v_i) \cap \chi(v_1)} \left(\mathbf{t}' \right)  = \mathbf{t} \right)$ is appended to the message $m_{v_i, v_1}$.

Then, only one player retains the original $R_{\chi(v_1)}$ (say the player containing $R_{\chi(v_2)}$) and all others store an identity map of $R_{\chi(v_1)}$ (with all entries set to a function value of $1$) to ensure we don't multiply $R_{\chi(v_1)}$ more than once. Finally, all the players containing $R_{\chi(v_i)} : i \in [2, k]$ compute $R'_{\chi(v_i)} = R_{\chi(v_1)} \tensprod{} m_{v_i, v_1}$ with their own version of $R_{\chi(v_1)}$ (either actual or the identity map) as follows. For every tuple $\mathbf{t} \in R_{\chi(v_1)}$ and for all tuples $\mathbf{t}' \in m_{v_i, v_1} : \pi_{\chi(v_i) \cap \chi(v_1)} (\mathbf{t}') = \pi_{\chi(v_i) \cap \chi(v_1)} (\mathbf{t})$, the tuple $\mathbf{t}'' = (\mathbf{t}'', f(\mathbf{t}) \cdot f(\mathbf{t}'))$ is appended to $R'_{v_i}$.

\begin{algorithm}
\caption{Algorithm for $\cT$ is a Star} \label{algo:starAlgo-faq-1}
\begin{algorithmic}[1]
\small
\State{\textbf{Input:} A star $P = (v_1, \dots, v_k) \in \cT$ and functions $\{R_{\chi(v_i)} : i \in [k]\}$. Note that $v_1$ is the center and the others are leaves.}
\State{\textbf{Output:} $R'_{P}$}
\State{The player containing $R_{\chi(v_1)}$ broadcasts it to all players in $G$.}
\State{For every $i \in [2, k]$, the player containing $R_{\chi(v_i)}$ internally computes the the Up Message from $v_i$ to $v_1$, 
\begin{align*}	
m_{v_i, v_1} =  \underset{z_{1} \in \Dom(Z_1)}{ \tenssum{}^{(Z_{1})}} \dots \underset{z_{w} \in \Dom(Z_{w})}{\tenssum{}^{(Z_{w})}}  R_{\chi(v_i)},
\end{align*}
where 
\begin{align*}
\Gamma(v_{i, 1}) =\chi(v_i) \setminus \chi(v_1) = \{Z_{1}, \dots, Z_{w}\} \subseteq \{X_{\ell + 1}, \dots, X_{n}\}.
\end{align*}
Note that all the $\tenssum{}^{(Z_{j})}$s for every $j \in [m]$ are computed on the values annotated with the tuples in the function. 
} \label{step:faq-internal-comp}
\If{$i = 2$}
\State{The player containing $R_{\chi(v_2)}$ computes $R'_{\chi(v_2)} = R_{\chi(v_1)} \tensprod{} m_{v_2, v_1}$ internally.} \Comment{This product is computed on the annotated values on the function and the message.}
\Else 
\State{Converts $R_{\chi(v_1)}$ to an identity map i.e., all entries in it are assigned a value of $1$.}
\State{The player containing $R_{\chi(v_i)}$ computes $R'_{\chi(v_i)} = R_{\chi(v_1)} \tensprod{} m_{v_i, v_1}$ internally.} \Comment{This product is computed on the annotated values on the function and the message.}
\EndIf
\State{$R'_{P} = \tensprod{}_{i = 2}^{k} R'_{\chi(v_i)}$ } \Comment{This product is computed on a Steiner Tree packing like Theorem~\ref{thm:set-intersection}, where along with the set-intersection, the product of annotated values in the corresponding tuples is computed as well.} \label{step:faq-steiner}
\State{\Return{$R'_{P}$}}
\end{algorithmic}
\end{algorithm}

Note that all the $R'_{\chi(v_i)}$s are computed on the same attribute set $\chi(v_i)$ and the annotated tuples in each $R'_{\chi(v_i)}$ can be multiplied in constant time. In particular, this implies Step~\ref{step:faq-steiner} of our Algorithm can be computed on a Steiner Tree packing (like Theorem~\ref{thm:set-intersection}) resulting in an upper bound of 
\begin{align*}
\tO\left(\min_{\Delta \in [|V|]}\left(\frac{N \cdot r}{\ST(G, K, \Delta)} + \Delta \right) \right).
\end{align*}
We can now repeat the same arguments from the proof of Lemma~\ref{lemma:garity-acyclic-ub} (as stated earlier) until the root of $\cT$, which gives us the first term in the required upper bound. We can then apply the \emph{naive} protocol on the root of $\cT$ (that contains $\Co(\cH))$, solving it in $\tau_{\MCF}(G, K, n_2(\cH) \cdot d \cdot r \cdot N)$ rounds (using Lemma~\ref{lem:naive}). In total, we have the desired upper bound for computing $\FAQ_{\D, \cH, N, \cF}$. Note that our choices of $f$ and operator sequence $(\tenssum{}^{(i)})_{\ell < i \le n}$ were arbitrary and thus, our results hold for general $\FAQ_{\D, \cH, N, \cF}$. Finally, since our choice of a GHD was arbitrary, we have $y(\cT) \le y(\cH)$. This completes the proof.

\subsection{Lower Bound for General FAQs} \label{sec:faqlb2}
Before proving the lower bound, we state some followup observations based on the proof of Lemma~\ref{lemma:acyclic-lb-1}, which will be used crucially in the proof. 
\subsubsection{Existing Results}
Note that in the proof of Lemma~\ref{lemma:acyclic-lb-1}, we invoked an existing lower bound on Theorem~\ref{thm:jks} for $\TRB_{\frac{y(\cT)}{r}, N}(\hat{S}, \hat{T})$. We start by remarking the following.
\begin{rem} \label{rem:tribes-product} [From Section $2.1$ in~\cite{topology-1}] 
The lower bound on $\TRB$ is obtained on a product distribution $\hat{\D}$ on $\frac{y(\cT)}{r}$ variables, where for every pair of sets $(S_j, T_j)$ in the $\TRB_{\frac{y(\cT)}{r}, N}(\hat{S}, \hat{T})$ instance, we have
\[ |S_ j\cap T_j| \le 1,\]
where $j \in \left[\frac{y(\cT)}{r} \right]$.
\end{rem}
Note that this implies the following based on our lower bound arguments.
\begin{rem} \label{rem:join-op-size}
For all our $\BCQ_{\cH, N}$ instances, we have $|\Join_{e \in E(\cH)} R_{e}| \le 1$. 
\end{rem}

We are now ready to argue the lower bound for general FAQs, which follows from the fact that our hard $\BCQ_{\cH, N}$ instance for a $d$-degenerate $\cH$ is a hard $\FAQ_{\D, \cH, N, \cF}$ instance for the operator set $\left(\tenssum{}^{(\ell + 1)}, \tenssum{}^{(n)}\right)$ with $\tenssum{}^{(i)} = \tensprod{}$ or $(\D, \tenssum{}^{(i)}, \tensprod{})$ forms a commutative semiring with the same additive identity $\mathbf{0}$ and multiplicative identity $\mathbf{1}$ for every $\ell < i \le n$. 

We argue $\cR \left(\FAQ_{\D, \cH, N, \cF}, G, K \right) \ge \cR (\BCQ_{\cH, N}, G, K)$ and the above result follows. 
We start with the $\BCQ_{\cH,N}$ instance from Section~\ref{sec:garity-lb-general}. 
We construct a $\FAQ_{\D, \cH, N, \cF}$ instance from a given $\BCQ_{\cH,N}$ instance as follows. 
For each function $R_{e} : e \in E(\cH)$, we apply the following function $f$ on every tuple $\mathbf{t} \in \prod_{u \in v(e)} \Dom(u)$: we set $f(\mathbf{t}) = 1$ if $\mathbf{t} \in R_{e}$ and $0$ otherwise. 
Note that this implies we can define functions of the form $R_{e} = \{ (\mathbf{t}, 1)) : \mathbf{t} \in R_{v(e)}\}$ for every $e \in E(\cH)$. We now have a $\FAQ_{\D, \cH, N, \cF}$ instance of the form
\begin{equation}
\FAQ_{\D, \cH, N, \cF} = \underset{x_{\ell + 1} \in \Dom(X_{\ell + 1})}{\tenssum{}^{(\ell + 1)}}\dots \underset{x_{n} \in \Dom(X_{n})}{\tenssum{}^{(n)}}  \underset{S \in \cE} {\tensprod{}} \phi_{S}(\vx_{S}).
\end{equation}
Given this setup, we claim that $\BCQ_{\cH, N}$ is $1$ iff $\FAQ_{\D,\cH, N,\cF}$ is $1$ and $0$ otherwise. 
To see why this is true, notice that in all our hard instances of $\BCQ_{\cH,N}$, the corresponding join output $|\Join_{e \in E(\cH)} R_{e}| \le 1$ (from Remark~\ref{rem:join-op-size}). 
As a result, we can apply the sequence of operators $(\tenssum{}^{(i)})_{n \le  i < \ell}$ one-by-one from right to left. 
If $\tenssum{}^{(i)} = \tensprod{}$ for $n \le i < \ell$ , applying it on at most one value does not make any difference. 
Otherwise, since all commutative semirings of the form $(\D, \tenssum{}^{(i)}, \tensprod{})$ have the same 
additive identity $\mathbf{0}$ and multiplicative identity $\mathbf{1}$, we can conclude that
$\BCQ_{\cH, N}=1$ iff $\FAQ_{\D,\cH, N,\cF}=1$. 
Note that the choices of operators $(\tenssum{}^{(i)})_{\ell < i \le n}$ and $\D$ was arbitrary.
Thus, we have $\cR \left(\FAQ_{\D, \cH, N, \cF}, G, K \right) \ge \cR(\BCQ_{\cH, N}, G, K)$. 

\subsection{Restriction on Choice of $\cF$} \label{sec:free-vars}
Recall that any (hyper)graph $\cH$ can be decomposed into a core $\Co(\cH)$ and a forest $\Fo(\cH)$ using Construction~\ref{lemma:forest-core}. 
We would like to mention here that our upper and lower bounds in this paper hold only for the case when $\cF \subseteq V(\Co(\cH))$. 
For our upper bounds, this is necessary because we apply different algorithms on $\Fo(\cH)$ and $\Co(\cH)$.
For the lower bounds, we once again deal with $\Co(\cH)$ and $\Fo(\cH)$ independently and sum the bounds obtain from either of them.
We believe that expanding the choices of $\cF$ needs new techniques for both our upper and lower bounds and leave its investigation for future work.
\subsection{Hash-based Split of Relations} \label{app:hash-based-split}
In this section, we address the assumption that input functions in $\cH$ are completely assigned to players in $G$. We prove upper and lower bounds when the input relations are split based on certain kind of hashes. As a by-product, our lower bounds techniques help us prove bounds when input functions are not split but randomly assigned to players in $G$ (addressing the assumption of worst-case assignment of input functions to players in $G$). 

We define our setup in detail here.
\subsubsection{Our Setup} \label{faqsetup2}
We first state the condition we need on hashes used to split relations that is sufficient for our bounds. At a high level, given a rooted GHD $\cT$ and a hash family $\tilde{H}$, we say that $\tilde{H}$ is consistent with $\cT$ if for each parent-child pair $(u, v) \in \cT$, we have that $h_{\chi(v)}(\mathbf{t})$ is the same for every subset of tuples $\mathbf{t} \in R_{\chi(v)}$ that have the same projection to variables in $\chi(u) \cap \chi(v)$. Then, we state some realistic scenarios where these conditions are satisfied.
\begin{defn} \label{def:hash-consistent}
Given a hypergraph $\cH$ and GHD $\cT$ such that the root of $\cT$ is $\Co(\cH)$, we say a family of hash functions $\tilde{H}=\{h_e:\prod_{v\in e} \Dom(v)\to K|e\in\cE\}$ is {\em consistent} with $\cT$ and $K$ if the following holds. Let $r'$ be the root of $\cT$. If $e\subseteq \chi(r')$ (i.e., $e$ is assigned to the root of $\cT$), then $h_e$ can be arbitrary. Now consider a non-root node $v$ in $\cT$ and let $u$ be its parent.

First, we consider the projection $S_{u,v} =  \pi_{\chi(u) \cap \chi(v)} \left(R_{\chi(v)} \right)$. Then, for every tuple $\mathbf{s} \in S_{u,v}$, we have that $h_{\chi(v)}(\mathbf{t})$ is the same for {\em every} $\{\mathbf{t} \in R_{\chi(v)}:  \pi_{\chi(u) \cap \chi(v)} \left(\mathbf{t}\right) = \mathbf{s}\}$.

Further, we say the set of relations $\{R_e\}_{e \in\cE}$ are {\em split according to} $\tilde{H}$, if for any $e\in \cE$ and $\mathbf{t}\in R_e$, the tuple $\mathbf{t}$ is assigned to player $h_e(\mathbf{t})$.
\end{defn}
It turns out that if tuples in the relations are split according to a family of hash functions as in the above definition, then we can generalize Algorithm~\ref{algo:starAlgo-faq-1} to this case. 

Next, we observe that our condition on a family of hash functions being consistent with a GHD $\cT$ and $K$ is reasonable. In other words, we are assuming that all attributes of $R_e$ for every $e \in \cE$ are stored in a {\em global variable elimination order} that is compatible with $\cT$ (this follows from Definition~\ref{def:hash-consistent}). In particular, this implies for any non-root $v$ in $\cT$ and its parent $u$, we have that $\chi(u)\cap\chi(v)$ is a {\em prefix} of $\chi(v)$ according to this variable elimination order. This assumption on the variables in $R_e$ being stored in the variable elimination order of $\cT$ has been made before for GHD-based algorithms used to solve $\FAQ$~\cite{faq,ajar}. Further, `bit-map based' \cite{bitmap-1} hash functions $h_e$ satisfy the consistency property in Definition~\ref{def:hash-consistent}. We discuss a simple bit-map type hash function here. The hash is computed on the product space $\prod_{v \in \cV} \Dom(v)$ based on a global variable elimination order (that $\cT$ corresponds to) as discussed above. Consider any parent-child pair $(u, v) \in \cT$ and $h_{\chi(v)}(\mathbf{t})$ is computed as follows -- first, we compute a tuple $\mathbf{t}'$ based on $\mathbf{t}$, where we set $t'_{v} = t_{v}$ if $v \in \chi(v) \cap \chi(u)$ and $t'_{v} = 0$ otherwise. We now set $h_{\chi(v)}(\mathbf{t}) = h_{\chi(v)}(\mathbf{t}')$ i.e., the hash is always computed on values of the variables $\chi(v) \cap \chi(u)$. Note that our hash function satisfies Definition~\ref{def:hash-consistent}, as required.

Finally, we note that if the relations themselves are free of {\em skew} (which is an assumption made in~\cite{BKS13}), then any consistent family of hash functions will distribute the tuples in a relation (near) equally among players in $K$. In particular, when relations are free of skew, we could treat each relation with arity $r$ as a $r$-dimensional matching. Note that this implies for each relation $R_e$, each variable has a unique value in each tuple $\mathbf{t} \in R_e$. Thus, we have that $h_{e}(\mathbf{t})$ is unique for every $\mathbf{t} \in R_e$ and the number of tuples each player receives is roughly the same (i.e., $N$). 
\subsubsection{Main Theorem}
\begin{thm} \label{thm:garitypar}
Fix an arbitrary $G$, subset of players $K$, any $\cF \subseteq V(\Co(\cH))$ and $d$-degenerate hypergraphs $\cH$ with arity at most $r$. Further, assume that the set of relations are split according to hash family $\tilde{H}$  that is consistent with $\cT$ (where $y(\cT)=y(\cH)$) and $K$. Then, we have
\begin{equation} \label{eq:garityparub}
\cD (\FAQ_{\D, \cH, N, \cF}, G, K) = \tO\left( y(\cH) \cdot \min_{\Delta \in [|V|]}\left(\frac{N \cdot \left(r + \log(|K|) \right)}{\ST(G, K, \Delta)} + |K|\cdot\Delta \right) + \tau_{\MCF}\left(G, K, n_2(\cH) \cdot d \cdot r \cdot N\right)\right).
\end{equation}
Further, if $\tilde{H}$ is is a random hash family where $h_e$ for $e\in \cE$ are chosen independently and uniformly (conditioned on $\tilde{H}$ being consistent), we have (with high probability over the randomness in $\tilde{H}$):
\begin{equation} \label{eq:random-split-lb-1}
\cR (\FAQ_{\D,\cH, N,\cF}, G, K) \ge \tOm\left( \frac{y(\cH) \cdot N}{r \cdot \gamma(G,K)}+ \frac{n_2(\cH) \cdot N}{d \cdot  r \cdot \gamma(G,K)} \right),
\end{equation}
where $\gamma(G,K)$ is the minimum over all cuts $(A,B)$ separating $K$ of the quantity $\frac{|E(A,B)||K|^2}{(\min(|A|,|B|))^2}$, where $E(A,B)$ is the set of edges crossing the cut. Here, $y(\cH)$ and $n_2(\cH)$ are defined as in Definition~\ref{defn:garity-notation}.
\end{thm}
We note that when $G$ is a line $\gamma(G,K)$ is attained at the cut that equally cuts $K$ into two parts and since $|E(A,B)|=1$, we get that $\gamma(G,K)=O(1)$. In other cases, especially when $\gamma(G, K)$ is non-constant, our upper and lower bounds are tight up to a factor of roughly $r^{2} \log(|K|)  \frac{\gamma(G, K)}{\ST(G, K \Delta)}$ for the acyclic forest and $d^{2} r^{2} \frac{\gamma(G, K)}{\MC(G, K)}$ for the core. Before going into the upper bounds, we would like to note here that for simple graphs $\cH$, we can overcome the factor of $d$ in the lower bound. In particular, we can use ideas from the proof of Theorem~\ref{thm:aritytwo2}, to argue that the core $\cH$ has an independent set of size $\Omega(\max(y(\cH), n_{2}(\cH)))$. This leads to the following corollary.
\begin{cor} \label{cor:split-2}
For arbitrary $G$, subset of players $K$ and any simple $d$-degenerate graph $\cH$. Further, assume that the set of relations are split according to hash family $\tilde{H}$ that is consistent with $\cT$ (where $y(\cT)=y(\cH)$) and $K$. Then, we have
\begin{align*}
\mathcal{D} (\FAQ_{\D, \cH, N, \cF}, G, K) = \tO\left( y(\cH) \cdot \min_{\Delta \in [|V|]}\left(\frac{N \cdot \log(|K|)}{\ST(G, K, \Delta)} + |K|\cdot\Delta \right) + \tau_{\MCF}\left(G, K, n_2(\cH) \cdot d \cdot N\right)\right).
\end{align*}
and 
\begin{align*}
\cR (\FAQ_{\D,\cH, N,\cF}, G, K) \ge \tOm\left( \frac{y(\cH) \cdot N}{\gamma(G,K)}+ \frac{n_2(\cH) \cdot N}{\gamma(G,K)} \right).
\end{align*}
\end{cor}
\subsubsection{Upper Bound} \label{sec:faqub3}
The upper bounds follows from a slight modification of our algorithm to compute $\FAQ_{\D, \cH, N, \cF}$ when relations are not partitioned (Section~\ref{sec:faqub2}).

We present a proof sketch here. The idea is very similar to the proof in Section~\ref{sec:faqub2}. To that end, we modify Algorithm~\ref{algo:starAlgo-faq-1} as follows. Instead of broadcasting $R_{\chi(v_1)}$, the Steps $4$ to $10$ are applied individually on each tuple in $R_{\chi(v_1)}$ as follows. For simplicity of notation, we assume all the $\tenssum{}^{(i)}$s are $\sum$ and $\tensprod{}$ is $\prod$. 

We start with the case when $\cT$ is a star with $v_1$ as the center and $v_2,\dots,v_{|P|}$ as the leaves. Let's fix a tuple $\mathbf{t} \in R_{\chi(v_1)}$. It is broadcast to all players in $K$ along with a counter $c_{\mathbf{t}} \in [0,|K|]$.\footnote{This can be done via a Steiner tree packing with $\min_{\Delta \in [|V|]}\left(\frac{N \cdot r}{\ST(G, K, \Delta)}+\Delta\right)$ rounds.} Initially, we set $c_{\mathbf{t}} = 0$. For any player $\ell\in K$ and $j\in [2,|P|]$, define $R^{(\ell)}_{\chi(v_j)}$ to be the set of tuples in $R_{\chi(v_j)}$ mapped to $\ell$ by $h_{\chi(v_j)}$.  
Upon receiving $(\mathbf{t},c_{\mathbf{t}})$,
player $\ell$ checks if there exists a tuple $\mathbf{t}' \in R^{(\ell)}_{\chi(v_j)}$ such that $\pi_{\chi(v_j) \cap \chi(v_1)} (\mathbf{t'}) = \pi_{\chi(v_j) \cap \chi(v_1)}(\mathbf{t})$. If so, then player $i$ increments $c$ by one, (internally) computes the sum
\begin{equation} \label{eq:f-run-sum}
\sum_{\mathbf{t}' \in R^{\ell}_{v_j}: \pi_{\chi(v_j) \cap \chi(v_1)} (\mathbf{t'}) = \pi_{\chi(v_j) \cap \chi(v_1)}(\mathbf{t})} f_{\chi(v_j)}(\mathbf{t}')
\end{equation}
so that it can contribute to the running product:
\begin{equation} \label{eq:f-run-prod}
v(\mathbf{t})=f_{\chi(v_1)}(\mathbf{t})\cdot \prod_{j=2}^{|P|}\left(\sum_{\mathbf{t}' \in R^{\ell}_{v_j}: \pi_{\chi(v_j) \cap \chi(v_1)} (\mathbf{t'}) = \pi_{\chi(v_j) \cap \chi(v_1)}(\mathbf{t})} f_{\chi(v_j)}(\mathbf{t}')\right),
\end{equation}
which is the value corresponding to $\mathbf{t}$ for the corresponding $\FAQS$ query.
Note that the sums and products are computed on the values corresponding to the tuples as in Section~\ref{sec:faqub2}. At the end of the procedure, if $c_{\mathbf{t}} = |K|$, then $(\mathbf{t},v(\mathbf{t}))$ is added to the result $R'_{P}$, which is computed at the center of the star and we continue. Note that we can repeat the above procedure for each star in $\cT$ until we reach the root as we did in Section~\ref{sec:faqub2}.

Next, we argue the correctness of the above procedure, which follows from the consistency property defined in Section~\ref{faqsetup2}. Recall that the input relations are split based on a hash family consistent with $\cT$, we have for each parent-child pair $(u, v) \in \cT$, the hash function $h_{\chi(v)}(\mathbf{t})$ is the same for every subset of tuples $\mathbf{t} \in R_{\chi(v)}$ that have the same projection to variables in $\chi(u) \cap \chi(v)$. Our algorithm leverages this fact and ensures that each node (internally) computes the running sum~\eqref{eq:f-run-sum} over all such tuples. Further, the products are computed across nodes on a Steiner Tree packing as in Theorem~\ref{thm:set-intersection},\footnote{The algorithm for Theorem~\ref{thm:set-intersection} first thinks of its input as a vector and computes its component-wise AND. These vectors are subdivided among the edge disjoint Steiner trees and then the component-wise AND of the smaller vectors is done in a bottom-up fashion in a dedicated Steiner tree from the packing. In the current case, we want to compute component-wise product and we can just run the set intersection algorithms where instead of computing component-wise AND we use component-wise product.} resulting in an upper bound of  \[\tO\left( y(\cH) \cdot \min_{\Delta \in [|V|]}\left(\frac{N \cdot (r+\log(|K|))}{\ST(G, K, \Delta)} + |K| \cdot \Delta \right) \right),\]
where the $\log(|K|)$ additive term is to keep track of the counter $c_{\mathbf{t}}$ (and in the final bound is absorbed into the $\tO(\cdot)$).

To complete the proof, we use the following \emph{trivial protocol} on the root of $\cT$. In particular, any one designated player should still receive all the partitions from all relations. Since each player has $|K|$ partitions of all the remaining relations, the round complexity is given by $\tau_{\MCF}\left(G, K, n_2(\cH) \cdot d \cdot r \cdot N \right)$ (using Lemma~\ref{lem:naive}).
\subsubsection{Lower Bound} \label{sec:faqlb3}
The lower bound follows similarly from ideas in Section~\ref{sec:faqlb2}. 

%We will use the following lemma repeatedly in our arguments.
%\begin{lemma} [Additive Chernoff Bound]
%Let $X_1, \dots, X_{N}$ denote i.i.d random variables such that $0 \le X_i \le 1$ for every $1 \le i \le N$. Let $S = \frac{1}{N} \sum_{i = 1}^{N} X_i$ denote their mean and $E[S]$ their expected mean, where 
%\end{lemma}

We have already shown in Section~\ref{sec:faqlb2} that our hard $\BCQ_{\cH,N}$ for a $d$-degenerate $\cH$ is a hard instance for $\FAQ_{\D, \cH, N, \cF}$ as well. The only difference is that we cannot apply lower bounds on worst-case assignment directly anymore since the relations are partitioned now. We address the issue here. 

Similar to the proof of Lemma~\ref{lemma:acyclic-lb-1}, we consider an arbitrary cut $(A, B)$ of $G$ that separates $K$, where $A$ and $B$ denote the set of vertices in each partition ($A \cup B = V(G)$). For simplicity, we assume $|A| = |B|$ and later show how to get around this restriction. We recall the notation in the proof of Lemma~\ref{lemma:acyclic-lb-1}, which we will use here as well. Let $q_{\cH, \hat{S}, \hat{T}}$ be the query corresponding to a given $\TRB$ instance $\TRB_{\frac{y(\cH)}{r}, N}(\hat{S}, \hat{T})$, where the relations $\{R_{S_i}\}_{p_i \in I}$ and $\{R_{T_i}\}_{p_i \in I}$ are assigned to vertices in $A$ and $B$ respectively. The remaining relations in $q_{\cH, \hat{S}, \hat{T}}$ are assigned arbitrarily and we have $|I| = |\Fo(\cH)| = \frac{y(\cH)}{r}$.

Note that in the current scenario, we partition all relations, which includes $\{R_{S_i}\}_{p_i\in I}$ and $\{R_{T_i}\}_{p_i \in I}$. Since all the set pairs $(R_{S_i}, R_{T_i})$ for every $p_i \in I$ in the $\TRB$ instance are independent, we start by considering one such pair. In particular, let's consider $R_{S_i}$. We first note that the way we have defined $R_{S_i}$ every prefix has exactly {\em one} extension. Since $\tilde{H}$ is chosen so that the individual hash functions are independent and uniformly distributed (conditioned on $\tilde{H}$ being consistent), in this particular case because of the aforementioned property of the prefixes, each hash function is a uniformly random hash function. Thus, any tuple in $R_{S_i}$ is uniformly distributed among the players in $K$.

We now see how the tuples in $R_{S_i}$ and $R_{T_i}$ are split. In particular, since any tuple in $R_{S_i}$ (or $T_{S_i}$) is assigned uniformly to players in $K$, each tuple $\mathbf{t} \in R_{S_i}$ is assigned to either $A$ or $B$ with probability $\frac{1}{2}$ (since $|A| = |B|$). Likewise for tuples in $R_{T_i}$ (since the hash functions for different relations are independent as shown above). 
More formally, we write `$\mathbf{t} \in A$' if a given tuple $\mathbf{t}$ is assigned to vertices in $A$ (similarly for vertices in $B$). We have
\begin{align*}
\PR[\mathbf{t} \in A] = \PR[\mathbf{t} \in B] = \PR[\mathbf{t}' \in A] = \PR[\mathbf{t}' \in B] = \frac{1}{2},
\end{align*}
where $\mathbf{t} \in R_{S_i}$ and $\mathbf{t}' \in R_{T_i}$. Since the choice of $\mathbf{t}$ and $\mathbf{t}'$ are independent, we have 
\begin{align*}
\PR[(\mathbf{t} \in A) \wedge (\mathbf{t}' \in B)] = \PR[(\mathbf{t}' \in A) \wedge (\mathbf{t} \in B)] = \frac{1}{4}.
\end{align*}
Thus, in expectation the total number of tuples that satisfy the above property in $R_{S_o}$ and $R_{T_o}$ is $\frac{N}{4}$.
Moreover, this number is at least $\frac{N}{8}$ with probability $1-2^{-\Omega(N)}$ (by a direct application of the Chernoff bound).  
Abusing notation, let $R_{S_i}$ and $R_{T_i}$ denote the tuples of the original $R_{S_i}$ and $R_{T_i}$ that were split between $A$ and $B$, and assume that these relations are exactly of size $\frac{N}{8}$.

In addition to the above, using Remark~\ref{rem:tribes-product}, the distribution $\hat{\D}$ on $(\hat{S},\hat{T})$ we use in a blackbox manner has a property that $|S_i \cap T_i| \le 1$, i.e., there is at most only one value $a' \in p_i$ such that $S_i \cap T_i = a'$. Thus, we need this particular tuple to always be split (i.e., one copy goes to $A$ and other goes to $B$). Otherwise, we claim that there exists a protocol with $\frac{y(\cH)}{r}$ rounds. In particular, the transcript of the protocol will have only these tuples (s.t. $S_i \cap T_i = a'$), which is sufficient for computing $\TRB_{\frac{y(\cH)}{r}, N}(\hat{S}, \hat{T})$ and takes only $\frac{y(\cH)}{r}$ rounds as claimed earlier.

It turns out that we can ensure this split by conditioning our expectation on the event that the tuple containing $a'$ in both $R_{S_i}$ and $R_{T_i}$ is always split. Further, it is easy to see that $(1)$ $a'$ is split with probability at least $\frac{1}{4}$ and $(2)$ Conditioned on being split, $a'$ is still uniformly distributed over the $\frac{N}{8}$ locations in $R_{S_i}$ and $R_{T_i}$. By another application of the Chernoff bound, for at least $\frac{1}{8}$ of the $(S_i,T_i)$ pairs this special value $a'$ is split. In other words, we now have a smaller $\TRB$ instance across the $(A,B)$ cut, where we have $\frac{1}{8}$-th the number of set disjointness instances (let these be indexed by $I'$ with $|I'|=\frac{|I|}{8}$) where each set disjointness instance is $\frac{1}{8}$th the original size.
The other relations in $q_{\cH, \hat{S}, \hat{T}}$ can be partitioned randomly and assigned arbitrarily across $\MC(G, K)$.

We now consider the induced $\TRB$ function based on $I'$. Note that we have argued that with high probability, the number of set disjointness instances $I'$ is $\Omega\left(\frac{y(\cH)}{r}\right)$ (recall that $|\Fo(\cH)| = \frac{y(\cH)}{r}$). In particular, we have argued that the relations $\{R_{S_i}\}_{p_i \in I'}$ are assigned to vertices in $A$ and $\{R_{T_i}\}_{p_i \in I'}$ to $B$. In particular, Alice gets the $\frac{N}{8}$ tuples in the sets $\{S_i\}_{p_i\in I'}$ (corresponding to $R_{S_i}$) assigned to $A$ and Bob gets the $\frac{N}{8}$ tuples in the sets $\{T_i\}_{p_i \in I'}$ (corresponding to $R_{T_i}$) assigned to vertices in $B$ (ignoring the additional relations). It follows that if there exists a $z$ round protocol on $G$, then there is an $O(z\cdot |E(A,B)|)$ two-party protocol (see proof of Lemma~\ref{lemma:tree-lb-1} for a detailed discussion). Since $z \cdot |E(A,B)|$ is lower bounded from Theorem~\ref{thm:jks} by $\tilde{\Omega}\left( \frac{y(\cT) \cdot N}{r}\right)$, we have a lower bound of $\tilde{\Omega}\left(\frac{y(\cH) \cdot N}{r \cdot |E(A,B)|}\right)$.

Finally, we remove the restriction that $|A| = |B|$. More generally, instead of a uniform probability of $\frac{1}{4}$ of two tuples in $S_i$ and $T_i$ being split, we would have a probability of 
\begin{align*}
\frac{\min(|A|,|B|)\cdot \max(|A|,|B|)}{|K|^2}\ge \frac{\min(|A|,|B|)}{2\cdot |K|}.
\end{align*}
Generalizing the above argument where we replace the $\frac{1}{4}$ with the above probability, we get that if $z$ is the round complexity of a protocol to compute $q_{\cH, \hat{S}, \hat{T}}$, then we have
\begin{align*}
z\cdot |E(A,B)|\ge \tilde{\Omega}\left(\frac{y(\cH)\cdot (\min(|A|,|B|))^2}{r\cdot |K|^2}\cdot N\right).
\end{align*}
Our definition of $\gamma(G,K) = \frac{|E(A, B)| \cdot |K|^2}{(\min(|A|,|B|))^2}$ implies
$$z\ge \tilde{\Omega}\left(\frac{y(\cH) \cdot N}{r \cdot \gamma(G,K) }\right).$$ Similar arguments can be applied to other $\TRB$ instances, completing the proof.  

We would like to state here that similar ideas in the proof above can be used to obtain a lower bound of the form $\tilde{\Omega}\left(\frac{y(\cH) \cdot N}{r \cdot \gamma(G, K)}\right)$ for the case when the relations $\{R_{S_i}\}_{p_i \in I}$ and $\{R_{T_i}\}_{p_i \in I}$ are not split but only randomly assigned to players in $K$ (instead of a worst-case assignment). We briefly sketch the key ideas here. We start with the same assumption that $|A| = |B|$ and for each $p_i \in I$, we would like to ensure that $R_{S_i}$ and $R_{T_i}$ are not given to the same player. We can perform a similar analysis as above to show that at least $\frac{1}{2}$ of the number of set disjointness instances (i.e., $\frac{|I|}{2}$) will be distributed in the required way. Finally, we can remove the restriction on $|A| = |B|$, obtaining a lower bound of $\tilde{\Omega}\left(\frac{y(\cH) \cdot N}{r \cdot \gamma(G, K)}\right)$, as required. Note that this removes the assumption in Model~\ref{model:our_model} of worst-case assignment of functions in $\cH$ to players in $G$ at the cost of a (slightly) weaker lower bound.	
\section{Min-entropy of Matrix-Vector Multiplication} \label{app:matrix-vect}
We now prove Theorem~\ref{thm:CG-ext} and start by restating it.
\begin{thm}
Let the constant $\gamma>0$ be small enough. Let $\vx\in\F_2^N$, $\vA\in\F_2^{N\times N}$ and $\vY$ be random variables such that for every $y \in \ssupp(\vY)$, $\vx$ and $\vA$ are independent conditioned on $\vY = y$. Moreover for some reals $\eps_1,\eps_2 \geq 0$, we have 
\begin{align*}
\hinfe{\vA|\vY}{\eps_1} \ge (1 - \gamma) \cdot N^{2} \text{ and } \hinfe{\vx|\vY}{\eps_2} \ge \alpha \cdot N,
\end{align*}
where $\alpha\eqdef  3\gamma +\sqrt{2\gamma} + h(\sqrt{2\gamma})$ and $h(p)=-p\log_2{p}-(1-p)\log_2(1-p)$ for any $p : 0 < p \le 1$.
Then, we have 
\begin{align*}
\hinfe{\vA\vx|\vY}{\eps_1+\eps_2+2^{-\Omega(\gamma N)}}\ge \left(1-\sqrt{2\gamma}\right)\cdot N.
\end{align*}
\end{thm}
We state some preliminaries needed for our proof.
\subsection{Preliminaries}
\subsubsection{Min-entropy}
To be consistent with usual terminology used in the pseudorandomness literature~\cite{salil-book}, we will use the following equivalent definition as~\eqref{def:h-inf-eps-1}:
\begin{align*}
\hinfe{X}{\epsilon} = \sup_{X'\sim_\epsilon X} \hinf{X'},
\end{align*}
where $X'\sim_\epsilon X$ is over all distributions $X'$ that have statistical distance at most $\epsilon$ from $X$. Notice that in the expression, we \emph{do not} require $\ssupp(X') \subseteq \ssupp(X)$, neither do we restrict the domain of $X'$. 
%The equivalence of the two definitions is easy to see as in the latter definition, we can form the distribution $X'$ by moving $\epsilon$ probability from $X$ and distribute it evenly on sufficiently many newly introduced elements outside $\ssupp(X)$.

The following result will be useful in our analysis.
\begin{prop} \label{prop:vard}
Let $\cD_1$ and $\cD_2$ be two distributions. Let $f$ be a deterministic function on the $\ssupp(\cD_1) \cup \ssupp(\cD_2)$. If $\vard{\eps}{\cD_1}{\cD_2}$, 
then
\begin{align*}
\vard{\eps}{f\left(\cD_1\right)}{f\left(\cD_2\right)}.
\end{align*}
\end{prop}
For any event $\cE$, we will use $\ind{\cE}$ to denote the $0/1$-indicator variable for $\cE$. We will use $\cU_m$ to denote the uniform distribution on $\bits^m$, where $\bits$ is the finite field of two elements.

\subsubsection{Matrices and Vectors}
We will deal with vectors $\vx\in\bits^n$ in this section as well as matrices $\vA\in \bits^{m\times n}$ for $m\le n$.\footnote{We are changing notation only for this section of the appendix. We have used $n$ to denote the number of variables in a query but in this section, we will use it to define the dimension of vector and matrices as is the norm in linear algebra.} All vectors by default are column vectors and all indices start from $1$. We will use $\vA_i$ to denote the $i$th row of $\vA$ and for any subset $S\subseteq [m]$, $\vA_{S}$ denotes the submatrix indexed by the rows of $\vA$ indexed by $S$. Given two vectors $\vx,\vy\in \bits^n$, we will use $\ip{\vx}{\vy}$ to denote their inner product over $\F_2$.

\subsection{Proof of Theorem~\ref{thm:CG-ext}} \label{sec:matrix-vec-main}
We will argue the general version of Theorem~\ref{thm:CG-ext}, which follows from the result below for $n=m=N$.
\begin{thm} \label{thm:CG-ext-app}
Let the constant $\gamma$ be small enough. Let $\vx\in\bits^n$ and $\vA\in\bits^{m\times n}$ (for $m\le n$) be distributed such that there exists a random variable $\vY$ where for every $y\in\supp{\vY}$, conditioned on $\vY=y$, $\vx$ and $\vA$ are independent. Further, for some reals $\eps_1,\eps_2\geq 0$,
\begin{align*}
\hinfe{\vA|\vY}{\eps_1}\ge (1-\gamma)mn,
\end{align*}
and
\begin{align*}
\hinfe{\vx|\vY}{\eps_2}\ge \alpha\cdot n,
\end{align*}
where
\begin{align*}
\alpha\eqdef  3\gamma +\sqrt{2\gamma} + h(\sqrt{2\gamma}).
\end{align*}
Then, we have
\begin{align*}
\hinfe{\vA\vx|\vY}{\eps_1 + \eps_2 + 2^{-\Omega(\gamma m)}} \geq \left(1-\sqrt{2\gamma}\right)\cdot m.
\end{align*}
\end{thm}

In the rest of this section, we will argue Theorem~\ref{thm:CG-ext-app}. 
\begin{enumerate}
\item First, we prove the theorem for the case where $\eps_1 = \eps_2 = 0$ and $\vY$ is deterministic. This is done as follows.
\begin{enumerate}[label=(1\alph*)]
\item We will argue that $\vA$ has high enough min-entropy in `most' rows. This will define what is called a {\em block-source}. The details are in Section~\ref{sec:blk-src}.
\item We then argue that any $\vA$ that is a sufficiently good block source has the following property: the inner product $\ip{\vA_i}{\vx}$ is close to a random bit as long as $\vx$ has min-entropy at least $\alpha\cdot n$. Further, we can make this argument for each row with high enough min-entropy by only adding up the `closeness' for each such row. The details are in Section~\ref{sec:CG-sse}.
\end{enumerate}
\item Then, in Section~\ref{app:together}, we remove the assumption that $\eps_1 = \eps_2 = 0$ and $\vY$ is deterministic. 
\end{enumerate}

\subsection{$\vA$ is a good enough block source} \label{sec:blk-src}
We begin with the definition of a block-source\footnote{This is a more specific definition than the usual definition. We go with the more specific definition since it suffices for our purposes.}.
\begin{defn} \label{def:blk-src}
A random variable $\vA'$ over $\bits^{m\times n}$ is an $(\eta,n')$-block source for some $\eta\in [0,1]$ and $n'\le n$ if there exists a subset 
\begin{align*}
S\subseteq [m]\text{ with } |S| \leq \eta m
\end{align*}
such that for every ${A}\in\supp{\vA'}$ and every $i\not\in S$, we have 
\begin{align*}
\minent{\vA'_i| \vA'_{[i-1]}={A}_{[i-1]}}\ge n'.
\end{align*}
\end{defn}
We remark that in the above definition we do condition on all rows in $[i-1]$ (and not just $[i-1]\setminus S$).

Ideally, we would like to argue that our ${\vA}$ is a $(\gamma, (1-\gamma)n)$-block source. We will instead argue something slightly weaker, which is nonetheless powerful enough to help us prove Theorem~\ref{thm:CG-ext-app}. In particular, we will argue that for a certain notion of `badness', $(1)$ There are very few bad matrices (Section~\ref{sec:bad-matrices}) and $(2)$ matrices that are not bad are indeed good block sources (Section~\ref{sec:good=>blk-src}).

\subsubsection{Bad matrices} \label{sec:bad-matrices}
Before we proceed, we will need a couple of other definitions.
\begin{defn}
For every ${A}\in\supp{{\vA}}$ and $i\in [m]$, define
\[p_i({A})\eqdef \Pr\left[\vA_i={A}_i \big| \vA_{[i-1]}={A}_{[i-1]}\right], \qquad \text{and} \qquad q_i({A})\eqdef -\log_2(p_i({A})).\]
\end{defn}
\begin{defn}
For any $\tau>0$, we refer to ${A}\in\supp{{\vA}}$ as $\tau$-rare if there exists an $i\in [m]$ such that
\[p_i({A}) < 2^{-n(1+\tau)} \big(\text{or equivalently}, q_i({A}) > n(1+\tau)\big).\]
\end{defn}
The following lemma justifies the naming above:
\begin{lemma} \label{lem:rare}
$\displaystyle \prob{}{\vA\text{ is }\tau\mhyphen\text{rare}} < m\cdot 2^{-\tau\cdot n}.$
\end{lemma}
\begin{proof}
Call a matrix ${A}\in\supp{{\vA}}$ to be $\tau$-{\em rare at } $i\in [m]$ if $p_i({A})< 2^{-n(1+\tau)}$-- denote this event by $\cE(i,{A})$. We next show that
\[\prob{}{\vA\text{ is }\tau\mhyphen\text{rare at }i} < 2^{-\tau\cdot n},\]
which would complete the proof via the union bound. In the rest of the proof we prove the above bound.

Indeed, consider the following sequence of relations:
\begin{align}
\prob{}{\vA\text{ is }\tau\mhyphen\text{rare at }i} &= \sum_{{A}\in\bits^{m\times n}} \prob{}{\vA={A}}\cdot \ind{\cE(i,{A})}\notag\\
\label{eq:s1}
 &= \sum_{{A}_{[i-1]}\in\bits^{(i-1)\times n}} \prob{}{\vA_{[i-1]}={A}_{[i-1]}}\cdot \sum_{{A}_i\in \bits^n}\sum_{{A}_{[i+1:m]}\in\bits^{(m-i)\times n}} \prob{}{\vA={A}|\vA_{[i-1]}={A}_{[i-1]}}\cdot \ind{\cE(i,{A})}\\
\label{eq:s2}
 &= \sum_{{A}_{[i-1]}\in\bits^{(i-1)\times n}} \prob{}{\vA_{[i-1]}={A}_{[i-1]}}\cdot \sum_{{A}_i\in \bits^n}\prob{}{\vA_i={A}_i|\vA_{[i-1]}={A}_{[i-1]}}\cdot \ind{\cE(i,{A})}\\
\label{eq:s3}
&< \sum_{{A}_{[i-1]}\in\bits^{(i-1)\times n}} \prob{}{\vA_{[i-1]}={A}_{[i-1]}}\cdot \sum_{{A}_i\in \bits^n} 2^{-n(1+\tau)}\\
\label{eq:s4}
&\le 2^{-\tau\cdot n},
\end{align}
as desired. In the above,~\eqref{eq:s1} follows from definition of conditional probability.~\eqref{eq:s2} follows from the fact that $\ind{\cE(i,{A})}$ is the same for all matrices that agree in ${A}_{[i]}$ while~\eqref{eq:s3} follows from the definition of the event $\cE(i,{A})$. Finally,~\eqref{eq:s4} follows from the fact that there are $2^n$ possibilities for ${A}_i$.
\end{proof}

Next, we argue that every matrix ${A}$ that is not $\tau$-rare has few rows $i\in [m]$ for which $q_i({A})$ is small. It is crucial to note the difference between the situation here and what we need from block sources (in Definition~\ref{def:blk-src}). In Definition~\ref{def:blk-src}, the set of `bad' rows is the same for all matrices in the support of the distribution. On the other hand, in the lemma below, we show that for every matrix ${A}$, there exists a set of `bad rows'. (So ultimately, we would like to flip the order of quantifiers.)
\begin{lemma} \label{lem:markov}
Let ${A}\in\supp{{\vA}}$ be such that it is not $\tau$-rare. Then there exists a subset $S\subseteq [m]$ of size
\[|S|\le \sqrt{\tau+\gamma}\cdot m\]
such that for every $i\not\in S$, we have
\[q_i({A})\ge \left(1-\sqrt{\tau+\gamma}\right)\cdot n.\]
\end{lemma}
\begin{proof}
Define for every $i\in [m]$,
\begin{equation} \label{eq:q-n-q}
q'_i({A}) = n(1+\tau)-q_i({A}).
\end{equation}
Note that since $\minent{\vA}\ge (1-\gamma)mn$, we have
\begin{equation} \label{eq:q-n-q-2}
\sum_{i=1}^m q_i({A}) \ge (1-\gamma)mn.
\end{equation}
Since ${A}$ is not $\tau$-rare, we have (from Lemma~\ref{lem:rare}):
\begin{align*}
q_{i}(A) \le n \cdot (1 + \tau) \quad \forall i \in [m],
\end{align*}
which in turn implies
\begin{align*}
q'_i({A})  = n (1 + \tau) - q_i(A) \geq 0 \quad \forall i \in [m].
\end{align*}
Further, we have from~\eqref{eq:q-n-q}:
\begin{align*}
\sum_{i = 1}^{m} q'_i({A})  &= \sum_{i = 1}^{m} n(1+\tau) - \sum_{i = 1}^{m} q_i({A}) \\
& \le (1 + \tau) mn - (\gamma - 1)mn \qquad \text{ from~\eqref{eq:q-n-q-2}}\\
& = (\tau + \gamma) m n.
\end{align*}
Thus, by a Markov/counting argument, we have that the fraction of rows $i\in [m]$ for which we have $q'_i({A})\ge \sqrt{\tau+\gamma}\cdot n$ is at most $\sqrt{\tau+\gamma}$. Let this set of rows be $S$. Then note that for every $i\not\in S$, we have
\begin{align*}
q_i({A}) & =(1+\tau)n-q'_i({A}) \\
& \ge n(1+\tau-\sqrt{\tau+\gamma}) \\
& \ge n(1-\sqrt{\tau+\gamma}),
\end{align*}
as desired.
\end{proof}
Before we proceed for notational convenience, define
\[\eta\eqdef \sqrt{\tau+\gamma}.\]
We are now ready for our final set of definitions:
\begin{defn} \label{def:bad-set}
For every ${A}\in\supp{{\vA}}$ that is not $\tau$-rare, define $\bad{{A}}$ to be a subset $S\subseteq [m]$ such that
\begin{itemize}
\item[(1)] $|S| \leq \eta\cdot m$; and
\item[(2)] For every $i\not\in S$, $q_i({A})\ge (1-\eta)n$.
\end{itemize}
If $A$ is $\tau$-rare, then we set $B(A)=\bot$.
\end{defn}
Note that Lemma~\ref{lem:markov} shows that the function $\bad{A}$ is well defined for $A \in \ssupp(\vA)$ which is not $\tau$-rare. For other $A$'s, $\bad{A}$ is defined to be the `exception' symbol $\bot$.
\begin{defn} \label{def:bad-matrix}
Let $\tau,\delta>0$. We call ${A}\in\supp{{\vA}}$ to be $(\tau,\delta)$-{\em bad} if
\begin{itemize}
\item[(1)] ${A}$ is $\tau$-rare; or
\item[(2)] There exists an $i\in [m]$ such that
\begin{equation} \label{eq:bad-at-i}
\prob{}{\bad{\vA}=\bad{{A}}| \vA_{[i-1]}={A}_{[i-1]}} <\delta.
\end{equation}
\end{itemize}
\end{defn}

We will argue in Section~\ref{sec:good=>blk-src} that $\vA$ conditioned on $B(\vA)$ leads to a block source. But first we argue (using arguments similar to those used in the proof of Lemma~\ref{lem:rare}) that the total probability mass on bad matrices is small.
\begin{lemma} \label{lem:bad-matrix-prob}
For every $\tau,\delta>0$,
\[\eps_{\textrm{bad}}\eqdef \prob{}{\vA\text{ is }(\tau,\delta)\mhyphen\text{bad}} \le m\cdot 2^{-\tau n}+m\cdot \delta\cdot 2^{h(\eta)m}.\]
\end{lemma}
\begin{proof}
Call  a matrix ${A}\in\supp{{\vA}}$ to be bad at $i\in [m]$ if~\eqref{eq:bad-at-i} holds (and ${A}$ is not $\tau$-rare). Further, denote this event by $\cE'(i,{A})$. Then consider the following sequence of relations:
\begin{align}
\prob{}{\vA\text{ is bad at }i} &=\sum_{{A}\in\bits^{n\times n}} \prob{}{\vA={A}}\cdot\ind{\cE'(i,{A})}\notag\\
\label{eq:step1}
 &=\sum_{{A}_{[i-1]}\in\bits^{(i-1)\times n}} \prob{}{\vA_{[i-1]}={A}_{[i-1]}} \cdot \sum_{{A}_{[i:m]}\in\bits^{(m-i+1)\times n}} \prob{}{\vA={A}|\vA_{[i-1]}={A}_{[i-1]}}\cdot\ind{\cE'(i,{A})}\\
\label{eq:step2}
 &=\sum_{{A}_{[i-1]}\in\bits^{(i-1)\times n}} \prob{}{\vA_{[i-1]}={A}_{[i-1]}} \cdot \sum_{\substack{S\subseteq [m],\\ |S|\leq \eta m}}  \sum_{\substack{{A}_{[i:m]}\in\bits^{(m-i+1)\times n},\\ \bad{{A}}=S}} \prob{}{\vA={A}|\vA_{[i-1]}={A}_{[i-1]}}\cdot\ind{\cE'(i,{A})}\\
\label{eq:step3}
 &=\sum_{{A}_{[i-1]}\in\bits^{(i-1)\times n}} \prob{}{\vA_{[i-1]}={A}_{[i-1]}} \cdot \sum_{\substack{S\subseteq [m],\\ |S|\leq \eta m}}  \prob{}{\bad{\vA}=S|\vA_{[i-1]}={A}_{[i-1]}}\cdot\ind{\cE'(i,{A})}\\
\label{eq:step4}
 &<\sum_{{A}_{[i-1]}\in\bits^{(i-1)\times n}} \prob{}{\vA_{[i-1]}={A}_{[i-1]}} \cdot \sum_{\substack{S\subseteq [m],\\ |S|\leq \eta m}} \delta\\
\label{eq:step5}
&\le 2^{h(\eta)m}\cdot \delta.
\end{align}
In the above~\eqref{eq:step1} follows from definition of conditional probability,~\eqref{eq:step2} follows by re-arranging terms,~\eqref{eq:step3} follows by noting that $\cE'(i,{A})$ is the same for all matrices that agree on ${A}_{[i-1]}$ and have the same $\bad{{A}}$,~\eqref{eq:step4} follows from definition of $\cE'$ and~\eqref{eq:step5} follows from the fact that the number of subsets of $[m]$ of size at most $\eta m$  (for $\eta<1/2$) is upper bounded by $2^{h(\eta)m}$.

Taking union bound over all $m$ values of $i$ over the bound in~\eqref{eq:step5} along with Lemma~\ref{lem:rare} completes the proof.
\end{proof}

\subsubsection{A good matrix is a block source}
\label{sec:good=>blk-src}

We now argue that $\vA$ conditioned on $B(\vA)$ leads to a block source (that will suffice for our purposes):
\begin{lemma}
\label{lem:good=>blk-src}
Let $S \in \ssupp(\bad{\vA})$ be a subset of $[m]$ (thus $|S|\leq \eta m$) or $S=\bot$ and $\eps_{\mathrm{bad}}(S) \eqdef \Pr[\vA\text{ is } (\tau,\delta)\mhyphen\text{bad}\ |\bad{\vA}=S]$.
Then $\vA|\bad{\vA} = S$ is $\eps_{\mathrm{bad}}(S)$-close to a $(\eta,n(1-\eta)-\log(1/\delta))$-block source.\end{lemma}

The proof of the above lemma follows from a similar argument in Claim 9 from~\cite{amnon} though there is a bug in the published proof~\cite{erratum} (which we correct below).

\begin{proof}[Proof of Lemma~\ref{lem:good=>blk-src}]
Let ${A}$ be a matrix with $\bad{A} = S$ and assume ${A}$ is not $(\tau,\delta)$-bad (note that this implies $S\neq \bot$). By definition, for every $i\not\in S$, we have
\begin{equation}
\label{eq:qi-lb}
q_i({A})\ge (1-\eta)n.
\end{equation}
Now fix any $i\not\in S$. Then note that
\begin{align}
\label{eq:step6}
\prob{}{\vA_i={A}_i|\vA_{[i-1]}={A}_{[i-1]},\bad{\vA}=S}
& \le \frac{\prob{}{\vA_i={A}_i|\vA_{[i-1]}={A}_{[i-1]}}}{\prob{}{\bad{\vA}=S|\vA_{[i-1]}={A}_{[i-1]}}}\\
\label{eq:step8}
&\le \frac{p_i({A})}{\delta}\\
\label{eq:step9}
&\le 2^{-n(1-\eta)+\log(1/\delta)}.
\end{align}
In the above,~\eqref{eq:step6} follows from the fact that for any three events $\cE_1,\cE_2,\cE_3$, we have $\prob{}{\cE_1|\cE_2,\cE_3}\le \frac{\prob{}{\cE_1|\cE_3}}{\prob{}{\cE_2|\cE_3}}$, \footnote{$\prob{}{\cE_1|\cE_2,\cE_3} = \frac{\prob{}{\cE_1, \cE_2,\cE_3}}{\prob{}{\cE_2,\cE_3}} \le \frac{\prob{}{\cE_1, \cE_3}}{\prob{}{\cE_2,\cE_3}} = \frac{\prob{}{\cE_1|\cE_3}}{\prob{}{\cE_2|\cE_3}}$. Here, the inequality follows from the fact that any marginal probability values are always at most their original joint probability values.}~\eqref{eq:step8} follows from definition of $p_i(\cdot)$ and ${A}$ not being $(\tau,\delta)$-bad and~\eqref{eq:step9} follows from~\eqref{eq:qi-lb}.

Taking into account that ${A}$ can be $(\tau,\delta)$-bad, we have that $\vA|\bad{\vA} = S$ is $\eps_{\mathrm{bad}}(S)$-close to an $(\eta,n(1-\eta)-\log(1/\delta))$-block source.
\end{proof}

\subsection{A good block source $\vA$ leads to $\vA\vx$ with high min-entropy} \label{sec:CG-sse}
In this section, we prove that a good block source is enough for our purposes.
\begin{lemma}
\label{lem:blk-src-enuf}
Let $\vA'$ be an $(\eta,n(1-\zeta))$-block source and $\vx \in \F_2^n$ with $\minent{\vx}=\alpha n$ such that
\[\alpha\ge 2\Delta+\zeta.\]
Further, the distributions $\vA'$ and $\vx$ are independent.
Then, there exists a subset $T\subseteq [m]$ with $|T|\geq(1-\eta)m$ such that $(\vA'\vx)_T$ is $\eps$-close to $\cU_{|T|}$ for
\[\eps\le |T|\cdot 2^{-\Delta n}.\]
\end{lemma}
We now state a result that we will be using in the proof of the above lemma.
\begin{thm}[\cite{sse}] \label{thm:ip-sse}
Let $\vy$ and $\vz$ be independent random variables on $\F_2^n$ such that
\[\minent{\vy}+\minent{\vz}\ge (1+\Delta)n.\]
Then, $\left(\vy,\ip{\vy}{\vz}\right)$ is $\eps_{\textrm{IP}}$-close to $\cD_{\vy}\times\cU_1$, where $\cD_{\vy}$ is the distribution for $\vy$ and
\[\eps_{\textrm{IP}}\le 2^{-\frac{\Delta n}{2} -1}.\]
\end{thm}
We are now ready to prove Lemma~\ref{lem:blk-src-enuf} (via a simple modification of Lemma $6$ in~\cite{david}).
\begin{proof}[Proof of Lemma~\ref{lem:blk-src-enuf}]
Since $\vA'$ is an $(\eta,n(1-\zeta))$-block source, there exists a subset $T\subseteq [m]$ of size at least $(1-\eta)m$ such that for all $i\in T$ and for every ${A}_{[i-1]}$:
\begin{equation} \label{eq:cond-minent-lb}
\minent{\vA'_i|\vA'_{[i-1]}={A}_{[i-1]}}\ge n(1-\zeta).
\end{equation}

For notational simplicity, we assume $T=[n']$, where $n'=(1-\eta)m$. 
To prove the lemma, we will prove by induction on $i$ from $n'$ to $0$ that for every ${A}_{[i]}$, the distribution of
\begin{align*}
\left(\vx,\vA'_{[i+1, n']}\vx\right) \text{ conditioned on } \vA'_{[i]}={A}_{[i]} \text{ is } (n'-i)\cdot \eps_{\mathrm{IP}}\text{-close to } \cD_{\vx|A_{[i]}}\times \cU_{n'-i}, 
\end{align*}
where $\cD_{\vx|A_{[i]}}$ is the distribution for $\vx|\vA'_{[i]}={A}_{[i]}$. Note that this claim for $i=0$ and Proposition~\ref{prop:vard} (where the deterministic function just drops the $\vx$ `part') is sufficient to prove the lemma.

The base case of $i=n'$ is trivial. Let us assume that the induction hypothesis is true for $i+1$. That is for every ${A}_{[i+1]}$, we have that the distribution of 
\begin{align*}
\left(\vx,\vA'_{[i+2, n']}\vx\right) \big | \vA'_{[i+1]}={A}_{[i+1]} \text{ is } (n'-i-1)\cdot \eps_{\mathrm{IP}} \text{-close to } \cD_{\vx|A_{[i+1]}} \times \cU_{n'-i-1}.
\end{align*}

We will now argue the claim for $i$. Towards that end fix an arbitrary ${A}_{[i]}$ and let $\hat\cD$ be the distribution for $\vA'_{i+1}|\vA'_{[i]}={A}_{[i]}$. Since the claim on the distribution in the above paragraph holds for every ${A}_{i+1}$, we have that 
\begin{align*}
\left(\vA'_{i+1},\vx, \vA'_{[i+2, n']}\vx\right)\big|\vA'_{[i]} = {A}_{[i]} \text{ is } (n'-i-1)\cdot \eps_{\mathrm{IP}} \text{-close to } \cD_2\eqdef \hat\cD\times \cD_{\vx|A_{[i]}}\times \cU_{n'-i-1},
\end{align*}
where we denote the correspondent distribution of $\left(\vA'_{i+1},\vx, \vA'_{[i+2, n']}\vx\right)\big|\vA'_{[i]} = {A}_{[i]}$ by $\cD_1$. 

Next, we apply Proposition~\ref{prop:vard} on $\cD_1$ and $\cD_2$ (where the deterministic function puts the second component as the new first component and the new second component is the inner product of the earlier first two components), to get that 
\begin{align*}
\left(\vx, \vA'_{[i+1, n']} \vx \right) \big| \vA'_{[i]} = A_{[i]} \qquad \sim_{(n'-i-1)\cdot \eps_{\mathrm{IP}}} \qquad \left(\vx,\left(\ip{\vA'_{i+1}}{\vx},u_{i+2},\dots,u_{n'}\right)^{\mathrm{T}}\right) \big | \vA'_{[i]} = A_{[i]},
\end{align*}
where the bits $(u_{i + 2}, \dots, u_{n'})$ are independent and uniformly random. Since these bits are independent of $(\vx,\ip{\vA'_{i+1}}{\vx})$, we can invoke Theorem~\ref{thm:ip-sse} with $\mathbf{y} = \vx$ and $\mathbf{z} = \vA'_{i+1}$ (recall that $\vx$ and $\vA'_{i+1}$ are independent) to obtain the following:
\begin{align*}
\left(\vx,\left(\ip{\vA'_{i+1}}{\vx},u_{i+2},\dots,u_{n'}\right)^{\mathrm{T}}\right) \big| \vA'_{[i]} = A_{[i]} \text{ is } \eps_{\mathrm{IP}} \text{-close to } \left(\vx, (u_{i+1}, u_{i+2}, \cdots, u_{n'})^{\mathrm{T}}\right)\big | \vA'_{[i]} = A_{[i]},
\end{align*} 
which is exactly $ \cD_{\vx|A_{[i]}} \times \cU_{n'-i}$. Then, by the triangle inequality, we have that the distribution of
\begin{align*}
\left(\vx,\vA'_{[i+1,n']}\vx\right)\big|\vA'_{[i]}={A}_{[i]} \text{ is } (n'-i)\cdot \eps_{\mathrm{IP}} \text{-close to } \cD_{\vx|A_{[i]}}\times \cU_{n'-i},
\end{align*}
as desired.

We finally note that the assumption of $T=[n']$ is almost WLOG since in the more general case we do the above argument for all $i\in [m]$ but make the above argument only for indices in $T$ (while the conditioning also happens for rows not in $T$).
\end{proof}

\subsection{Putting everything together} \label{app:together}
We now have all the pieces at our disposal and are finally ready to prove Theorem~\ref{thm:CG-ext-app}. We first note that Lemmas~\ref{lem:good=>blk-src} and~\ref{lem:blk-src-enuf} imply the following result (if $\eps_1=\eps_2=0$).
\begin{lemma}  \label{lem:main}
Let $\vx$ and $\vA$ be independent variables such that
\[\hinf{\vx}\ge (2\Delta+\sqrt{\tau+\gamma})\cdot n +\log(1/\delta),\]
and 
\[\hinf{\vA}\ge (1-\gamma)mn.\]
Then the distribution on $\vA\vx$ is $\eps_{\mathrm{bad}}+m\cdot 2^{-\Delta n}$-close to a distribution with min entropy at least $(1-\sqrt{\tau+\gamma})n$.
\end{lemma}
\begin{proof}
For each fixing of $\bad{\vA}=S$, using Lemma~\ref{lem:good=>blk-src} and Lemma~\ref{lem:blk-src-enuf} (since $\vA$ and $\vx$ are independent), we have that $\vA \vx$ conditioned on $B(\vA)=S$ is $\eps_{\mathrm{bad}}(S)+m\cdot 2^{-\Delta n}$-close to a distribution
with min entropy at least $(1-\sqrt{\tau+\gamma})n$, where recall that $\eps_{\mathrm{bad}}(S)=\Pr_{}[\vA\text{ is } (\tau,\delta)\mhyphen\text{bad}|B(\vA) = S]$. 

Then, taking into account all possibilities of $\bad{\vA}$ (i.e., the distribution on $\vA\vx$), which in turn is a convex combination of distributions, is $\eps'$-close to a distribution with min entropy at least $(1-\sqrt{\tau+\gamma})n$, where
\[\eps' \leq \sum_{S\subseteq [m], |S|\leq \eta m \text{ or } S=\bot} \left(\eps_{\mathrm{bad}}(S)+m\cdot 2^{-\Delta n}\right)\cdot \Pr[\bad{\vA}=S] = \eps_{\mathrm{bad}}+m\cdot 2^{-\Delta n},\]
as desired.
\end{proof}
Finally, to prove Theorem~\ref{thm:CG-ext-app}, we will extend Lemma~\ref{lem:main}. Before doing so, we statethe parameters that we instantiate Lemma~\ref{lem:main} with below.
\[\Delta=\tau=\gamma,\]
\[\delta = 2^{-\gamma m -h(\sqrt{2\gamma})m},\]
Note that these imply the claimed parameters in Theorem~\ref{thm:CG-ext-app}.

Now assume we are given $H_\infty^{\eps_1}(\vA'|\vY) \geq (1-\gamma)nm, H_\infty^{\eps_2}(\vx'|\vY) \geq \alpha n$. Moreover, for every $y \in \ssupp(\vY)$, we have that conditioned on $\vY = y$, $\vA'$ and $\vx'$ are independent. Our goal is to prove $H_\infty^{\eps_1+\eps_2+2^{-\Omega(\gamma m)}}(\vA'\vx'|\vY) \geq (1-\sqrt{2\gamma})m$.

We can assume that there are two independent events $\calE_1, \calE_2$ with $\Pr[\calE_1] \geq 1-\eps_1$ and $\Pr[\calE_2] \geq 1-\eps_2$ and for every $y \in \ssupp(\vY)$, we have
\begin{enumerate}[label=(D\arabic*)]
\item $H_\infty(\calE_1\vA'|\vY = y) \geq (1-\gamma)nm$,
\item $H_\infty(\calE_2\vx'|\vY = y) \geq \alpha n$, 
\item Conditioned on $\vY = y$, $(\vA', \bone_{\calE_1})$ and $(\vx', \bone_{\calE_2})$ are independent.
\end{enumerate} 
(D1) and (D2) are satisfied by the definition of conditional smooth min-entropy. (D3) can be assumed due to the facts that $\vA'$ and $\vx'$ are independent conditioned on $\vY= y$, (D1) only involves $\calE_1$ and $\vA'$, and (D2) only involves $\calE_2$ and $\vx'$.

Then we can construct a distribution $(\vA, \vx)$ joint with $\vY$ such that for every $y \in \ssupp(\vY)$, we have
\begin{enumerate}[label=(E\arabic*)]
\item $\Pr[\vA = A|\vY = y] \geq \Pr[\calE_1, \vA' = A|\vY=y]$ for every $A \in \F^{m \times n}$, and moreover $H_\infty(\vA|\vY = y) \geq (1-\gamma)nm$,
\item $\Pr[\vx = x|\vY = y] \geq \Pr[\calE_2, \vx' = x|\vY=y]$ for every $x \in \F^{n}$, and moreover $H_\infty(\vx|\vY = y) \geq \alpha n$,
\item Conditioned on $\vY = y$, $\vA$ and $\vx$ are independent.
\end{enumerate}
To see how to guarantee (E1), we focus on a $y \in \ssupp(\vY)$. Start with the vector $\theta$ such that $\theta_A = \Pr[\calE_1, \vA' = A|\vY=y]$ for every $A \in \F^{m \times n}$.
Notice that $\|\theta\|_1 = \Pr[\calE_1|\vY = y] \leq 1$ and by (D1), we have $\|\theta\|_\infty \leq 2^{-(1-\gamma)nm}$.  We then increase coordinates of $\theta$ so that $\|\theta\|_1 = 1$ and $\|\theta\|_\infty \leq 2^{-(1-\gamma)nm}$ is maintained. This is doable since $2^{-({1-\gamma})nm}\cdot 2^{nm} \geq 1$. Then, we can guarantee (E1) by defining $\Pr[\vA = A|\vY = y] = \theta_A$  for the new vector $\theta$.  Similarly we can guarantee (E2). (E3) can be guaranteed since the conditioning on $\vY= y$ changes only the distribution of $\calE_1$ and $\calE_2$ but not $\vA$ and $\vx$. Thus, for every $y \in \ssupp(\vY)$ and $z \in \F^m$, we have
\begin{align}
\Pr[\calE_1, \calE_2, \vA'\vx' = z|\vY = y] &= \sum_{A,x:Ax = z} \Pr[\calE_1, \calE_2, \vA' = A, \vx' = x | \vY = y] \nonumber \\
&= \sum_{A,x:Ax = z} \Pr[\calE_1, \vA' = A | \vY = y] \Pr[\calE_2, \vx' = x | \vY = y] \nonumber\\
&\leq \sum_{A,x:Ax = z} \Pr[\vA = A | \vY = y] \Pr[\vx = x | \vY = y]\nonumber \\
&= \sum_{A,x:Ax = z} \Pr[\vA = A, \vx = x | \vY = y] \nonumber\\
&= \Pr[\vA\vx = z|\vY = y]. \label{ineq:A'x'-to-Ax}
\end{align}
The second and third equalities are by (D3) and (E3) respectively and the inequality is by (E1) and (E2).

The proof of Lemma~\ref{lem:main} already shows that $H^{\eps_*}_\infty(\vA\vx) \geq (1-\sqrt{2\gamma})m$ for some $\eps_* = 2^{-\Omega(\gamma m)}$ (since the distribution on $\vA\vx$ is $2^{-\Omega(\gamma m)}$ close to a distribution with min entropy at least $(1-\sqrt{2\gamma})m$), as we have (E1), (E2) and (E3). By the definition of $H^{\eps_*}_\infty$, there exists an event $\calE_*$ such that $\Pr[\calE_*] \geq 1- \eps_*$ and $\Pr[\calE_*, \vA\vx = z|\vY = y] \leq 2^{-(1-\sqrt{2\gamma})m}$ for every $y \in \ssupp(\vY)$ and $z \in \F^m$. Thus, by \eqref{ineq:A'x'-to-Ax}, there exists an event $\calE_*'$ with $\Pr[\calE_*'] = \Pr[\calE_*] \geq 1-\eps_*$ such that for every $y \in \ssupp(\vY)$ and $z \in \F^m$:
\begin{align*}
\Pr[\calE_*', \calE_1, \calE_2, \vA'\vx' = z|\vY = y] \leq \Pr[\calE_*, \vA\vx = z|\vY = y]\leq 2^{-(1-\sqrt{2\gamma})m}, %\text{ for every $y \in \ssupp(\vY)$ and $z \in \F^m$}.
\end{align*}
Notice that $\Pr[\calE_*', \calE_1, \calE_2] \geq 1-(\eps_* + \eps_1 + \eps_2)$ by union bound. By the definition of conditional smooth min-entropy, we have $H^{\eps_* + \eps_1 + \eps_2}(\vA'\vx'|\vY) \geq (1-\sqrt{2\gamma})m$, which completes the proof of Theorem~\ref{thm:CG-ext-app}.

\section{Missing details from Section~\ref{SEC:MATRIX-CHAIN}} \label{app:miss-matrix-chain}

\subsection{The case of $k\ge N$} \label{app:k>=N}
For simplicity, we assume $k$ is an integer power of $2$. In the first iteration, each $P_i$ with odd $i$ sends $A_i$ to $P_{i+1}$, who then computes $B^1_{i+1} := A_iA_{i+1}$; this iteration takes $N^2$ rounds. In the second iteration, each $P_i$ with $i \mod 4 = 2$ sends $B^1_i$ to $P_{i+2}$, who then computes $B^2_{i+2} = B^1_iB^1_{i+2}$; this iteration takes $N^2 + 1$ rounds. In general, in the $t$-th iteration for $t \in [\log k]$, each $P_i$ with $i \mod 2^t = 2^{t-1}$ sends $B^{t-1}_i$ to player $P_{i+2^{t-1}}$, who then computes $B^t_{i+2^{t-1}} = B^{t-1}_iB^{t-1}_{i+2^{t-1}}$; the iteration takes $N^2 + 2^{t-1}-1$ rounds. So in a total of $O(N^2\log k + k)$ rounds, player $P_k$ will know the product $A_1A_2\cdots A_k$. Using additional $O(k + N)$ rounds, $P_0$ can send $x$ to $P_k$. The whole protocol takes $O(N^2\log k + k)$ rounds. The bound can be slightly improved to $O(N^2\log (k/N) + k)$ by running the merging procedure for only $\log (k/N)$ iterations. 

\subsection{Proof of Lemma~\ref{lemma:correct-probability}} \label{app:matrix-chain-lb}
\begin{proof}
By the definition of $\hinfe{X|Y}{\epsilon}$, there exists an event $\mathcal{E}$ such that $\Pr[\mathcal{E}] \geq 1 - \epsilon$ and for every $x \in \ssupp(X), y \in \ssupp(Y)$, we have $\Pr[\mathcal{E}, X=x|Y = y] \leq 2^{-L}$. In particular, $\Pr[\mathcal{E}, X = f(y)|Y = y] \leq 2^{-L}$ for every $y \in \ssupp(Y)$. So, we have $\Pr[\mathcal{E}, f(Y) = X] \leq 2^{-L}$, which implies $\Pr[f(Y) = X] \leq 2^{-L} + \epsilon$ since $\Pr[\mathcal{E}] \geq 1-\epsilon$.
\end{proof}

\subsection{Why Shannon entropy does not work for our proof of Lemma~\ref{lem:min-ent-comm}} \label{app:shannon}
We will use $\hshan{\cD}$ to denote the Shannon entropy of $\cD$, which is defined as follows:
\[\hshan{\cD} =-\sum_{x\in \supp{\cD}} \prob{X\sim \cD}{X=x}\log_2\left(\prob{X\sim \cD}{X=x}\right).\]
If we had used Shannon entropy instead of (smooth) min-entropy, we would have to prove a bound of the following form. Let $f:\F_2^{N\times N}\to\F_2^m$ be an arbitrary function. Then, as long as $\hshan{\vx}$ and $\hshan{\vA}$ are big enough, $\vA\vx|f(\vA)$ should have entropy strictly bigger than $\hshan{\vx}$ (or be close to a distribution that has high enough entropy) as long as $m$ is a small fraction of $N^2$. We now give an example to show this is not possible.

Fix arbitrary linearly independent vectors $\vx^*_1,\dots,\vx^*_t$ for $t=\alpha N$. Then, define the following distribution on $\vx$: probability mass of $1-\alpha$ is distributed uniformly over the span of $\vx^*_1,\dots,\vx^*_t$ (call this span $S$) and the remaining mass of $\alpha$ is distributed uniformly over a null space of $S$. Then note that 
\[\hshan{\vx}= (1-\alpha)\cdot t +\alpha\cdot(N-t) =2\alpha(1-\alpha)\cdot N.\]

Now consider the case when $\vA$ is uniformly distributed (i.e. $\hshan{\vA}=N^2$) and $f(\vA)=\left(\vA\vx^*_1,\dots,\vA\vx^*_t\right)$. Then note that if $\vx\in S$, then $\hshan{\vA\vx|f(\vA)}=0$. This implies that
\[\hshan{\vA\vx|f(\vA)} \le (1-\alpha)\cdot 0 +\alpha\cdot N,\]
which is about a factor two smaller than $\hshan{\vx}$ (for small enough $\alpha>0$).

\subsection{Existing lower bound techniques for the Matrix-Chain Problem} \label{app:prev-work-mcm}
We remark that the existing technique of \cite{topology-3} that `stitches' the lower bounds induced by cuts can only give a lower bound of $\Omega(N)$: for any edge $(P_i, P_{i+1})$ on the path, we can only prove a lower bound of $\Omega(N)$ on the number of bits that need to be exchanged between $P_i$ and $P_{i+1}$, since if suffices for $P_i$ to send the product $\vA_k\vA_{k-1}\cdots\vA_1\vx$ to $P_{i+1}$. The lower bound given by \cite{topology-3} is then the minimum number of rounds needed to make sure that $\Omega(N)$ bits are exchanged between $\{P_0, P_1, \cdots, P_i\}$ and $\{P_{i+1}, P_{i+2}, \cdots, P_{k+1}\}$ for every $i$, which can only be $\Omega(N)$. However, this analysis does not capture a very simple fact: $P_i$ needs to know $\vA_k\vA_{k-1}\cdots\vA_1\vx$ before it can be sent to $P_{i+1}$.

\end{document}